%% file: arxiv.tex
\pdfoutput=1

\pdfoutput=1

\documentclass[11pt]{scrartcl}
\usepackage[a4paper, total={16cm, 24cm}]{geometry}
\usepackage{microtype} 
\usepackage{authblk}
\title{(Hopefully not too Long) Title}
\date{\vspace{-1.5cm}}

\author{Chris Dong}

\author[2]{Sonja Kraiczy}

\author[2]{Rohit Vasishta}

\author[2]{Markus Brill}

\author[3]{Wesley~H.~Holliday}

\author{Niclas Boehmer}
\affil[1]{Hasso Plattner Institute, University of Potsdam, Germany} 
\affil[2]{University of Oxford, UK}
\affil[3]{University of California, Berkeley, USA}

\usepackage{xcolor}
\usepackage{graphicx}
\usepackage{pgfplots}
\pgfplotsset{compat=1.17}

\usepackage{amsmath, amssymb, amsthm, mathtools, nicefrac, siunitx}

\usepackage{thmtools}
\usepackage{thm-restate}

\usepackage[pagebackref]{hyperref}
\hypersetup{
	pdfencoding=auto,
	psdextra,
	colorlinks=true,
	citecolor=green!50!black,
	linkcolor=red!60!black,
}

\usepackage[nameinlink]{cleveref}

\usepackage{natbib}
\usepackage[inline]{enumitem}
\usepackage{booktabs}
\usepackage{subcaption}

\usepackage[textsize=tiny]{todonotes}

\newcommand{\restatehere}[1]{%
	\marginline{\vspace{0.6cm}\footnotesize \hyperlink{original#1}{\hypertarget{restated#1}{[Main]}}}%
	\csname #1\endcsname*%
}

\usepackage[utf8]{inputenc} 
\usepackage[T1]{fontenc}    
\usepackage{hyperref}       
\usepackage{url}            
\usepackage{booktabs}       
\usepackage{amsfonts}       
\usepackage{nicefrac}       
\usepackage{microtype}      
\usepackage[dvipsnames]{xcolor}

\usepackage{graphicx} 

\usepackage{amsthm}
\usepackage{amsmath}
\usepackage{pgfplots}
\usepackage{mdframed}
\usepackage{cleveref}
\usepackage[normalem]{ulem}
\usepackage{soul}
\usepackage{amsfonts}
\usepackage{amssymb}
\usepackage{multirow}
\usepackage{bbm}
\usepackage{mathtools}
\usepackage{pifont}
\usepackage{caption}
\usepackage{subcaption}
\usepackage{nicefrac}
\usepackage{tikz}
\usetikzlibrary{positioning,arrows,shapes,calc}
\usepackage{colortbl}
\usepackage{booktabs}
\usepackage{array}
\usepackage{enumitem}
\usepackage{tabularx}
\usepackage{verbatim}
\usepackage{makecell}
\usepackage[most]{tcolorbox}
\usepackage{todonotes}
\usepackage{comment}
\usepackage{algorithm}
\usepackage{algpseudocode}
\usepackage{thmtools}
\usepackage{thm-restate}
\usepackage[most]{tcolorbox}

\theoremstyle{definition}
\newtheorem{definition}{Definition}

\newtheorem{example}{Example}
\theoremstyle{plain}
\newtheorem{theorem}{Theorem}
\newtheorem{lemma}{Lemma}
\newtheorem{proposition}{Proposition}
\newtheorem{corollary}{Corollary}

\newtheorem{remark}{Remark}
\usetikzlibrary{positioning,calc,fit,backgrounds}
\usepackage{textcomp}

\usepackage{tikz}
\usetikzlibrary{arrows.meta, shapes.geometric, calc,positioning,fit}
\usepackage[most]{tcolorbox}

\definecolor{lightpurple}{HTML}{F4EEFF}
\definecolor{darkpurple}{HTML}{5B2A86}
\definecolor{softpurple}{HTML}{FAF7FF}
\definecolor{midpurple}{HTML}{8A5FBF}

\newtcolorbox{interpretationbox}[1]{
    enhanced,
    colback=softpurple,
    colframe=darkpurple,
    boxrule=0.6pt,
    arc=2mm,
    left=7pt,
    right=7pt,
    top=7pt,
    bottom=6pt,
    width=0.92\linewidth,
    center,
    before skip=6pt,
    after skip=6pt,
    attach boxed title to top left={xshift=8pt,yshift=-2.5mm},
    boxed title style={
        colback=darkpurple,
        colframe=darkpurple,
        arc=1.5mm,
        boxrule=0pt,
        left=4pt,
        right=4pt,
        top=1pt,
        bottom=1pt
    },
    title={#1},
    fonttitle=\small\bfseries,
    coltitle=white
}
\definecolor{candA}{HTML}{C9A66B}   
\definecolor{candC}{HTML}{7FA6C9}   
\definecolor{candX}{HTML}{7FB0AA}   
\definecolor{candY}{HTML}{C98A9A}   
\definecolor{candZ}{HTML}{9CAF6C}   
\definecolor{approve}{HTML}{F3EFD9} 
\definecolor{winered}{rgb}{0.5,0.1,0.1}
\renewcommand\emph[1]{{\color{winered}{\textit{#1}}}}

\newcommand{\score}[2]{s_{#1}(#2)}

\newcommand{\spvc}{\textsc{SPVC}}
\newcommand{\pvc}{\textsc{PVC}}
\newcommand{\tc}{\textsc{DSS}}
\newcommand{\ru}{\textsc{ARU}}
\newcommand{\mmod}{\textsc{MMOD}}

\newcommand{\dvc}{\textsc{DVC}}
\newcommand{\cone}{\operatorname{cone}}
\newcommand{\po}{\textsc{DPS}}

\newmdtheoremenv[
  backgroundcolor=gray!10,
  linecolor=black,
  linewidth=0.8pt,
  roundcorner=3pt
]{solutionconcept}{Solution Concept}

\newtcolorbox{principle}{
  colback=gray!10,
  colframe=black,
  boxrule=0.5pt,
  arc=2pt,
  left=6pt,
  right=6pt,
  top=4pt,
  bottom=4pt
}

\newcommand{\sonja}[1]{{\color{red}[Sonja: #1]}}

\newcommand{\rohit}[1]{{\color{brown}[Rohit: #1]}}
\newcommand{\chris}[1]{{\color{blue}[Chris: #1]}}

\title{Where Should Society Draw the Line?\\ A Social Choice Approach to Collective Consent}

\begin{document}

\maketitle

\begin{abstract}

Society constantly has to determine the boundaries of what it deems acceptable, from legislative decisions to the guardrails governing autonomous systems. We initiate the axiomatic study of \emph{collective consent}: given individuals' attitudes toward options, which options should receive societal consent? We organize our analysis around three principles: sufficient support, minority protection, and dominance by decisively better options. Each captures a distinct reason for withholding societal consent from an option. For each principle, we develop a corresponding solution concept that transparently implements the principle and is canonical in a mathematically precise sense. For example, for minority protection, the resulting concept is a consent-adapted version of Moulin's Proportional Veto Core. 

Balancing multiple principles simultaneously is more challenging. To address this, we develop a game-theoretic characterization of our veto core that naturally gives rise to a family of related concepts. From this family, we identify the \emph{Approval-Weighted Veto Core} as particularly desirable. By making minorities' blocking power depend on the approval support of the options being challenged, it smoothly interpolates between proportional minority protection and majority support. 
Experiments on five datasets spanning \textit{high-stakes} decision-making (such as political elections, ethical AI evaluations, and moral decision-making)  show that solution concepts violating a principle in theory also violate it empirically.
\end{abstract}

\section{Introduction}\label{blocking}

Modern AI systems constantly make decisions on what outputs are acceptable. For example, large language models are fine-tuned to determine which responses to present or withhold \citep{ouyang2022training,christiano2017deep}, and autonomous vehicles must resolve which driving behaviors are acceptable  \citep{awad2018moral}. Likewise, society constantly has to draw lines on which options to tolerate. 
Legislators decide which behaviors are permitted; review boards decide 
which medical treatments may be administered; operators of 
autonomous energy grids decide which dispatch actions the system may 
take without human approval. These decisions share a common structure: presented with a set of options, society 
needs to determine to which of these it gives \textit{consent}, and which 
should be ruled out.

{Historically, giving consent is broadly understood as the action of an \emph{individual}. As such, the concept has been studied across disciplines, e.g., in political philosophy as a foundation of political legitimacy \citep{peter2023political} and in research ethics as a check on decisions affecting communities \citep{weijer2000protecting}.
Such works may focus on justifying the authority of collective decisions over individuals, or equate the consent of a group of people to the consent of a representative individual.
By contrast, in this paper, we focus on the formal analysis of a process of deriving \emph{societal} consent from individual consent.
} 
The underlying assumption of our work is that the legitimacy of any set of options with societal consent 
ultimately rests on how well it represents and respects the input of 
those it affects.
This raises the fundamental question we 
address in this paper:
\emph{Given individuals' attitudes toward
options, which 
options should society consent to?}

To address this question, we initiate the formal study of deriving
collective consent from a social choice perspective. Social 
choice theory studies how individual preferences over a set of 
candidates should be aggregated into collective decisions, with rich 
variation in how preferences are represented, how they are aggregated, 
and what efficiency, fairness, and algorithmic guarantees are sought \citep{HandbookofCSC, 10.5555/3180776}. Two prominent settings within this framework 
are single-winner voting, in which a single candidate must be selected, 
and multi-winner voting \citep{lackner2023multi,faliszewski2017multiwinner,faliszewski2020multiwinner}, in which a subset of candidates of a given size must be selected. In both settings, individuals (henceforth referred to as \textit{voters}) typically report either an approval 
set, a ranking over candidates, or a utility value for each candidate.

Our formal model of deriving collective consent departs from this tradition in two ways. First, on 
the input side, each voter reports \textit{both} a utility for each 
candidate and a set of candidates she finds acceptable, i.e., she consents to. Second, on 
the output side, we place no size constraint on the set of candidates 
receiving societal consent: depending on the instance, every candidate 
may pass, none may, or any subset in between (see 
Appendix \ref{app:RelWork} for a more detailed comparison to prior work).

Beyond these technical differences, the question of collective consent 
calls for a distinctive set of normative principles. 
Unlike standard social choice settings, where the produced output is the implemented outcome, 
{in collective consent society agrees on a pre-selection only excluding indefensible candidates. Which of the candidates from the set of consented candidates are eventually enacted is decided downstream.}\footnote{With this pre-selection being the key idea of the paper, the reader may also think of (societal and individual) ``approval'' instead of ``consent'' without misinterpreting this paper.
}
Each 
{candidate's}
consent decision must therefore be defensible on its own, 
rather than only the returned set as a whole.
Our work is 
organized around three different reasons why a candidate should be 
excluded from societal consent, which we refer to as \emph{blocking 
principles}. We seek to design solution concepts—functions mapping 
voters' attitudes to a subset of candidates—that satisfy them\footnote{We 
intentionally frame these principles in the negative to emphasize that 
individuals and minorities retain the power to withhold consent. These principles precisely form the \textit{boundary} of what should be blocked. 
}:

\begin{description}
    \item[\textbf{Sufficient Support} \normalfont{(Local Blocking)}] 
    A candidate should not receive societal consent if only a few 
    voters consent to it.
    \item[\textbf{Minority Protection} \normalfont{(Coalitional Blocking)}] 
    A candidate should not receive societal consent if this comes at 
    a ``disproportionate'' expense of some minority.
    \item[\textbf{Presence of Decisively Better Candidates} 
    \normalfont{(Relational Blocking)}] A candidate should not receive 
    societal consent when a clearly superior candidate exists.
\end{description}

In \Cref{sec:solution-concepts}, we formalize and justify the three principles: 
What counts as ``sufficient'' support is application-dependent, but the 
underlying requirement is fundamental: without it, there is no 
democratic backing as to why society should consent to the candidate. We formalize this principle by 
requiring that every candidate receiving societal consent must be approved 
by at least a fixed fraction of voters. Beyond broad popularity, we 
must also ensure that a consented candidate does not come at the 
disproportionate expense of a cohesive minority. 
We formalize this principle through the notion of \textit{veto power} 
\citep{Moulin1981,Moulin1982}: a coalition of voters should be able to 
block a fraction of the candidate space proportional to its size, 
regardless of how the rest of the electorate behaves. This guards against majoritarian exploitation, in which a (possibly slim) 
majority approves candidates that impose substantial harm on a cohesive 
minority. Instead, the principle forces the electorate toward consensus 
candidates tolerable to minorities as well.
Lastly, when a clearly superior candidate is 
available, granting societal consent to an inferior one is 
hard to justify: candidates disapproved by some voters 
should not be selected when others are unanimously approved.

We seek solution concepts that satisfy formalizations of all three 
principles. In doing so, we additionally require \emph{transparency 
and accountability}: whenever a candidate is excluded, we must be 
able to point to the voters or candidates responsible for the 
exclusion. In \Cref{sec:solution-concepts}, we present a solution concept 
tailored to each principle. For sufficient support, we characterize 
the $\alpha$-approval threshold solution concept that accepts all candidates that are approved by an $\alpha$ fraction of the voters as the unique concept satisfying the 
principle together with a small set of desirable properties, 
including independence of other candidates. For minority protection, 
we propose the \emph{Disapproval Veto Core (DVC)}, an adaptation of 
Moulin's proportional veto core \citep{Moulin1981,Moulin1982} that incorporates approval 
information.\footnote{Notably, applying the unmodified proportional 
veto core to the societal consent problem leads to undesirable 
outcomes: almost all candidates are approved by almost all 
voters, yet the only chosen candidate receives no approvals at all 
(see \Cref{fig:PVC-Bad} in Appendix \ref{app:RelWork} for details and 
for additional related work on the proportional veto core).} We 
further show that any solution concept satisfying basic 
transparency and well-behavedness axioms must veto at least the 
candidates blocked by DVC in case they want to provide minority protection, even only in pathological extreme cases. 
Finally, for the presence of decisively 
better candidates, we introduce the \emph{disapproval Smith set}, an 
adaptation of the classical Smith set sensitive to approval 
information.

In \Cref{sec:Endowments_Char_DVC}, we develop an alternative 
characterization of DVC as the equilibrium of a \emph{veto budgeting 
game}, in which each voter is endowed with a budget she may spend to 
block candidates. This game-theoretic perspective gives rise to a 
family of DVC variants obtained by adjusting voters' budgets and 
candidates' prices. Among these, we identify the \emph{Approval-Weighted 
Veto Core (AWVC)} as particularly attractive: it lies in a Goldilocks 
zone between proportional veto power and sufficient support, 
combining minority protection with the majority threshold while 
smoothly interpolating between the two. Composing AWVC with the 
disapproval Smith set then yields a principled way to satisfy all 
three blocking principles simultaneously.

In \Cref{sec:experiments}, we complement these theoretical results with experiments on five 
datasets spanning political elections, ethical AI evaluation, and 
moral decision-making scenarios. On every dataset, solution concepts that do not 
satisfy one of our principles also violate it empirically, confirming 
that the principles are not just abstract normative ideas but 
practically relevant and tangible.

\subsection{Preliminaries}

Let $N= \{1,\dots, n\}$ be a finite set of voters and $C$ a set of candidates. %
Throughout the paper, we consider two kinds of candidate sets $C$ with associated measure $\mu$. First, finite sets, in which case we denote as $m$  the number of candidates and set $\mu$ to be the counting measure, i.e., $\mu(X) = \frac{1}{m} \lvert X \rvert$ for all $X\subseteq C$. Second, polytopes in $\mathbb R^d$ for some $d\in \mathbb N$, in which case we set $\mu$ to be the normalized Lebesgue measure over the affine hull of $C$.
In social choice, voters are traditionally assumed to have transitive, complete preference relations over $C$. Subsuming this, we endow each of our voters $i\in N$ with a utility function $u_i : C\to \mathbb R$ and write $c\succ_i d$ if and only if $u_i(c)> u_i(d) $, as well as $c\succsim_i d$ if and only if $u_i(c)\ge u_i(d) $. When $C$ is a polytope, we assume $u_i$ to be continuous.\footnote{For finite $C$, clearly every complete and transitive preference relation can be represented by a utility function. Further, for infinite $C$, \citet{debreu1954representation} showed that every complete and transitive relation satisfying mild additional requirements can be represented by a continuous utility function.}
The collection $u= (u_i)_{i\in N}$ is called a \emph{utility profile}.
Further, each voter $i\in N$ has an \emph{approval set} $A_i\subseteq C$, and we say that \emph{$i$ approves $c$} if $c\in A_i$.\footnote{
Here, ``approval'' is shorthand for individual consent; we use the term to align with the social choice literature and to reserve ``consent'' for the collective notion.}
We assume voters' approval sets and utility functions are consistent in the sense that each voter has a threshold $\tau_i\in \mathbb R$ such that for each $c\in C$ it holds that $c \in A_i$ iff $u_i(c)\ge \tau_i$.
The collection $\tau=(\tau_i)_{i\in N}$ is called a \emph{threshold profile}. A \emph{societal consent instance} (or simply \emph{instance}) is a pair of a utility profile and threshold profile $\mathcal I = (u, \tau)$. Given an instance $\mathcal I$ and a candidate $c\in C$, let  \(\score{c}{\mathcal I} = |\{i\in N\mid c\in A_i\}|\) denote the \emph{approval score of $c$}, i.e., the number of voters who approve candidate \(c\) in \(\mathcal I\). 
A \emph{solution concept} $f$ maps each instance $\mathcal I =(u, \tau)$ to $f(\mathcal I)\subseteq C$, which we interpret as the set of candidates receiving \emph{societal consent}. Note that there are no size restrictions on this subset, i.e., $f$ may return the empty set in some instances, and the entire candidate set in others. 


\section{Blocking Framework and Solution Concepts}
\label{sec:solution-concepts}
In this section, we formalize our blocking principles from \Cref{blocking} and design first solution concepts.

\subsection{Local Blocking: Sufficient Support}
\label{subsec:threshold-concepts}

We begin with local blocking: a candidate $c$ may be unacceptable simply because too few voters approve it. The appropriate threshold is application-dependent: some settings demand broad backing by, say, $90\%$ of voters, while in others $40\%$ may be sufficient. We therefore parameterize the requirement by a threshold $\alpha\in[0,1]\cup\{\infty\}$ and say that a solution concept \(f\) has an \emph{\(\alpha\)-approval threshold} if, for every instance \(\mathcal I=(u,\tau)\) and \(c\in f(\mathcal I)\),
$\score{c}{\mathcal I}\geq \alpha n$ and $\alpha$ is the supremum for which the above holds. 
If \(\alpha\geq\frac12\), we also say that \(f\) satisfies \emph{majority threshold}.
The least restrictive solution concept satisfying the requirement accepts every candidate whose approval score reaches the threshold. Here, we use the convention that $\infty n $ is strictly larger than any natural number. 
\begin{definition}[\(\alpha\)-threshold concept]
    The \emph{\(\alpha\)-threshold concept} $f_\alpha$ returns
    $
        \{c\in C \mid \score{c}{\mathcal I}\geq \alpha n\}.
    $
\end{definition}

Threshold concepts implement an appealing paradigm of collective consent: society fixes a single standard, and each candidate is judged against it independently of the rest. We show that, in some formal sense, they are the \textit{only} concepts realizing this paradigm. 

We first formalize what it means to restrict instances: 
for $\mathcal I=(u,\tau)$, $\mathcal I' = (u',\tau')$ and a selection of candidates $D\subseteq C$, we say that the instances are \emph{equivalent restricted to $D$} and write $\mathcal I \equiv_D \mathcal I'$, if $\succ_i\mid_{D\times D} = \succ_i'\mid_{D\times D}$ and $c\in A_i \Longleftrightarrow c\in A_i'$ for all $i\in N$ and $c\in D$.
We can now formalize the idea of evaluating each candidate independently.
    A solution concept $f$ satisfies \emph{independence of other candidates}, if $c\in f(u, \tau) \Longleftrightarrow c\in f(u', \tau')$ whenever $(u, \tau) \equiv_{\{c\}}(u', \tau')$.

Further, for the criterion to treat all candidates fairly, it should hold that candidates in the same position are treated equally. For this purpose, we say $c,d$ are \emph{equivalent} in instance $\mathcal I$, if for all $e\in C$, and $i\in N$, we have that $u_i(c) = u_i(d)$.

    A solution concept satisfies \emph{anonymity}, if it is invariant to permutation of the voter identities. A solution concept satisfies equality, if in all instances $\mathcal I$ and all $c,d\in C$ with $u_i(c) = u_i(d)$ for all $i\in N$, we have that $c\in f(\mathcal I) \Longleftrightarrow d\in f(\mathcal I)$. 
     A solution concept satisfies \emph{approval monotonicity}, if $c\in f(u, \tau) \Longrightarrow c\in f(u, \tau')$  whenever some voter adds $c$ to their approval set by lowering their threshold to $\tau_i'= u_i(c) < \tau_i$ and nothing else changes.

The key axiom for the characterization is independence of other candidates, and \citet{lackner2025approval} use an analogous independence axiom to characterize an analogue of threshold concepts in the approval-based shortlisting setting.



\begin{restatable}[Appendix \ref{ssec:ThresholdChar}]{proposition}{ThresholdCharacterization}
    A solution concept is equal to a threshold concept if and only if it satisfies independence of other candidates, equality, anonymity, and approval monotonicity.
\end{restatable}



\subsection{Coalitional Blocking: Minority Protection}

Threshold concepts judge candidates by a single statistic, the approval score. 
But approval scores cannot tell whether disapproval of a candidate $c$ is scattered across 
voters who barely disapprove of $c$ and otherwise disagree on everything, or 
whether it comes from a cohesive minority that uniformly ranks $c$ near the 
bottom of its preferences. The latter case is precisely where coalitional 
blocking matters: a small but coherent group with strong shared objections 
should be able to rule $c$ out, even if other voters derive high value from it.

This motivates the formalization of blocking coalitions.
To maintain our goal of transparency, 
it is crucial that a concept implementing coalitional blocking lets us pinpoint, for every excluded candidate, the coalition responsible for the exclusion. Inspired by work from \citet{DoPe2026a} in multiwinner voting, we formalize 
this requirement through witness functions:
\begin{definition}
    A \emph{witness function} $w$ for a solution concept $f$ is a function that maps a candidate $c$ and instance $\mathcal I$ to a collection of coalitions $w(c,\mathcal I)\subseteq 2^N$ such that
$c\in f(\mathcal I)$ iff $ w(c,\mathcal I) = \emptyset$. 
\end{definition}
Each coalition $S\in w(c,\mathcal I)$ is called a \emph{witness} for $c$'s exclusion or a \emph{blocking coalition} for $c$.
By themselves, witness functions do not restrict which coalitions ought 
to be able to block; they only require that we can identify who is responsible for each 
exclusion---a minimal form of accountability. 

\paragraph{Quantifying minority protection.}
Next, we focus on the protection aspect of our desired principle: a coalition that 
strongly objects to $c$ should be able to exclude it from societal consent, 
regardless of what the rest of the electorate reports. To formalize the extent to which a solution concept allows for such protection of minorities, following \citet{Moulin1982}, we measure a coalition's blocking power 
through its \emph{veto power} $v_f(s)$: the largest measure of candidates a coalition of size $s$ can force $f$ to exclude, 
even if the remaining voters have adversarial preferences.
\begin{definition}[veto power] For a solution concept $f$ and coalition size $s \le n$, the \emph{veto power} \(v_f(s)\) is the supremum 
over all values $\nu$ such that for every coalition \(S\subseteq N\) with \(|S|=s\) and every well-behaved\footnote{In continuous spaces, we want to avoid technicalities and therefore call a set $T\subseteq C$ \emph{well-behaved} if there exists some continuous utility function $u_i$ over $C$ such that $u_i(c)> u_i(d)$ for all $c\in C\setminus T$ and $d\in T$.} set \(T\subseteq C\) with \(\mu(T)\leq \nu\), there is a partial instance \(\mathcal I_S\) such that
$f(\mathcal I_S,\mathcal I_{N\setminus S})\cap T=\emptyset$ 
for every partial instance \(\mathcal I_{N\setminus S}\). 
\end{definition}

The implications of veto power can be illustrated well in scenarios with a 
polarized electorate, where there is maximal disagreement between a minority 
group $S$ and the majority group $N\setminus S$, i.e., the candidates most liked by the majority are most harmful to the minority. The veto power then 
prescribes that the minority group $S$ can exclude a $v_f(|S|)$ fraction of 
candidates from societal consent. A minimal minority-protection requirement 
is thus $v_f(s) > 0$ for some $s < n/2$, while if one looks for
proportional veto power, $v_f(s)$ should scale with $s/n$. \citet{Moulin1981} showed that $s/n$ is the largest possible veto power that is achievable without empty societal consent sets.




\paragraph{The Disapproval Veto Core.}
Following the work of \citeauthor{Moulin1981}, we now define a solution concept that admits a witness function that exactly identifies the coalitions 
needed for proportional minority protection. 
For an instance 
$\mathcal I = (u,\tau)$, $S\subseteq N$, and $c \in C$, let 
$B_S(c)=\{d\in C \mid {d\succ_i c} \text{ for all } i\in S\}$ be the set of candidates that every voter in \(S\) strictly prefers to \(c\). Thus, \(B_S(c)\) measures how much of the candidate space the coalition unanimously regards as better than \(c\).

\begin{definition}[Disapproval Veto Core]
  Under the \emph{disapproval veto core} (\dvc), candidate \(c\) is \emph{blocked by a non-empty coalition \(S\subseteq N\)} if
(i) \(\mu(B_S(c)) > 1-\frac{|S|}{n}\), and
      (ii) \(c\notin A_i\) for all \(i\in S\). 
  The \dvc{} consists of all candidates that are not blocked; the associated witness function $w_\dvc$ returns for candidate $c$ the coalitions blocking $c$.
\end{definition}

Condition~(i) requires a blocking coalition to unanimously rank more than a 
$1 - \nicefrac{|S|}{n}$ fraction of the candidate space above $c$; 
equivalently, $c$ lies in the coalition's bottom $\nicefrac{|S|}{n}$-part of 
the candidate space.
Condition~(ii) ensures a voter cannot help block a candidate she approves: 
approval of $c$ disqualifies a voter from contributing to its exclusion. Blocking coalitions under DVC satisfy four desirable properties, which we formalize next using the witness-based framework. 

First, and normatively most important, for minority protection, we consider idealized instances $\mathcal I_S$ where the electorate is polarized, i.e., split into two perfectly cohesive parts, a minority $S\subseteq N$ and a majority $N\setminus S$ with opposing ordinal preferences, such that all voters in $S$ disapprove of all candidates in $C$. 
A natural demand is that if this minority constitutes $\frac sn$ of the electorate, they should be able to veto up to the most harmful $\frac sn$ candidates. 
Therefore, let $C'\subseteq C$ be a lower contour set of the voters in $S$ with $\mu(C')<\omega \frac sn$.
We say that $(f,w)$ satisfies \emph{polarized minority protection}
if $S\in w(c,\mathcal I_S)$ for all $S\subseteq N$ and $c\in C'$. 

Second, consider a set of voters $S$ which blocks a candidate $c$. If one of these voters expands 
their approval set without changing the approval of $c$, then for this voter, the quality of $c$ remains as before, while for all other candidates their quality can only have increased. As this voter became more accepting of other candidates, $c$ should remain blocked by $S$. Formally, we say that $(f,w)$ satisfies \emph{approval competition monotonicity}, if $S\in w(c, (u, \tau))$ implies that $S\in w(c,(u, \tau'))$ for all profiles $(u, \tau)$ and $\tau'$ derived from $\tau$ by letting one voter approve more candidates but not change their approval of $c$.

In a similar spirit, if a voter changes her ranking such that afterwards, there are more candidates strictly better than $c$, then $c$ should remain blocked. Formally, $(f,w)$ satisfies  \emph{upper contour monotonicity}, if $S\in w(c,(u, \tau))$ implies that $S\in w(c,(u', \tau'))$, where $u'$ is derived from $u$ by letting a voter $i\in S$ change $u_i$ to $u_i'$ such that $u_i(d) \ge \tau_i \Longleftrightarrow u'_i(d) \ge \tau'_i$ and 
$\{d\in C\mid u_i(d)> u_i(c)\} \subseteq \{d\in C\mid u'_i(d)> u'_i(c)\}$.

Finally, \emph{locality} states that a whether a set $S$ is a witness does not depend on voters outside of $S$: $S\in w(c,(u,\tau)) \Longleftrightarrow S\in w(c,(u',\tau'))$ whenever the two instances coincide over $S$, i.e., $u\mid_{S} = u'\mid_{S}$ and $\tau\mid_{S} = \tau'\mid_{S}$.\footnote{Note that locality only is postulated for instances of equal electorate size.} 

Importantly, only \dvc{} refinements can satisfy all these properties.
Thus, DVC constitutes the minimal blocking requirement imposed by admitting a witness function satisfying these axioms. In the finite setting with rankings and approvals, upper contour monotonicity can be split into candidate-specific modifications, which enables an instructive step-by-step explanation starting from a polarized profile for why a candidate was blocked (Appendix \ref{ssec:WitnessCharDVC}).


\begin{restatable}{theorem}{WitnessCharacterizationGeneral} \label{prop:DVC_Witness_Characterization_Continuous}
    Let $w$ be a witness function for $f$. If $(f,w)$ satisfies polarized minority protection, locality, upper contour monotonicity, and approval competition monotonicity, then $f\subseteq DVC$.\footnote{A full proof characterizing a class of solution concepts containing $\dvc$ appears in Appendix \ref{ssec:WitnessCharDVC}.}   
\end{restatable}

\subsection{Relational Blocking: Decisively Better Candidates}
\label{subsec:global-blocking}
The third blocking principle works via comparing candidates: a candidate may be unacceptable to society because a decisively better one exists. 
Canonical comparisons arise when unanimously approved candidates are present: such candidates should be considered strictly better than all candidates disapproved by at least one voter, and these should be blocked. 


\begin{definition}[exclusive unanimity]
Let $U(\mathcal I) =\{c\in C \mid \score{c}{\mathcal I}=n\}$ denote unanimously approved candidates. 
A solution concept \(f\) satisfies \emph{exclusive unanimity} if for every instance \(\mathcal I\) with $U(\mathcal I) \neq \emptyset$, 
$f(\mathcal I) \subseteq U(\mathcal I)$. 
$f$ satisfies \emph{unanimous approval consistency} if $f(\mathcal I) = U(\mathcal I)$ whenever $U(\mathcal I) \neq \emptyset$.
\end{definition}

To construct a solution concept based on this principle, we turn to the standard tool for comparing candidates in social choice theory: the pairwise majority relation. This relation may be cyclic, however \citep{CondorcetCycles}.  The classical response is to consider the \textit{Smith set}---the set of candidates that reach every other candidate via paths along the majority relation---which serves as a baseline notion of quality with respect to pairwise comparisons.
To adapt the Smith set to our setting, we incorporate approval information as follows: 
for an instance \(\mathcal I=(u,\tau)\),  define the \emph{disapproval-sensitive majority relation} \(R=R(\mathcal I)\) by setting \(c \mathrel{R} d\) if and only if 
$n(c,d) \ge n(d,c)$, where 
\[n(c,d) = \left |\{i\in N \mid c\in A_i,\ d\notin A_i\} \right| + \left|\{i\in N \mid c\succ_i d,\ c,d\notin A_i\}\right|.\]
With respect to $R$, unanimously approved candidates win the pairwise comparison against all non-unanimously approved candidates. We therefore consider the Smith set over this relation.

\begin{definition}[disapproval Smith set]
    The \emph{disapproval Smith set} is
    $
        \tc(\mathcal I)
        =
        \{c\in C \mid c \mathrel{R^*} d \text{ for every } d\in C\},
    $
    where \(R^*\) denotes the transitive closure of \(R\), i.e., \(c \mathrel{R^*} d\) if there exists a finite sequence of candidates $
    c=d_0,d_1,\dots,d_t=d$
such that \(d_{j-1} \mathrel{R} d_j\) for every \(j=1,\dots,t\).
\end{definition}

Under this solution concept, the blocked candidates are exactly those outside $\tc(\mathcal I)$, i.e., the candidates that fail to reach some other candidate via $R^*$.
\begin{restatable}{theorem}{dss}
The disapproval Smith set satisfies unanimous approval consistency. 
\end{restatable}

\section{Veto Budgets for Combining 
Blocking Principles}
\label{sec:Endowments_Char_DVC}

In the previous section, we introduced solution concepts motivated by individual blocking principles. We first confirm that each concept adheres to its corresponding principle and only this principle.

\begin{restatable}{proposition}{RulesAndPrinciples}\label{prop:rules_and_principles}
$\dvc$ satisfies $v_{\dvc}(s) = \frac sn$ and $f_\alpha$ satisfies $\alpha$-approval threshold.\footnote{Note that we have already stated unanimous approval consistency of $\tc$.}
Further, $\dvc$ and $\tc$ have an approval threshold of $0$.
$f_\alpha$ has a veto power function with $v(s) = 0$ for all $s\le n- \alpha n$, $\tc$ has $v_\tc(s)= 0$ for all $s< \frac n2$.
$\dvc$ and $f_\alpha$, $\alpha<1$ violate exclusive unanimity.
\end{restatable}

We now seek concepts that jointly 
satisfy multiple principles, and we will do so by refining \dvc{}. 
To gain a better understanding of \dvc{} and motivate the following definitions, we 
characterize it via a \emph{veto budgeting game} in which 
voters receive individual budgets to spend on excluding candidates. 
Scaling voters' budgets 
then yields variants of \dvc{} with increased veto power and non-trivial approval thresholds. En passant, we obtain a second explanation system for \dvc{}.

\subsection{Veto Budgeting Games and a Family of Disapproval Veto Cores}\label{ssec:veto_budgeting_games}

Consider a \emph{veto budgeting game} in which each voter receives budget $1/n$ and allocates it across candidates via a measurable function $\lambda_i:C\to[0,1]$ satisfying
$\int_C \lambda_i(c)\,d\mu < \frac{1}{n}$. Given an allocation profile $\lambda=(\lambda_i)_{i\in N}$, the candidates receiving societal consent are those not fully purchased for exclusion:
$
    W(\lambda)
    =
    \left\{
        c\in C
        \mid
        \sum_{i\in N}\lambda_i(c)<1
    \right\}.$
For $T\subseteq C$, set $p_i(T) = \inf_{c\in T}{u_i(c)}$ as the \emph{pessimist utility} of voter $i\in N$ for $T$.
A coalition $S$ has an \emph{$\varepsilon$-deviation} from $\lambda$ if it can reallocate its own budget via some $\lambda'_S$ so that, for every response $\lambda'_{N\setminus S}$ of the remaining voters, every member of $S$ improves her pessimist value by at least $\varepsilon$:
$p_i( W(\lambda') )
    \ge
    p_i( W(\lambda) )+\varepsilon$ for all $i \in S$. 
From a normative perspective, when determining which candidates receive societal consent, deviations arising from voters participating in blocking their approved candidates are illegitimate. We therefore say that a deviation $\lambda'_S$ is \emph{disapproval-restricted} if $c \notin A_i$ for all $i\in S$ and $c\notin W(\lambda'_S)$, i.e., a voter refuses to partake in a deviation if this results in some of her approved candidates being blocked.
\begin{theorem}[Veto-budget characterization of \dvc; Appendix~\ref{app:RelWorkBudgetingGames}]\label{thm:GT_Char_PVC}

For every instance $\mathcal I=(u,\tau)$ and candidate $c\in C$, we have $c\notin \dvc(\mathcal I)$ iff in the veto budgeting game there exists $\varepsilon>0$ such that every allocation profile $\lambda$ with $c\in W(\lambda)$ admits a disapproval-restricted $\varepsilon$-deviation.
\end{theorem}


\Cref{thm:GT_Char_PVC} yields a second explanation system for DVC: rejection corresponds to the existence of a disapproving coalition that can challenge every allocation sparing $c$.
Conversely, if $c$ is selected by \dvc{}, then there exists an allocation witnessing that no such disapproving coalition can justify excluding it.
Appendix~\ref{app:RelWorkBudgetingGames} discusses related work \citep{Moulin1981,KizilkayaKempe2023,BudgetingGames}.

Moreover, the budgeting-game interpretation of \dvc{} suggests a natural family of 
solution concepts obtained by scaling voters' veto budgets. In the standard 
\dvc{}, each voter has budget $1/n$, giving a coalition $S$ aggregate 
budget $|S|/n$. We introduce a parameter $\omega \geq 0$ that multiplies 
these budgets.

\begin{definition}[$\omega$-DVC]
  Let $\omega \ge 0$ be given. Under the \emph{$\omega$-disapproval veto core ($\omega$-\dvc)}, candidate~$c$ is \emph{blocked by a non-empty coalition} $S\subseteq N$ if (i) ${\mu(B_S(c))} > 1-\omega \nicefrac{s}{n}$ and
  (ii) $c \notin A_i$ for all $i\in S$. $\omega$-\dvc{} consists of all candidates that are not blocked.
\end{definition}

For $\omega = 1$, this coincides with \dvc{}; increasing $\omega$ grants 
every coalition more blocking power and leads to the following combination of sufficient-support and minority protection: 
\begin{restatable}{proposition}{omegadvc}\label{omega-dvc}
$\omega$-\dvc{} has an approval threshold of $\alpha=\max(0,1-\nicefrac{1}{\omega})$ and veto power ${v(s)=\omega \nicefrac{s}{n}}$ for $s \le \frac n \omega$.
\end{restatable}
\definecolor{vetoblue}{RGB}{31,119,180}
\definecolor{vetoorange}{RGB}{255,127,14}
\definecolor{vetogreen}{RGB}{44,160,44}

\begin{figure}[t]
\centering
\begin{minipage}[t]{0.46\textwidth}
    \centering
    \vspace{0pt}
    \begin{tikzpicture}
    \begin{axis}[
        width=\linewidth,
        height=0.5\linewidth,
        xlabel={$s$},
        ylabel={$v(s)$},
        xlabel near ticks,
        ylabel near ticks,
        xlabel style={font=\small},
        ylabel style={font=\small},
        x tick label style={font=\footnotesize},
        y tick label style={font=\footnotesize},
        xmin=0, xmax=50,
        ymin=0, ymax=1.05,
        xtick={0,10,20,30,40,50},
        ytick={0,0.25,0.5,0.75,1},
        ymajorgrids=true,
        xmajorgrids=true,
        grid style={dashed, gray!25},
        legend style={
            at={(0.03,0.97)},
            anchor=north west,
            draw=none,
            fill=none,
            font=\footnotesize,
            row sep=-1pt,
            inner sep=1pt,
        },
        legend cell align=left,
        samples=200,
        domain=0:50,
        smooth,
        no markers,
        axis line style={gray!60},
        tick style={gray!60},
        clip=false,
    ]
    \addplot[very thick, color=vetoblue]   {x/100};
    \addlegendentry{\dvc}
    \addplot[very thick, color=vetoorange] {2*x/100};
    \addlegendentry{$2$-\dvc}
    \addplot[very thick, color=vetogreen]  {x/(100-x)};
    \addlegendentry{AWVC}
    \end{axis}
\end{tikzpicture}
    \caption{Comparison of veto power in AWVC, DVC, and $2$-DVC for $n=100$.}
    \label{fig:veto-power-plot}
\end{minipage}
\hfill
\begin{minipage}[t]{0.50\textwidth}
    \centering
    \vspace{0pt}
    \small
    \setlength{\tabcolsep}{4pt}
    \renewcommand{\arraystretch}{1.12}

    \resizebox{\linewidth}{!}{%
    \begin{tabular}{@{} l l c c @{}}
    \toprule
    Concept
    & \makecell[c]{Veto power\\of minorities}
    & \makecell[c]{Approval\\threshold}
    & \makecell[c]{Unan. approval\\consistency} \\
    \midrule
    $\omega$-\dvc
    & $v(s)=\omega \frac sn$
    & $\max(0,1-\nicefrac{1}{\omega})$
    & $\times$\\

    AWVC
    & $v(s)= \frac s {(n-s)}$
    & $\alpha = \frac{1}{2}$
    & $\times$\\

    \tc
    & $v(s) = 0 \, \forall s< \frac n2$
    & $\alpha = 0$
    & $\checkmark$ \\

    $\tc\circ$AWVC
    & $v(s)= \frac s {(n-s)}$
    & $\alpha = \frac{1}{2}$
    & $\checkmark$\\
    \bottomrule
    \end{tabular}%
    }
    \captionof{table}{Comparison of solution concepts and blocking principles.}
    \label{tab:rules_and_principles}
\end{minipage}
\vspace{-0.5cm}
\end{figure}


\subsection{The Goldilocks Zone Between Proportional Veto Power and Majority Disapproval}
\Cref{omega-dvc} shows that uniform budget scaling creates a simple but rigid tradeoff: increasing \(\omega\) raises the approval threshold of \(\omega\)-\dvc{}, but it also multiplies the veto power of every coalition by the same factor. In particular, \(2\)-\dvc{} satisfies majority threshold but doubles the veto power of all minorities. This is undesirable, as small coalitions and majorities are diametrical to each other: to satisfy one principle, we should not be forced to uniformly increase requirements on the other.
We therefore seek a solution concept in the ``Goldilocks zone'': one that allows a majority of voters to block disapproved candidates without granting small coalitions the full force of uniformly scaled veto power. The \emph{approval-weighted veto core} implements this idea by making voter endowments candidate-specific: a coalition's blocking power is measured relative to the candidate's number of approvers, rather than the full electorate.

\begin{definition}[Approval-Weighted Veto Core]
Under the \emph{approval-weighted veto core} (AWVC), a candidate \(c\) with
\(\score{c}{\mathcal I}>0\) is \emph{blocked by a non-empty coalition}
\(S\subseteq N\) if (i) \( \nicefrac{|S|}{\score{c}{\mathcal I}} > 1-\mu(B_S(c))\) and (ii) \(c \notin A_i\) for all \(i\in S\). Candidates with \(\score{c}{\mathcal I}=0\) are blocked by convention. The
AWVC consists of all candidates that are not blocked.
\end{definition}

Indeed, AWVC interpolates between \(\dvc{}\) and \(2\)-\dvc{} while combining the majority-threshold guarantee of \(2\)-\dvc{} and a veto-power function close to that of $\dvc{}$ for small coalition sizes.

\begin{restatable}{theorem}{AWVCThm}
    The Approval-Weighted Veto Core has veto power \(v(s)=\frac{s}{n-s}\) for $s\le \frac n2$ and
approval threshold \(\alpha=\frac12\). Moreover,
$2\text{-}\dvc{}\subseteq AWVC\subseteq \dvc{}$.    
\end{restatable}
The veto power formula \(v(s)=s/(n-s)\) explains the Goldilocks behavior. To compute the veto power of a coalition \(S\) of size \(s\), we consider the hardest completion for that coalition: all voters outside \(S\) approve the candidate \(c\), so \(\score{c}{\mathcal I}=n-s\). AWVC then measures the coalition's blocking power as \(s/(n-s)\). 
This yields a smooth interpolation between \dvc{} and \(2\)-\dvc{} (see \Cref{fig:veto-power-plot}): for small coalitions, \(s/(n-s)\) is close to \(s/n\), while for coalitions close to a majority it approaches \(2s/n\). Thus, AWVC obtains majority threshold without uniformly doubling the veto power of all minorities. 
An additional bonus is that AWVC satisfies the following monotonicity property.

\begin{restatable}{proposition}{monotondis}
AWVC satisfies \emph{support monotonicity of disapproving voters}: if a voter $i$
joins the instance with $c\notin A_i$, then
$c\notin AWVC(\mathcal I)$ implies
$c\notin AWVC(\mathcal I\cup\{(u_i,\tau_i)\})$.
\end{restatable}

By contrast, $\omega$-\dvc{} concepts may at first block candidates but grant them societal consent after additional disapprovers of the candidate join the election, which can regularly be observed in experiments (see Appendix \ref{app:RelWorkBudgetingGames} and \ref{App:SuppMonDisapprovingVoters} for further information).
Depending on the application, this may be viewed as a serious flaw of the concept.

We conclude by noting that we can achieve all three blocking principles by composition of solution concepts and refer to \Cref{tab:rules_and_principles} for a summary of guarantees and to Appendix \ref{sapp:RulesAndPrinciples} for omitted proofs and the formal definition of $\tc$ concatenation. We further generalize the following result.

\begin{restatable}{theorem}{ConcatenateBlockingGuarantees}
    $f= \tc \, \circ \,${AWVC} inherits $v_f(s) = \frac{s}{n-s}$ for $s\le \frac n2$, approval threshold $\alpha = \frac 12$, and unanimous approval consistency.
\end{restatable}

\section{Empirical Evaluation}\label{sec:experiments}

In this section, we evaluate our solution concepts empirically with a special focus on the extent to which they satisfy our notions of sufficient support, minority protection, and presence of decisively better candidates. 
We give a summary of our setup and findings here and relegate details to \Cref{app:exp}. All our solution concepts are tractable for finite $C$ (Appendix \ref{app:finite_C_Computation}).

\paragraph{Setup} Besides our concepts $0.5$-threshold (majority threshold), disapproval veto core (DVC),  approval-weighted veto core (AWVC),  disapproval Smith set applied on top of the solution produced by approval-weighted veto core
(DSS $\circ$ AWVC),  disapproval Smith set (DSS), we consider two baselines: under Approval Relative Utilitarianism (ARU) (cf.~\citealt{DhillonMertens1999}), a candidate has societal consent if its average utility exceeds the voters' average approval threshold. Under Minimize Margin of Disapproval (MMOD), we select all candidates $c$ that minimize $\sum_{i\in N} \max(0, \tau_i - u_i(c))$.

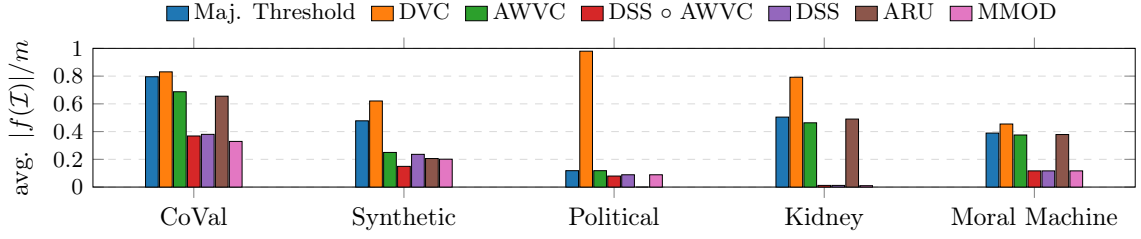
\begin{figure}[t]
    \centering
    \resizebox{\textwidth}{!}{%
        \input{Current_Work/Neurips/plots/diagram1_selection_fraction.tex}%
    }
    \caption{%
    For every (solution concept, dataset) pair, we plot the average fraction \(|f(\mathcal I)|/m\) of selected candidates
    over all instances \(\mathcal I\) in that dataset, where \(m\) is the number
    of candidates in the instance.
}
    \label{fig:selection_fraction}
\end{figure}

We consider five datasets\footnote{Note that in the political and CoVal datasets, approval and utility information need not be consistent with each other, i.e., there are voters for which a disapproved candidate might have a higher utility than an approved candidate. We adjust our solution concept to deal with these as described in Appendix \ref{app:threshold_based_rule_fix}.} with finite candidate sets, whose size we denote as $m$. 
\emph{CoVal}: $319$ instances in which annotators rank four AI-generated responses to ethical prompts and optionally flag 
responses as unacceptable \citep{coval2026}. 
\emph{Synthetic}: $600$ instances drawn from six configurations of standard 
statistical cultures, with $50$ voters and 
$5$ candidates.
\emph{Political}:  Nine instances from surveys in real-world political elections
in which participants provided both approval and cardinal/score information.
\emph{Kidney Exchange}: $60$ instances built from a survey on kidney allocation priorities with $301$ voters who reported a weight vector over $8$ features of recipients \citep{keswani2026moral}.  
Candidates are $100$ uniformly sampled recipients, where a voter approves a candidate iff her utility for it 
is positive.
\emph{Moral Machine}: $60$ instances with $n=137$ 
voters whose 24-dimensional moral-feature weights, learned from 
autonomous-vehicle dilemmas \citep{awad2018moral}, induce linear utilities over $m=100$ 
sampled policy vectors. Approvals are generated either by per-voter 
top-$\alpha$ thresholds or by agreement on sampled dilemma pairs.

\paragraph{Selectiveness (\Cref{fig:selection_fraction})}
\Cref{fig:selection_fraction} shows, for each solution concept and 
dataset, the average fraction of candidates $|f(\mathcal I)|/m$ the concept 
selects on instances from the dataset. \dvc{} is consistently the most 
permissive concept, indicating that 
few coalitions of voters jointly dislike enough of the candidate space 
to trigger blocking. In line with this, AWVC selects roughly as many 
candidates as Maj.~Threshold throughout. The most striking effect comes 
from incorporating the disapproval Smith set: $\tc$ and $\tc\circ$AWVC 
narrow sharply (down to near-singletons on Kidney), revealing that the 
other concepts admit many candidates that are dominated by 
unanimously-better candidates.

\paragraph{Sufficient Support} We evaluate whether our solution concepts ensure that each selected candidate is approved by a sufficient number of voters.
We find that \dvc{} fails dramatically on every dataset, selecting a 
candidate that is approved by less than half of the voters in $52$ of $319$ CoVal elections, all $9$ 
Political elections, and all $60$ instances of both Kidney and Moral 
Machine; on Political and Kidney, the median of the minimum approval fraction of a selected candidate drops to 
$0.04$ and $0.08$, meaning \dvc{} routinely admits candidates almost 
nobody approves of. The disapproval Smith set also fails occasionally, 
as it carries no approval-threshold guarantee on its own, i.e., non-majority approved candidates are selected on $7$ CoVal, 
$1$ Political, and $106$ Synthetic violations. As guaranteed, AWVC produces no violations on any dataset 
and generally selects broadly-supported candidates, with median 
worst-approval ranging from $0.51$ (Kidney) to $0.80$ (CoVal). 
Composing AWVC with the disapproval Smith set strengthens this further, 
pushing the median worst-approval to $0.96$ on Kidney and to $1.00$ on 
CoVal and Moral Machine (see \Cref{tab:submajority}).

\begin{table}[t]
    \centering
    \resizebox{\textwidth}{!}{%
        \input{Current_Work/Neurips/plots/diagram5_low_omega.tex}    }
    \caption{%
     \emph{Minority Protection}. The first column (\#) gives
    the number of instances in which the concept selected at least one
    candidate with \(\omega^*(c)<1\). The
    second column (med \(\omega^*\)) gives the median across elections
    of \(\min_{c\in f(\mathcal I)} \omega^*(c)\); this value can be  \(\infty\)
    in case all selected candidates are unanimously approved
    and therefore unblockable.\vspace{-0.5cm}
    }
    \label{tab:low_omega}
\end{table}

\paragraph{Minority Protection (\Cref{tab:low_omega})}
We would like to measure to what extent our solution concepts satisfy minority protection on our datasets; note that a concept's veto power is not a suitable metric for this, as it is defined as a worst-case measure over all instances (instead of making a statement about a specific instance).
Furthermore, quantifying for a candidate to what extent its selection violates minority protection is much harder than for sufficient support. 
Therefore, we turn to the veto budgeting perspective from \Cref{thm:GT_Char_PVC} and for a given candidate $c$ ask for the maximum $\omega$ such that $c$ remains unblocked in $\omega$-DVC:
Formally, 
    we say that $c$ has \emph{critical endowment} $\omega^*(c) {\ge 0}$ if $\omega^*(c) = \max \{\omega \ge 0 \mid \text{ $c\in$ $\omega$-\dvc}\}$.\footnote{The measure is related to the critical-epsilon measure of \citet{chooi2026findingcommongroundsea}; we discuss the relation in Appendix~\ref{app:RelWorkBudgetingGames}.}
Interpreting this measure using \Cref{thm:GT_Char_PVC}, \(\omega^*(c)\) measures the robustness of \(c\) to increased veto budgets. 
If \(\omega^*(c)>1\), then \(c\) 
can be blocked only when voters receive more than the standard budget. 
If \(\omega^*(c)<1\), then \(c\) is vulnerable and can already be blocked with a fraction of the standard budget. 
Fortunately, the critical endowment for a candidate can be computed efficiently.

\begin{restatable}{theorem}{measurecomp}\label{thm:measure}Let $c\in C$.
The critical endowment $\omega^*(c)$ can be computed in time $O((mn)^{1+o(1)})$ when $C$ is finite, where $m$ denotes the number of candidates in $C$.
\end{restatable}

Accordingly, in \Cref{tab:low_omega}, we show for each concept and dataset 
the number of instances in which a candidate $c$ with $\omega^*(c) < 1$ 
gets selected, as well as the median across elections of the worst 
critical endowment $\min_{c \in f(\mathcal I)} \omega^*(c)$.
The concepts without coalitional guarantees 
fail noticeably: Maj.~Threshold selects candidates with 
$\omega^*(c) < 1$ in $50$ of $319$ CoVal and $37$ of $60$ Kidney 
elections, and ARU fails on $27$ of $60$ Kidney elections.
Comparing the concepts with minority protection guarantees, AWVC dominates \dvc{} on the 
median worst-endowment ($1.91$--$2.55$ across datasets versus 
\dvc{}'s $1.00$--$1.75$), and composing with the disapproval Smith set 
strengthens this further: $\tc\circ$AWVC reaches median $\infty$ on 
CoVal and Moral Machine, where every selected candidate is 
unanimously approved and therefore unblockable under any 
$\omega$-\dvc{}. 

\paragraph{Decisively Better Candidates}
Only the CoVal ($175$ out of $319$)  and Moral Machine ($60$ out of $60$) datasets contain elections with at least 
one unanimously approved candidate. On these datasets, every concept that does not guarantee exclusive unanimity indeed provides violations regularly:
Maj.~Threshold, 
\dvc{}, AWVC, and ARU each violate exclusivity on $132$--$154$ of the 
$175$ CoVal elections and on all $60$ Moral Machine elections.  
\dvc{} produces the most violations here, with a total of $304$ not unanimously approved candidates on CoVal and $2028$ on Moral Machine (see \Cref{tab:unanimity}).

\section{Discussion}
We initiate the analysis of deriving societal consent through the lens of social choice theory. Identifying three main blocking principles: local, coalitional, and relational, we formalize these and propose canonical solution concepts satisfying one or multiple of these formalizations. We analyze these concepts axiomatically and experimentally, observing that solution concepts without theoretical guarantees exhibit these flaws in practice, too. There are several promising strands of future work. First, this paper assumes that voters act truthfully; it would be interesting to explore to what extent our concepts exhibit incentive compatibility properties. Second, an important open problem is the computation or approximation of the \textit{entire} proportional veto core and \dvc{}-variations for infinite $C$. The boundary shapes of the sets selected by these concepts appear to be intricate, e.g.,   for \dvc{} they can be non-convex, so even the question of whether they can be represented concisely is open. Third, we provide theoretical explanation systems for blocking by $\dvc$ variations. For future work, it would be very interesting to expand these theoretical results to algorithms that, given an election and a candidate as input, automatically generate explanation systems and to examine how plausible these are to humans. Fourth, while AWVC naturally interpolates between $v(s) = \frac sn$ veto power for small $s$ and majority threshold for $s$ close to $\frac n2$, it remains open whether there are natural concepts that do the same for other parameter choices, for example, a super-majority threshold of $\frac 23$. 

\section*{Acknowledgments}
This work was supported by the Advanced Research + Invention Agency (ARIA) under the project ``Aggregating Safety Preferences for AI Systems:
A Social Choice Approach (ASPAI).''
We thank the reviewers of the 3rd Workshop on New Directions in Social Choice at EC 2026 for their helpful feedback.

\bibliography{bibliography}

\appendix

%
\section{Related Work}\label{app:RelWork}
Our framework connects to three strands of social choice theory: 
(i) frameworks for selecting variable-sized subsets of candidates, 
(ii) the proportional veto core literature on minority protection, 
and (iii) models combining preference and approval input. In particular, all related works discussed in parts (i) and (ii) assume access to either approval or cardinal/ordinal preferences.

\paragraph{Selecting a Variable-Sized Subset of Candidates.} 
The standard setting in social choice is single-winner voting~\citep{HandbookofCSC}, 
where the goal is to identify a single best candidate; sets of candidates are 
returned only when the rule identifies multiple candidates as being tied for winning. Conceptually, this differs 
from our setting in two ways: single-winner rules aim to identify the single best candidate 
and return multiple candidates only when this question cannot be cleanly answered, 
whereas society may legitimately consent to multiple incomparable candidates, or to none at all, which is not allowed in the world of single-winner voting.

A second related line is multiwinner voting, in which, given an integer $k$, the goal 
is to select a $k$-candidate subset~\citep{lackner2023multi,faliszewski2017multiwinner}. Multiwinner rules 
typically implement one of three principles: individual excellence (selecting the $k$ best candidates), 
proportionality (ensuring that cohesive groups of voters receive representation in the committee by candidates they like proportional to their size), or diversity 
(ensuring each voter likes at least one selected candidate). Our setting differs in three 
respects: (a) we impose no fixed size on the consent set; (b) we require every 
accepted candidate to clear a quality threshold, which neither proportionality nor 
diversity guarantees, since both can admit candidates liked by only a small subgroup; 
and (c) approvals and disapprovals are treated asymmetrically, with disapprovals 
carrying in some sense more weight in determining consent through our notion of minority protection. 

With regards to point (c), a notable exception within multiwinner voting is \citet{kraiczy2025proportionality}, 
who study trichotomous (up/down/neutral) preferences and propose a group-veto notion: 
a coalition cohesively disapproving of a candidate set may prevent too many of their 
disapproved candidates from joining the committee. However, their notion of veto operates on 
\emph{combinations} of candidates, whereas in our setting consent is decided 
per candidate. One could, in principle, refine the group-veto of \citet{kraiczy2025proportionality} using ordinal preferences, which would ultimately yield 
a weakening of $\omega$-DVC.

A few works in multiwinner voting allow the committee size to vary~\citep{faliszewski2020multiwinner,freeman2020proportionality,lackner2025approval}. 
Conceptually closest is \citet{lackner2025approval}, who study 
\emph{approval-based shortlisting with variable committee size} using only approval 
ballots. Among other results, they present a threshold rule characterization among shortlisting rules using an independence axiom where a candidate is chosen or not chosen based on its approval score.  They additionally 
combine unanimous approvals and disapprovals with $\ell$-stability, requiring 
candidates whose approval scores differ by less than $\ell$ to be selected together. Also related, though not part of multiwinner voting, is the work of \citet{brandl2019axiomatic}, who axiomatically characterize the rule which, given transitive preferences of the voters, outputs all candidates with above average Borda score.
Our setting differs in that we have transitive input—preferences \emph{and} approvals—which 
both enables and necessitates different axioms; e.g., $\ell$-stability is not 
appropriate in our setting.

\paragraph{Minority Protection and the Proportional Veto Core.}
A central concern of our work is minority protection, which we formalize via the 
notion of veto power introduced by \citet{Moulin1981,Moulin1982}. For our solution concepts, the proportional 
veto core (PVC) of \citet{Moulin1981,Moulin1982} developed for ordinal preferences provides the natural starting point, 
and we develop several adaptations to our setting. The PVC has since been 
applied to settings with infinite candidate spaces and purely ordinal preferences, 
typically with the goal of computing or querying some statement, policy, or candidate 
within the core: \citet{VolumePVC} use it for federated learning, 
\citet{chooi2026findingcommongroundsea} for aggregating individual statements via 
generative AI, and \citet{alamdari2024policyaggregation} for policy aggregation.

More recently, \citet{halpern2025ApprovalVC} define a proportional veto principle for approval preferences. 
Because cohesive coalitions are rare under approvals alone, 
they relax coalitional cohesiveness and let voters jointly block a disapproved candidate 
when each individual voter strictly prefers (i.e., approves) sufficiently many alternatives. 
They additionally propose an ordinal strengthening in their appendix; this version 
turns out to be too restrictive for our purposes, as it, e.g., blocks all candidates not top ranked by any voter.

Several works further develop the theoretical and algorithmic foundations of PVC and its variants. 
\citet{ianovski2023computingproportionalvetocore} give a polynomial-time algorithm 
for computing the PVC with finitely many candidates. \citet{kondratev2024veto} 
systematically study variants of the PVC, including the eating rule, population 
monotonicity axioms, independence of last-ranked candidates, and asymptotic 
choice-set size. \citet{KizilkayaKempe2023,kizilkaya2025k} generalize the PVC to 
obtain favorable metric distortion guarantees.

Further, we note that our modifications to the proportional veto core address an issue similar to one that \citet{halpern2025ApprovalVC} raise: PVC may select unanimously disapproved candidates (\Cref{fig:PVC-Bad}).
\begin{figure}[t]
    \vspace{0pt}
    \centering
\begin{tikzpicture}[baseline=(current bounding box.north),
    x=0.58cm,
    y=0.62cm,
    voter/.style={font=\small},
    cand/.style={
        circle,
        minimum size=4mm,
        inner sep=0pt,
        font=\footnotesize\itshape,
        text=black,
        draw=black!20,
        line width=0.25pt
    },
    appr/.style={
        rounded corners=0.8mm,
        draw=black!30,
        fill=approve,
        minimum width=5.6mm,
        minimum height=5.6mm
    },
    sep/.style={font=\scriptsize, text=black!70},
    aset/.style={font=\scriptsize},
    scale = 1.2
]

\definecolor{candone}{RGB}{141,160,203}   
\definecolor{candtwo}{RGB}{102,194,165}   
\definecolor{candthree}{RGB}{252,141,98}  
\definecolor{candnmone}{RGB}{166,216,84}  
\definecolor{candn}{RGB}{231,138,195}     
\definecolor{candd}{RGB}{240,128,128}     

\def\xvoter{0}
\def\xstart{1.45}
\def\xgap{1.55}
\def\ysep{1.1}

\newcommand{\candColor}[5]{%
    \node[cand,fill=#5] (#1) at (#2,#3) {#4};
}

\newcommand{\rowbase}[9]{%
    \node[voter] at (\xvoter,#3) {$v_{#2}$};

    \candColor{r#1a}{\xstart}{#3}{#4}{#5}
    \node[sep] at (\xstart+0.5*\xgap,#3) {$\succ$};

    \node[sep] (r#1b) at (\xstart+1*\xgap,#3) {$\cdots$};
    \node[sep] at (\xstart+1.5*\xgap,#3) {$\succ$};

    \candColor{r#1c}{\xstart+2*\xgap}{#3}{#6}{#7}
    \node[sep] at (\xstart+2.5*\xgap,#3) {$\succ$};

    \candColor{r#1d}{\xstart+3*\xgap}{#3}{$d$}{candd}
    \node[sep] at (\xstart+3.5*\xgap,#3) {$\succ$};

    \candColor{r#1e}{\xstart+4*\xgap}{#3}{#8}{#9}
}

\rowbase{1}{1}{0}{$c_2$}{candtwo}{$c_n$}{candn}{$c_1$}{candone}
\rowbase{2}{2}{-\ysep}{$c_3$}{candthree}{$c_1$}{candone}{$c_2$}{candtwo}

\node[voter] at (\xvoter,-2*\ysep) {$\vdots$};
\node[sep] at (\xstart,-2*\ysep) {$\vdots$};
\node[sep] at (\xstart+0.5*\xgap,-2*\ysep) {};
\node[sep] at (\xstart+1*\xgap,-2*\ysep) {$\ddots$};
\node[sep] at (\xstart+1.5*\xgap,-2*\ysep) {};
\node[sep] at (\xstart+2*\xgap,-2*\ysep) {$\vdots$};
\node[sep] at (\xstart+2.5*\xgap,-2*\ysep) {};
\node[sep] at (\xstart+3*\xgap,-2*\ysep) {$\vdots$};
\node[sep] at (\xstart+3.5*\xgap,-2*\ysep) {};
\node[sep] at (\xstart+4*\xgap,-2*\ysep) {$\vdots$};

\rowbase{N}{n}{-3*\ysep}{$c_1$}{candone}{\scalebox{0.8}{$c_{n-1}$}}{candnmone}{$c_n$}{candn}

\begin{scope}[on background layer]
    \foreach \n in {r1a,r1b,r1c}
        \node[appr, fit=(\n)] {};
    \foreach \n in {r2a,r2b,r2c}
        \node[appr, fit=(\n)] {};
    \foreach \n in {rNa,rNb,rNc}
        \node[appr, fit=(\n)] {};
\end{scope}

\rowbase{1}{1}{0}{$c_2$}{candtwo}{$c_n$}{candn}{$c_1$}{candone}
\rowbase{2}{2}{-\ysep}{$c_3$}{candthree}{$c_1$}{candone}{$c_2$}{candtwo}
\rowbase{N}{n}{-3*\ysep}{$c_1$}{candone}{\scalebox{0.8}{$c_{n-1}$}}{candnmone}{$c_n$}{candn}

\end{tikzpicture}
\caption{Compare with \citet{halpern2025ApprovalVC}. The proportional veto core only selects the unanimously disapproved $d$ in this example, while all other candidates have $n-1$ of $n$ possible approvals. $\omega$-\dvc{} and AWVC fix this issue by also blocking $d$ and returning the empty set. Indeed, since this profile contains $m>n$ candidates, any rule with proportional veto power views each individual voter as a cohesive group and may allow her to block at least one candidate.
If a minority protection degree less than proportional is desirable for the application at hand, another normatively acceptable societal consent set may be $\{c_1,\dots,c_n\}$.}
\label{fig:PVC-Bad}
\end{figure}
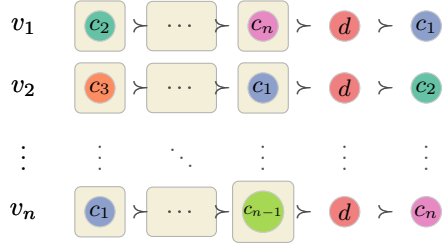

\paragraph{Combined Preference and Approval Input.}
Few prior works assume voters report both ordinal or cardinal preferences \emph{and} 
approvals. A notable exception is \citet{Brams2009}, who consider ordinal preferences 
together with approvals over finite candidate sets; however, different from our work, they propose rules selecting 
the best candidates rather than a societal consent set. Other work in this space 
outputs both a ranking and an approval set, so its output format is a superset of ours; 
however, such functions are constrained by the underlying ranking, with the approval 
set typically required to be an upper set thereof. For example, \citet{dong2021preference} 
axiomatize a distance measure between preference--approval ballots and use it to 
induce a ranking, and \citet{KrugerArrovian} establish an Arrovian impossibility in 
this setting that crucially relies on the underlying ranking. Neither of these two works focuses on aspects of minority protection or similar.

\paragraph{Deliberative approaches to collective consent.} In this paper, we have adopted a social-choice theoretic perspective on collective consent, according to which collective consent is determined by the individual consent and individual preferences of the members of society. Kyi et al.~\citeyearpar{Kyi2026} pursue a different approach to collective consent through deliberation with their proposal for \textit{collective consent assemblies}. Although we have focused on aggregation rather than deliberation, we see these approaches as complementary rather than competitive. As Kyi et al. note, ``Sometimes, participants in a consent assembly may disagree on their decisions\dots. In the Deliberation phase, it is important for facilitators to encourage negotiation and discussion between consent assembly participants
to ensure they reach a conclusion that most are satisfied with.'' If negotiation and discussion do not result in unanimous opinion, then aggregation methods of the kind studied in this paper may be useful to gauge collective consent after deliberation when disagreement persists.

\section{A Collection Of Proofs
}\label{app:TheoryProofs}

Given $S\subseteq N$, $\mathcal I \mid_S$ denotes the \emph{partial instance} of $\mathcal I$ restricted to $S$. Formally, $\mathcal I \mid_S = (u\mid_S, \tau\mid_S)$, where $u\mid_S = (u_i)_{i\in S}$ and $\tau\mid_S = (\tau_i)_{i\in S}$ denote the \emph{partial utility- and threshold profiles}, respectively, restricted to $S$.

\subsection{An Axiomatic Characterization of Threshold Concepts}\label{ssec:ThresholdChar}

\ThresholdCharacterization*
\begin{proof}
    Let $\mathcal I_0,\dots, \mathcal I_n$ be instances in which all voters report the same non-degenerated utility function $u^i$, with the first $x$ voters in $\mathcal I_x$ approving of $C$ and the remaining $n-x$ voters approving no candidate.
    First, assume that on one of these instances, $f(\mathcal I_x)\neq \emptyset$. Take the minimal such $x^*$, and set the threshold to $\alpha = \frac{x^*}{n}$. We claim that on all profiles, $f(\mathcal I) = \{c\in C \mid \score{c}{\mathcal I} \ge \alpha n = x^*\}$.

    \begin{itemize}
        \item \textbf{$c$ is elected if it meets the threshold.} Fix any $c\in f(\mathcal I_{x^*})$. We show that $c\in f(\mathcal I)$ iff $\score{c}{\mathcal I} \ge \alpha n$ for all instances $\mathcal I$.
    Firstly, consider the modified instance $\mathcal I'$ obtained from $\mathcal I$ such that some voters report an empty approval set with $\score{c}{\mathcal I'} = x^*$  
    We may permute the voter labels to arrive at the instance $\mathcal I''$ such that the first $x^*$ voters are the ones approving $c$.
    By anonymity, it holds that $c\in f(\mathcal I'') \Longleftrightarrow c\in f(\mathcal I')$.
    Further, note that $\mathcal I''  \equiv_{\{c\}} \mathcal I_{x^*}$. Therefore, by independence of other candidates, $c\in f(\mathcal I_{x^*})\Longrightarrow c\in f(\mathcal I'')\Longrightarrow c\in f(\mathcal I')$.
    Further, note that we can obtain an instance $I'''$ from $I'$ be letting all voters $i$ who changed to the empty approval set now set their threshold precisely to $\tau_i = u_i(c)$. By approval monotonicity, $c\in f(\mathcal I')\Longrightarrow c\in f(\mathcal I''')$. Finally, since $\mathcal I'''  \equiv_{\{c\}} \mathcal I$, we have that $c\in f(\mathcal I)$. Therefore, we have proven that $c$ is chosen on all instances where it meets the approval threshold $\alpha$. 
    \item \textbf{$c$ is not elected if it does not meet the threshold.} For this, choose $c$ as before, but take any instance $\mathcal I$ where $\score{c}{\mathcal I}<x^*$. We may permute the voter identities such that the approvers of $c$ are the first few voters in $\mathcal I'$. Anonymity implies $c\in f(\mathcal I) \Longleftrightarrow c\in f(\mathcal I') $
    Together with $\mathcal I'  \equiv_{\{c\}} \mathcal I_{\score{c}{\mathcal I}}$, this implies 
    $c\in f(\mathcal I) \Longleftrightarrow c\in f(\mathcal I_{\score{c}{\mathcal I}}) $. 
    However, by assumption, $\score{c}{\mathcal I} < x^*$, and by definition of $x^*$ $\mathcal I_1,\dots, \mathcal I_n$, we have that $ f(\mathcal I_{\score{c}{\mathcal I}}) =\emptyset$. Therefore, $c\notin f(\mathcal I)$. As desired, we have shown that $c$ is not elected whenever it does not meet the approval threshold $\alpha n$. 
    \item If $f$ outputs the empty set on all instances, we can repeat the second part of our previous argument and obtain that $f(\mathcal I)=\emptyset$ for all instances $\mathcal I$. Therefore, $f$ selects $c$ on any given instance if and only if $c$ meets the approval threshold $\alpha= \infty$.
    \item \textbf{Extending this to other candidates.} Let now a candidate $d\neq c$ be given and any instance $\mathcal I$. Let the set of approvers of $d$ in this instance be denoted via $N(d)\subseteq N$. Then, we define the following instance $\mathcal I'$: here, all voters $i\in N$ report a continuous, non-degenerate utility function such that $u'_i(c) = u'_i(d) = 0.5$. We further set $\tau'_i= 0$ if $i\in N(c)$ and $\tau'_i = 1$ if $i\in N\setminus N(c)$. Then, by what we have previously shown, $c\in f(\mathcal I')$ holds true if and only if $\score{c}{\mathcal I'}\ge \alpha n$. Further, by equality, we have that $d \in f(\mathcal I')$ if and only if $c\in f(\mathcal I')$. Since, by design, $\score{d}{\mathcal I} = \score{c}{\mathcal I'}$, we obtain in total that $d \in f(\mathcal I)$ if and only if $\score{d}{\mathcal I}\ge \alpha n$, as desired. \qedhere
    \end{itemize}
\end{proof}

\subsection{Witness-Based Characterizations}\label{ssec:WitnessCharDVC}
We first state the characterization in the standard finite setting, as this is more intuitive and allows for a step-by-step explanation of why each candidate was blocked. 

To split upper contour monotonicity into  axioms which modify the preferences step-by-step, consider the two following ideas. First, if we justify the blocking of a candidate with a set of voters $S$ for which this candidate is ``bad'' and one of these voters moves $c$ down in their ranking, then $c$ becomes even worse. Therefore, $S$ should still be able to block this candidate. Formally, we say that
$(f,w)$ satisfies  \emph{candidate montonicity}, if $S\in w(c,(u, \tau))$ implies that $S\in w(c,(u', \tau))$, where $u'$ is derived from $u$ by letting a voter in $S$ move $c$ lower in her rankings but keep all other pairwise comparisons and all approvals intact. Note that this operation is only allowed if it does not interfere with the approval threshold.

Second, if a set of voters $S$ blocks a candidate $c$, and swaps the ranking of two candidates $d, d'$ which have the same approval status and are both ranked strictly above $c$, or both ranked strictly below $c$, then their relative quality to $c$ has not changed. $c$ should therefore remain to be vetoed. The same reasoning applies if $d_1,d_2,\dots$ form an equivalence class w.r.t. voter $i$'s preferences, and we move $d_1$ into its own equivalence class adjacent to $d_2,\dots$, or vice versa pick some $d_1\neq c$ which forms its own equivalence class and add it to some adjacent $d_2,\dots \neq c$. Formally, $(f,w)$ satisfies \emph{weak independence of other candidates} if $S\in w(c,(u,\tau))$ implies $S\in w(c,(u',\tau))$, where $u'$ is obtained from $u$ by letting one voter in $S$ change the position of a candidate $d\neq c$ (and only of that candidate) in the ranking such that $d\succ_i' c$ if originally $d\succ_i c$, or $c\succ_i' d$ if originally $c\succ_i d$, respectively.

Polarized minority protection can be parameterized: inbstead of blocking their worst share of $\frac sn$ candidates, we could allow the minority of size $s$ to block their worst $\omega \frac sn$ candidates for some $\omega\ge 0$. We call the resulting axiom $\omega$-polarized minority protection.
We are now ready to characterize the parameterized family of $\dvc$-rules. 

\begin{proposition}
    Let $w$ be a witness function for $f$ . If $(f, w)$ satisfies $\omega$-polarized minority protection,
locality, candidate monotonicity, weak independence of other candidates, and approval competition monotonicity, then $f \subseteq \omega\text{-}\dvc$. In particular, for $\omega = 1$, $f\subseteq \dvc$.
\end{proposition}
\begin{proof}
    Let $c$ be blocked by $\omega$-\dvc{} on the instance $\mathcal I = (\succ,\tau)$. Instead of arguing over utility functions, we directly consider the relations that they induce: recall that $\succsim_i$ denotes the relation such that $c\succsim_i d$ iff $u_i(c) \ge u_i(d)$. By definition, there exists some $S\subseteq N$ blocking $c$, i.e., for $s= \lvert S \rvert$ we have $\omega \frac sn > 1-\mu(B_S(c))$. We start from the polarized instance $I_S$, with all voters in $S$ unanimously bottom ranking $C' = C\setminus B_S(c)$ in an arbitrary linear order $\succ'_i$ such that $c = \max_{\succ_i'} C'$. By $\omega$-polarized minority protection, $S\in w(c,\mathcal I_S)$. We now first use weak independence of other candidates to rearrange for each $i\in S$ the set $C'$. First, define the lower contour set of $c$ for voter $i$ to be $L= \{d\in C\setminus\{c\}\mid c\succsim_i d\}$. It holds that $L\subseteq C'$. Therefore, first, if $(C'\setminus\{c\}) \setminus L\neq \emptyset$, we apply weak independence of other candidates to change the ranking $\succ_i'$ below $c$ and arrive in an instance $\mathcal I^1$ that satisfies $c\succ_i^1 C'\setminus L \succ_i^1 L $, with $S\in w(c,\mathcal I^1)$.
    Iteratively, for each $i\in S$, we select the highest-ranked candidate $d\in L$ which is currently not at the right position in $\succsim_{i}^1$ (with the reference point being $\succsim$). We repeatedly apply weak independence of other candidates to put it into the right position, and continue this iteratively until we arrive at an instance $\mathcal I^2$ such that $\succ_i^2$ satisfies $c\succ_i^2 C'\setminus L \succ_i^2 L $ and coincides with $\succ_i$ on $L$, with $w(c,\mathcal I^2)$. 
    Then, we apply monotonicity to move $c$ below $C'\setminus L$ and obtain $\mathcal I^3$ with $C'\setminus L \succ_i^3 c \succsim_i^3 L $ and $S\in w(c,\mathcal I^3)$. Note that $\succsim_i^3$ already coincides with $\succsim_i$ on all orderings within $L\cup \{c\}$ and both rank $L\cup \{c\}$ at the bottom.  Through further applications of independence of other candidates, this time on 
    $C\setminus (L\cup\{c\})$, we obtain $\mathcal I^4$ such that $\succsim_i^4 = \succsim_i$, with $S\in w(c,\mathcal I^4)$. We then use approval competition monotonicity to add the approval thresholds and finally arrive 
    in $I^5$ such that $(\succsim_i, \tau_i) = (\succsim_i^5, \tau_i^5)$, and $S\in w(c,\mathcal I^5)$. 
    
    Repeating this for every $i\in S$ yields an instance $I^*$ which, compared to $\mathcal I$, has identical ordinal and approval preferences for all voters in $S$, and $S\in w(c,\mathcal I^*)$. 
    Finally, by locality, we can rearrange the preferences on $N\setminus S$ as we desired and still obtain $S\in w(c,\mathcal I)$, i.e., $c\notin f(\mathcal I)$.
\end{proof}

Summarized, we can start from a polarized profile with linear opposing preferences between $S$ and $N\setminus S$. For each voter, we only perform single- or two-candidate changes within the upper and lower contour set of the blocked candidate $c$, then move $c$ down, and perform more of these changes. Finally, we add the approval thresholds, and have reconstructed the profile in which we want to demonstrate why $c$ is blocked.
Next, we present the characterization in the general case. 

We now turn our attention to the general case with possibly continuous $C$. Just as in the discrete setting, we prove a generalization of \Cref{prop:DVC_Witness_Characterization_Continuous} which holds for all $\omega \ge 0$. 
\begin{theorem}
    Let $w$ be a witness function for $f$. If $(f, w)$ satisfies $\omega$-polarized minority protection, locality, upper contour monotonicity, and approval competition monotonicity, then $f \subseteq \omega\text{-}\dvc$.
\end{theorem}
\begin{proof} 
    It is possible to prove this statement for $\omega$-\dvc{} with the same argument but using $\omega$-polarized minority protection. Let $c$ be blocked by $\omega$-\dvc{} on the instance $\mathcal I = (u,\tau)$. By definition, there exists some $S\subseteq N$ blocking $c$, i.e., for $s= \lvert S \rvert$ we have $\omega \frac sn > 1-\mu(B_S(c))$. 
    Therefore, define $u_i'$ via the Euclidean distance $u_i'(d) = dist(d, C\setminus B_S(c))$. 
    This utility function satisfies that for all $d\in B_S(c)$, $e\in C'\coloneqq C\setminus B_S(c)$, we have $u'_i(d) > u'_i(e)$. 
    We start from the polarized instance $\mathcal I_S$, with all voters in $S$ unanimously reporting $u_i'$ and empty approval ballots, and all voters in $N\setminus S$ reporting $-u_i'$ (and arbitrary approval ballots). 
    By $\omega$-polarized minority protection, as $c\in C'$, $S\in w(c,\mathcal I_S)$. As $ B_S(c)= \bigcap_{j\in S} B_j(c) \subseteq B_i(c)$, we let voter $i\in S$ change her utility function from $u_i'$ to $u_i$ to obtain an instance $\mathcal I^1$. Upper contour monotonicity, implies that $S\in w(c,\mathcal I^1)$. We then use approval competition monotonicity to add the approval thresholds and arrive 
    in $I^2$ such that $(u_i, \tau_i) = (u_i^2, \tau_i^2)$, and $S\in w(c,\mathcal I^2)$. 
    
    Repeating this for every $i\in S$ yields an instance $\mathcal I^*$ which, compared to $\mathcal I$, has identical utility and approval preferences for all voters in $S$, and $S\in w(c,\mathcal I^*)$. 
    Finally, by locality, we can replace utility functions and thresholds on $\mathcal I^*\mid_{N\setminus S}$ with $\mathcal I\mid_{N\setminus S}$ as we desire and still obtain $S\in w(c,\mathcal I)$, i.e., $c\notin f(\mathcal I)$.
\end{proof}

\begin{remark}
    Note that these results entail a witness-based characterization of Moulin's PVC in the standard finite setting and in the infinite setting. For this, remove approvals from the input format and remove the axiom weak independence of approvals.
\end{remark}

\subsection{Disapproval Smith Set}
\dss*
\begin{proof}

If $c,d$ are unanimously approved, then by definition $n(c,d) = 0 = n(d,c)$. If $e$ is not unanimously approved, then by definition $n(c,e) > 0 > n(e,c)$. Therefore, $\tc(\mathcal I) = \{ c \in C\mid \score{c}{\mathcal I} = n\}$.

\end{proof}

\subsection{Proofs for Solution Concepts and Blocking Principles}\label{sapp:RulesAndPrinciples}
To concatenate $\tc$, we must be able to choose from subsets $D\subseteq C$.
For this, we can define a restriction $R^*(D)$ via $c R^*(D) d$ iff
$c,d \in D$ and there exists $c= d_0 R d_1 R\dots R d_k = d$ such that $d_j\in D$ for all $j\le k$. Then, $\tc(\mathcal I, D) = \{c\in D\mid c R^* d \forall d\in D\}$.

We state all results related to blocking principles, then prove them together.
\RulesAndPrinciples*
\omegadvc*
\AWVCThm*
\ConcatenateBlockingGuarantees*
Note that we can generally inherit veto-power and approval-threshold lower bounds of $g$ while satisfying exclusive unanimity by concatenating two rules $f\circ g$, under conditions:
\begin{enumerate}
    \item $f$ must be well-defined for the subsets $g(\mathcal I) \subseteq C$ that $g$ may return.
    \item $f$ must satisfy exclusive unanimity
    \item When $\mathcal I$ contains unanimously approved candidates and $g(\mathcal I)\neq \emptyset$, then $g(\mathcal I)$ must contain at least one of the unanimously approved candidates.
\end{enumerate}

\begin{proof}
    \textbf{Veto power}:
    For proving that the concept satisfies at least the claimed level of veto power, let $T\subseteq C$ be well-behaved. Then, it is ranked at the bottom w.r.t. some utility function $u^*$, i.e., $u^*(c) < u^*(d)$  for all $c\in T$ and $d\in C\setminus T$. Let $S\subseteq N$ be of size $s$ such that $u^i = u^*$ for all $i\in S$.
    
    \begin{itemize}
        \item Starting with $\omega$-\dvc{}, let $\mu(T)< \omega \frac sn$. It is clear that, by definition, $S$ blocks every candidate $c\in T$ as $B_S(c) \supseteq C\setminus T$.
        Vice versa, let $\mu(T)> \omega \frac sn$. If all $i\in N\setminus S$ submit thresholds such that $A_i = C$, then the only coalitions eligible for blocking under $\omega$-\dvc{} are $S'\subseteq S$. 
        We claim that some $c\in T$ cannot be blocked by any subset of $S$: For this purpose, start with $C^0 = C$, enumerate the voters $i\in S$ and let them, one by one, pick a set of their lowest-utility candidates $L_i$ of measure $\omega \frac 1n$ to remove it from $C^{i-1}$. After all $\lvert S \rvert$ voters are done, a set $T'$ of measure $\mu(T') = \omega \frac sn$ was removed from $C$.\footnote{If $C$ is finite, we can adapt the argument by allowing voters to remove fractions of candidates, or refer to \citet{Moulin1982}.} By additivity of measure, there exists some $c\in T\setminus T'$. We claim that no subset of $S$ can block $c$: indeed, each voter $i$ ranks $L_i$ below $c$. Since all sets $L_i$ are disjoint by construction, any $S'\subseteq S$ satisfies that $L_{S'}(c) \supseteq \bigcup_{i\in S'} L_i$, and therefore $\mu (L_{S'}(c)) \ge \omega \frac{\lvert S'\rvert}{n}$. This concludes the proof.
        \item For AWVC, let $\mu(T)< \frac s{n-s}$. For each $c\in T$, we have $B_S(C) \supseteq C\setminus T$.
        As long as all voters $i\in S$ disapprove all of $T$, we have that $\score{c}{\mathcal I}\le n-s$, therefore $1-\mu(B_S(c)) = \mu (C\setminus B_S(c))\le \mu(T)< \frac s{n-s} \le \frac s{\score{c}{\mathcal I}}$, therefore $S$ can block $c$. Vice versa, if $\mu(T)> \frac s{n-s}$, we once more let $N\setminus S$ approve of all candidates. By an analogous argument as for $\dvc$, some $c\in T$ is not blocked.
        \item  For $\tc$, let $\mu(T)> 0$, $ s< \frac n2$, and consider any candidate $c\in T$. Let all other voters $i \in N\setminus S$ rank $c$ at the top of their preferences, i.e., take the Euclidean distance $u_i(d) = -dist(c,d)$ for all $d\in C$. Then, $n(c,d) - n(d,c)\ge 0$ for all $d\in C$, therefore $c\in \tc(\mathcal I)$, i.e., $S$ was not successful in blocking all of $T$.
        \item Clearly, $\tc \circ${AWVC} has at least the same veto power as AWVC. Vice versa, let $S$ now report any partial instance $(u_S,\tau_S)$ with the goal of blocking all of $T$, where $\mu(T) >\frac s{n-s}$. We know that there exists some $c\in T$ not blocked by any subset of $S$.  Let all voters $i \in N\setminus S$ rank $c$ at the top of their preferences, e.g., when $C$ is infinite, take the Euclidean distance $u_i(d) = -dist(c,d)$ for all $d\in C$ and let them set $\tau_i = u_i(c)$. Then, $n(c,d) - n(d,c)> 0$ for all $d\in C\setminus\{c\}$. Therefore, $c$ is not blocked by AWVC, and afterwards the only candidate elected by $\tc$. 
        \item For $f_\alpha$, let $S$ be of size $\le n-\alpha n$. Let $c\in T$ be any candidate. Then, if all voters in $N\setminus S$ approve $c$, $\score{c}{\mathcal I} \ge \alpha n$. Therefore, $c\in f_\alpha(\mathcal I)$.
        \item For $f\circ g$, note that to block $T$, the voters can report the same partial instance $\mathcal I\mid_S$ that they would report under $g$. Since $g(\mathcal I)\cap T=\emptyset$ for all adversarial responses, the same holds true for $f\circ g (\mathcal I)\cap T \subseteq g (\mathcal I)\cap T=\emptyset$.
    \end{itemize}

    \textbf{Approval Threshold}: Let all voters in $S$ disapprove of $c$
    \begin{itemize}
        \item $\omega$-\dvc{}: Note that for $\omega\le 1$, the solution concept is non-empty. Therefore, simply let all voters report the empty approval ballot, and the solution concept therefore elects a candidate with approval score $0$. For $\omega >1$, let any coalition $S$ of size $s> \frac n\omega$ satisfies $\omega \frac sn > 1 \ge 1- \mu(B_S(c))$. Therefore, they can block $c$, hence $c$ needs at least $n- \frac n\omega$ supporters for societal consent. To see that candidates with a smaller approval score may get chosen, simply consider any subset of voters $S$ with $s\le \frac n\omega$, choose any nontrivial utility function $u^*$, and let $c$ be a maximizer of it. Let all $i\in N$ submit $u_i = u^*$, with the voters in $N\setminus S$ approving $c$ and the voters in $S$ disapproving of $c$. Then, $\omega \frac sn \le 1 = 1- \mu(B_{S'}(c))$ for all coalitions $S'\subseteq N$. Therefore, $c$ is not blocked, despite having an approval score of $\lceil n - \frac n\omega \rceil$.
        \item For AWVC, similarly, if $s> \frac n2$ voters $S$ disapprove of $c$, then the approval score of $c$ is strictly smaller than $\frac n2$. Therefore, $\frac s{\score{c}{I}} > \frac {\frac n2}{\frac n2}  = 1 \ge 1- \mu(B_S(c))$, and $S$ blocks $c$. With the same construction as for $\omega$-\dvc{}, one can show that candidates with a higher approval score than that can indeed be chosen.
        \item For $\tc \circ${AWVC}, we have that it clearly has at least the same approval threshold as AWVC. Further, in the same instance as constructed for AWVC, we know that $c$ is not blocked by AWVC. Further, since for all voters $c$ is a top-ranked candidate (or utility maximizer), it is then selected among the remaining candidates by $\tc$.
        \item For $\tc$ 
        consider $c$ maximizing some $u^*$ and $u_i = u^*$ for all $i\in N$ with empty approval sets. Then, the transitive majority relation is induced by $u^*$ and $c$ is maximal. 
        Hence, \tc{} selects $c$ as winner.
        \item For $f_\alpha$, clearly it satisfies precisely $\alpha$-approval threshold.
        \item For $f \circ g$, note that any candidate below approval threshold satisfies $c\notin g(\mathcal I)$. Since $f \circ g(\mathcal I) \subseteq g (\mathcal I)$, we have $c\notin f \circ g(\mathcal I)$.
    \end{itemize}

\textbf{Exclusive Unanimity}
    \begin{itemize}
        \item $\omega$-\dvc{}: if all voters submit the same normalized $u^*$ with $n-1$ voters submitting $\tau_i = 0$ and one voter submitting $\tau_i = 1$, then the maximizers of $u^*$ are unanimously approved. However, for any candidate $d$ with $\mu(B_S(d)) < \frac 1n $ not a maximizer of $u$ and $n$ large enough, we have that $\omega \frac 1n < 0.5 < 1- \frac 1n $, therefore $d$ is not blocked despite not being unanimously approved.\footnote{Note that for discrete $C$, we first choose $n$ and then $m$ to construct this counterexample.}
        \item AWVC: we can use the same construction as for $\omega$-\dvc{}, as the number of approvals is $n-1$, and $\frac{1}{n-1}$ converges to zero just like $\frac 1n$. Therefore, for large enough $n$, we obtain a profile in which a candidate with $n-1$ approvals is chosen despite the existence of a unanimously approved candidate.
        \item $\tc\circ${AWVC}: When there are unanimously approved candidates, these are not blocked by AWVC as only disapprovers are allowed to block by definition. Then, clearly, among all candidates not blocked by AWVC, $\tc$ selects precisely the unanimously approved ones as they have strictly positive majority margins against all candidates disapproved by some voter, but tie against each other. Therefore, $\tc\circ${AWVC} even is unanimous approval consistent.
        \item For $\tc$ 
        it is clear that unanimously approved candidates are the only maximal elements of the majority relation. Therefore, \tc{} even is unanimous approval consistent. 
        \item For $f_\alpha$, consider any profile in which one candidate is unanimously approved, and another one has approval score $\lceil \alpha n \rceil < n$.
        \item For $f\circ g$, note that if $\mathcal I$ contains a unanimously approved candidate, $g$ returns at least one of these. Since $f$ is defined for the input of $g$ and satisfies exclusive unanimity, it returns only a subset of the unanimously approved candidates $g$ chooses.
    \end{itemize}

    Finally, to prove that $\dvc\supseteq${AWVC}$\supseteq 2\text{-}\dvc$, we first note that candidates with approval score strictly smaller than $\frac n2$ are not selected by AWVC and 2-$\dvc$ (due to them satisfying majority threshold). For candidates $c$ with approval score at least $\frac n2$, we further have $\frac 1n \le \frac 1{\score{c}{\mathcal I}} \le \frac 2n$, from which the claim follows as the definitions of the three rules are identical up to these fractions.
\end{proof}

\section{Proofs of \Cref{ssec:veto_budgeting_games}}

We prove \Cref{thm:GT_Char_PVC} in a more general form:
\begin{theorem}[Veto-budget characterization of $\omega$-\dvc]
Let $\omega \ge 1$.
For every instance $\mathcal I=(u,\tau)$ and candidate $c\in C$, the following are equivalent:
\begin{enumerate}
    \item $c\notin \omega-\dvc(\mathcal I)$.
    \item In the veto-budgeting game with budget $\omega \cdot 1/n$ per voter, there exists $\varepsilon>0$ such that every allocation profile $\lambda$ with $c\in W(\lambda)$ admits a disapproval-restricted $\varepsilon$-deviation against $c$.
\end{enumerate}

\end{theorem}

    \begin{proof}
    We write $C\setminus B_S(c) = L_S(c)$ as the set of all candidates ranked weakly lower than $c$ by at least one voter in $S$. For this proof, let $N_c$ denote the set of voters disapproving $c\in C$.
    
        \textbf{``$\Longleftarrow$''}: This direction is simple, and we can deal with the finite and infinite $C$ case simultaneously.
        \begin{itemize}
            \item First, let $p\notin  \omega-\dvc(u,\tau)$. By definition, there exists $S\subseteq N_c$ such that $\mu(L_S(c)) < \omega \cdot \frac {\lvert S\rvert}n$. Indeed, since $ \mu(L_S(c)) < \omega \cdot \frac {\lvert S\rvert}n$ and $C$ is finite or all $u_i$ are continuous, there exists $\varepsilon>0$ such that $ \mu(L^\varepsilon_S(c)) < \omega \cdot \frac {\lvert S\rvert}n$ for the set $L^\varepsilon_S(c) = \{d\mid  \exists i\in S: u_i(d) \le u_i(c) +\varepsilon \}$, and we use precisely this $\varepsilon$.
            \item Now, consider any allocation $\lambda$ with $p\in W(\lambda)$. Our goal is to show that $\lambda$ admits a disapproval-restricted $\varepsilon$-deviation.
            Since coalition $S$ can spend strictly less but arbitrarily close to $\omega \cdot \frac {\lvert S\rvert}n$, consider any deviation $\lambda'_S$ such that $\lambda'_S(d) = 1$ for all $d\in L^\varepsilon_S(c)$. Then, no matter the specification of $\lambda'_{N\setminus S}$, we have that still $C\setminus W(\lambda') \supseteq L'$. Finally, $\inf_{d\in W(\lambda')} u_i(d) \ge u_i(c) +\varepsilon \ge \inf_{d\in W(\lambda )}u_i(d) +\varepsilon $ for all $i\in S$.
        \end{itemize}

        \textbf{``$\Longrightarrow$''} 
        First, we state some notation and observations about the veto budgeting game: denote for each $S\subseteq N$ by $\lambda_S(c) \coloneqq \sum_{i\in S} \lambda_i(c)$ the sum over all 
        the payoff of each voter depends only on $W(\lambda)$, which in turn only depends on for which candidates $d\in C$ we have $\lambda_N(d) \ge 1$. Since it does not matter which voter spends her budget on which candidate, the only thing that matters for a set $W$ to be the outcome of some allocation profile $\lambda$ whether $1-\mu(W)< 1$.
        The same holds for the payoff of a deviation $W(\lambda_S)$, i.e., whether a coalition $S$ can deviate to a set $W$ depends only on whether $1-\mu(S)< \frac sn$.

        We now start with the easier, finite $C$ case. 
        \begin{itemize}
            \item Let  $c\in \dvc(u,\tau)$. We can let each voter spend arbitrarily much budget as long as it is strictly less than $\frac \omega n$. Since $\omega \ge 1$, this means that $N$ can block all candidates but $1$. Therefore, there exists an allocation $\lambda$ such that $W(\lambda) = \{c\}$. We now prove that this allocation does not allow for disapproval-restricted deviations. Clearly, to improve from this deviation means that $c$ has to become blocked. Therefore, disapproval-restricted deviations can only involve disapprovers of $c$. Let $S\subseteq N_c$ be given with potential deviation $\lambda'_S$. For all voters in $S$ to strictly prefer $W(\lambda'_S)$ to $W(\lambda)$, a necessary condition is that $L_i(c)\subseteq C\setminus W(\lambda'_S)$ for all $i\in S$. Therefore, for all $d\in \bigcup_{i\in S} L_i(c)$, we have that $\lambda'_S(d) \ge 1$. However, since $d$ is contained in $\omega$-\dvc, we have that $\mu (\bigcup_{i\in S} L_i(c)) = \mu( L_S(c) ) \ge \omega \frac {\lvert S\rvert }n$, and therefore it is impossible for $S$ to purchase all of $L_S(c)$. 
            Therefore, there is no coalition $S\subseteq N_c$ for which a beneficial deviation $\lambda'_S$ exists.
        \end{itemize}
        Next, we deal with the continuous case of $C$ being a polytope in a similar fashion:
        \begin{itemize}
            \item Let  $c\in \dvc(u,\tau)$. We have to show that there exists no $\varepsilon$ such that all allocations sparing $c$ allow for disapproval-restricted $\varepsilon$-deviations. Therefore, let any $\varepsilon>0$ be given. We construct a corresponding, stable $\lambda$ as follows:
            we can let each voter spend arbitrarily much budget as long as, combined, it is strictly less than $\frac \omega n$. Since $\omega \ge 1$, this means all candidates but an arbitrarily small positive measure can be blocked. Therefore, for each small enough $\delta >0$, there exists $\lambda$ such that $W(\lambda) = O(c,\delta)$, where $O(c,\delta)$ is the open ball with radius $\delta$ within $C$. 
            Since each $u_i$ is continuous, we can further choose $\delta>0$ such that for all $d\in O(c,\delta)$, we have $u_i(d) \in  (u_i(c) - \frac{\varepsilon}{2}, u_i(c) + \frac{\varepsilon}{2})$ for all $i\in N$.
            We now prove that this allocation does not allow for disapproval-restricted deviations. Clearly, by choice of $O(c,\delta)$ any deviation $\lambda'_S$ with $c\in W(\lambda'_S)$ does not yield an $\varepsilon$-improvement compared to $W(\lambda)$ for any $i\in N$.
            Therefore, disapproval-restricted deviations can only involve disapprovers of $c$. Let $S\subseteq N_c$ be given with potential deviation $\lambda'_S$. For all voters in $S$ to strictly prefer $W(\lambda'_S)$ to $W(\lambda)$, a necessary condition is that $L_i(c)\subseteq C\setminus W(\lambda'_S)$ for all $i\in S$. Therefore, for all $d\in \bigcup_{i\in S} L_i(c)$, we have that $\lambda'_S(d) \ge 1$. However, since $d$ is contained in $\omega$-\dvc, we have that $\mu (\bigcup_{i\in S} L_i(c)) = \mu( L_S(c) ) \ge \omega \frac {\lvert S\rvert }n$, and therefore it is impossible for $S$ to purchase all of $L_S(c)$. 
            Therefore, there is no coalition $S\subseteq N_c$ for which a beneficial deviation $\lambda'_S$ exists.\qedhere
        \end{itemize}

\end{proof}

We note that for real-world explanation systems, an important next step is to understand how we can find distributions in polynomial time that are not just trivially stable because we veto almost every candidate but $c$. We impose two desiderata in the finite candidate set case:
\begin{itemize}
    \item No disapproval-restricted deviation should be possible
    \item The voters should play '`good'' distributions such that no group of voters can swap payments such that all afterwards pay for candidates they rank closer to the bottom than before.
\end{itemize}
By adapting a matching approach of \citet{kizilkaya2025k}, we should be able to satisfy these in the case that $\omega\ge 1$. Understanding this in general, however, would also allow us to extend the characterization to $\omega<1$. Interesting for this may be the work by \citet{kondratev2024veto}, who in a rather involved fashion prove demonstrates that all proportional veto core outcomes can be obtained by some picking order of the voters blocking appropriate fractions of their currently last preferred candidate. Translating his proof from the voting by veto framework using tokens and greatest common divisor arithmetic to the veto budgeting framework may yield the desired result.

\textbf{Further related work regarding the veto budgeting games} 
We remark that veto budgeting games and the corresponding characterization are inspired by three works. Firstly, in the standard, discrete social choice setting, \citet{Moulin1981} introduces the proportional veto core, defining that a candidate $c$ is not stable iff there exists a coalition of voters which has sufficient veto power to force the outcome to be strictly preferred to $c$ for all voters in the coalition. However, this description via veto power is neither tied to a specific action space, nor does it specify in which way the outcomes are determined once voters choose their actions.
In \citet{Moulin1982}, he does consider concrete game forms, but restricts his analysis of strong equilibria to single-valued social choice functions, including resolute solution concepts refining PVC, but not PVC itself.
Secondly, on a conceptual level, \citet{kizilkaya2025k} provide a perfect matching-based characterization of PVC: to determine whether $c$ is contained in PVC, voters and candidates are modeled as weighted nodes on a bipartite graph and an edge exists between voter and candidate iff the voter weakly prefers $c$ to that candidate.
Thirdly, also conceptually, works such as \citet{BudgetingGames} consider budgeting games, in which the voters use their budget to purchase candidates into a committee. The authors consider several kinds of ``stable'' budget distributions to characterize proportionality notions and solution concepts in committee voting.
On a technical level, our contribution is to generalize the existing characterizations on several fronts. Firstly, we specify a rigorous game structure. in which we allow the outcomes to be set-valued instead of resolute, and endow the voters with pessimist-preferences. Secondly, we also prove the characterization for uncountable, compact sets of candidates. Thirdly, our characterization covers all $\omega$-scaled veto powers, while the original frameworks only covers non-empty choice sets and $\omega$-\dvc{} may be empty for all $\omega>1$. 
On a conceptual level, the contribution of this characterization allows us to view $\omega$-\dvc{} as a process in which the voters obtain the same, fair share of ``individual veto budget.'' To evaluate why a candidate was blocked or receives societal consent, one can then reason over how the individual voters spend their budget and whether this choice was rational. This is in line with work of \citet{boehmer2026explanation}, who use price systems to evaluate how proportional committees are.

\subsection{Additional Details on the $\omega$-DVC}

\monotondis*
\begin{proof}
The proof is immediate: Suppose $c\notin AWVC(I)$ so there exists a non-empty set of voters $S$ such that
\[\frac{|S|}{s_c(\mathcal{I})}> 1-\mu(B_S(c))\]. Now suppose that voter $n+1$ with utility and threshold $(u_{n+1},\tau_{n+1})$ joins the election and disapproves $c$, then $\frac{|S|}{s_c(\mathcal{I})}> 1-\mu(B_S(c))$ still holds for coalition $S$, thus $S$ blocks $c$ in the instance $\mathcal{I}\cup\{(u_{n+1},\tau_{n+1})\}$, i.e. also $c\notin AWVC(\mathcal{I}\cup\{(u_{n+1},\tau_{n+1})\})$.
\end{proof}
We conclude this section by noting that \Cref{fig:dvc-example} constitutes a profile in which a previously blocked candidate receives societal consent due to a disapprover joining the election. In the depicted profile, only $a$ is elected before voter $v_5$ joins the election. Afterwards, $a$ becomes blocked, and only $c$ receives societal consent. This is questionable, because voter $v_5$ disapproves $c$ and also ranks $a$ above $c$. Note that this type of monotonicity, while intuitive, is not implied by minority protection (or the other two blocking principles): one may argue that indeed, through a voter joining the election, the previous blockers of $c$ are now a smaller fraction of the electorate and not cohesive enough to remain blocking for $c$. Further, $a$ should be blocked in this profile if the concept satisfies majority threshold.
However, one may also argue that societal consent should be monotonous in disapprovers joining the election, which is appealing in its own right. Then, it is crucial that one chooses candidate-wise scaled ``voter budgets,'' as does AWVC. 
\begin{figure}[h]
    \vspace{0pt}
        \centering
\begin{tikzpicture}[baseline=(current bounding box.north),
    x=0.58cm,
    y=0.62cm,
    voter/.style={font=\small},
    cand/.style={
        circle,
        minimum size=4mm,
        inner sep=0pt,
        font=\footnotesize\itshape,
        text=black,
        draw=black!20,
        line width=0.25pt
    },
    appr/.style={
        rounded corners=0.8mm,
        draw=black!30,
        fill=approve,
        minimum width=5.6mm,
        minimum height=5.6mm
    },
    sep/.style={font=\scriptsize, text=black!70},
    aset/.style={font=\scriptsize},
]

\def\xvoter{0}
\def\xstart{1.45}
\def\xgap{1.55}
\def\ysep{1.1}

\newcommand{\placecand}[4]{%
    \node[cand,#4] (#1) at (#2,#3) {};
}

\newcommand{\candAstyle}[3]{\node[cand,fill=candA] (#1) at (#2,#3) {$a$};}
\newcommand{\candCstyle}[3]{\node[cand,fill=candC] (#1) at (#2,#3) {$c$};}
\newcommand{\candXstyle}[3]{\node[cand,fill=candX] (#1) at (#2,#3) {$x$};}
\newcommand{\candYstyle}[3]{\node[cand,fill=candY] (#1) at (#2,#3) {$y$};}
\newcommand{\candZstyle}[3]{\node[cand,fill=candZ] (#1) at (#2,#3) {$z$};}

\newcommand{\drawcandidate}[4]{%
    \ifx#4a
      \candAstyle{#1}{#2}{#3}%
    \else\ifx#4c
      \candCstyle{#1}{#2}{#3}%
    \else\ifx#4x
      \candXstyle{#1}{#2}{#3}%
    \else\ifx#4y
      \candYstyle{#1}{#2}{#3}%
    \else\ifx#4z
      \candZstyle{#1}{#2}{#3}%
    \fi\fi\fi\fi\fi
}

\newcommand{\rowbase}[8]{%
    \node[voter] at (\xvoter,#2) {$v_{#1}$};

    \drawcandidate{r#1a}{\xstart}{#2}{#3}
    \node[sep] at (\xstart+0.5*\xgap,#2) {$\succ$};

    \drawcandidate{r#1b}{\xstart+1*\xgap}{#2}{#4}
    \node[sep] at (\xstart+1.5*\xgap,#2) {$\succ$};

    \drawcandidate{r#1c}{\xstart+2*\xgap}{#2}{#5}
    \node[sep] at (\xstart+2.5*\xgap,#2) {$\succ$};

    \drawcandidate{r#1d}{\xstart+3*\xgap}{#2}{#6}
    \node[sep] at (\xstart+3.5*\xgap,#2) {$\succ$};

    \drawcandidate{r#1e}{\xstart+4*\xgap}{#2}{#7}

}


\rowbase{1}{0}{c}{a}{y}{x}{z}{\{c\}}
\rowbase{2}{-\ysep}{y}{z}{x}{c}{a}{C}
\rowbase{3}{-2*\ysep}{x}{z}{a}{c}{y}{\varnothing}
\rowbase{4}{-3*\ysep}{x}{y}{a}{z}{c}{C}
\rowbase{5}{-4*\ysep}{y}{a}{c}{x}{z}{\{y\}}

\begin{scope}[on background layer]
    \node[appr, fit=(r1a)] {};
    \foreach \n in {r2a,r2b,r2c,r2d,r2e}
        \node[appr, fit=(\n)] {};
    \foreach \n in {r4a,r4b,r4c,r4d,r4e}
        \node[appr, fit=(\n)] {};
    \node[appr, fit=(r5a)] {};
\end{scope}

\rowbase{1}{0}{c}{a}{y}{x}{z}{\{c\}}
\rowbase{2}{-\ysep}{y}{z}{x}{c}{a}{C}
\rowbase{3}{-2*\ysep}{x}{z}{a}{c}{y}{\varnothing}
\rowbase{4}{-3*\ysep}{x}{y}{a}{z}{c}{C}
\rowbase{5}{-4*\ysep}{y}{a}{c}{x}{z}{\{y\}}


\end{tikzpicture}

\caption{An example profile on which $2$-DVC fails support monotonicity for disapproving voters.}
        \label{fig:dvc-example}
\end{figure}
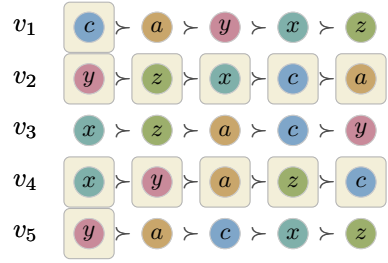%

\subsection{Proofs from Subsection \Cref{sec:quant}}
\omegadvc*
\begin{proof}\sonja{Edit this more.}If $0\leq \omega\leq 1$, then $\alpha =0$, and so $\alpha=\max\left\{0,1-\frac{1}{\omega}\right\}=0.$ \footnote{Here we take $\frac{1}{0}=\infty$.}.
So now, assume $\omega > 1$, and note that then $1-\frac{1}{\omega}> 0$.
If coalition $S$ is of size $s$ with \(s>\frac n\omega\), then \(1-\omega\frac sn<0\). Since \(\mu(B_S(c))\ge 0\), we have $ \mu(B_S(c))>1-\omega\frac sn.$ Therefore by definition of $\omega-\dvc$, if $S$ disapproved $c$, $S$ blocks $c$. So $\alpha \geq 1- \frac{1}{\omega}$.

To see that this bound is tight, consider an instance in which all voters rank \(c\) first and dissappovers $D$ of $c$ are of size $|D|\leq \frac{n}{\omega}$.
Then \(B_S(c)=\emptyset\) for every coalition \(S\). Moreover, every disapproving coalition \(S\subseteq D\) satisfies \(|S|\leq \frac n\omega\) which implies $\mu(B_S(c))=0 \not> 1-\omega\frac{|S|}{n}$, so no disapproving coalition blocks.
The number of approvers of $c$ is least  $n-\frac{n}{\omega}$ and thus $\alpha =1-\frac{1}{\omega}$.
\end{proof}

\section{Additional Details on Minority Protection Measures}\label{app:RelWorkBudgetingGames}


\subsection{Critical Endowment}
We restate the definition of critical omega introduced in the experimental section.
\begin{definition}
    we say that $c$ has \emph{critical endowment} $\omega^*(c) {\ge 0}$ if \[\omega^*(c) = \max \{\omega \ge 0 \mid \text{ $c\in$ $\omega$-\dvc}\}.\]
    We say that a solution concept $f$ \emph{guarantees critical endowment} $\omega$ on instance $\mathcal I$ if  $\omega^*(c) \ge \omega$ for all $c\in f(\mathcal I)$.
\end{definition}


\measurecomp*
\begin{proof}

Fix a candidate \(c\), and let
\[
    D(c)=\{i\in N\mid u_i(c)<\tau_i\}
\]
be the set of voters who disapprove of \(c\). For each voter \(i\), write
\(B_i(c)\) for the set of candidates that voter \(i\) strictly prefers to
\(c\). For \(T\subseteq D(c)\), define
\[
    f_\omega(T)
    =
    \frac{\left|\bigcup_{i\in T} B_i(c)^c\right|}{m}
    -
    \frac{\omega |T|}{n}.
\]
Then \(c\) is blocked under \(\omega\)-\dvc{} if and only if
\(f_\omega(T)<0\) for some non-empty \(T\subseteq D(c)\). Since
\(f_\omega(\emptyset)=0\), the empty coalition cannot witness blocking under
the strict inequality. Thus it suffices to minimize \(f_\omega(T)\) over all
\(T\subseteq D(c)\) and test whether the minimum is negative.

We first identify the values of \(\omega\) that can be critical. For any
non-empty \(T\), the condition \(f_\omega(T)<0\) is equivalent to
\[
    \frac{n\left|\bigcup_{i\in T}B_i(c)^c\right|}{m|T|}
    < \omega .
\]
Since \(\left|\bigcup_{i\in T}B_i(c)^c\right|\in\{0,\dots,m\}\) and
\(|T|\in\{1,\dots,n\}\), every possible breakpoint is of the form
\[
    \frac{ni}{mj},
    \qquad
    i\in\{0,\dots,m\},\quad j\in\{1,\dots,n\}.
\]
Thus there are at most \((m+1)n\) candidate values. In particular, for every
value tested by the algorithm, \(\omega=\frac{ni}{mj}\) has numerator and
denominator polynomially bounded in \(m\) and \(n\).

It remains to give a decision procedure for a fixed tested value of
\(\omega\). Construct a directed graph with source \(s\), sink \(t\), one node
for each voter in \(D(c)\), and one node for each candidate. Add an edge
\(s\to i\) of capacity \(m\omega\) for each \(i\in D(c)\), an edge \(d\to t\)
of capacity \(n\) for each candidate \(d\), and an edge \(i\to d\) of capacity
\(M\) whenever \(d\in B_i(c)^c\), where \(M\) is larger than any cut value
using only the finite-capacity source-voter and candidate-sink edges.

Consider a finite cut and let \(T\subseteq D(c)\) be the voters on the source
side. Because the edges \(i\to d\) have capacity \(M\), every candidate in
\(\bigcup_{i\in T}B_i(c)^c\) must also lie on the source side of any minimum
cut. Hence the cut value corresponding to \(T\) is
\[
    m\omega\bigl(|D(c)|-|T|\bigr)
    +
    n\left|\bigcup_{i\in T}B_i(c)^c\right|.
\]
The term \(m\omega |D(c)|\) is constant, so minimizing this cut value is
equivalent to minimizing
\[
    n\left|\bigcup_{i\in T}B_i(c)^c\right|-m\omega |T|
    =
    mn f_\omega(T).
\]
In particular, the cut corresponding to \(T=\emptyset\) has value
\(m\omega |D(c)|\), and \(c\) is blocked if and only if the minimum cut value is
strictly smaller than this value.

For a tested value \(\omega=\frac{ni}{mj}\), multiplying all capacities by
\(mj\) makes the source-voter capacities \(m\omega\) integral:
$mj\cdot m\omega = mni.$
The candidate-sink capacities become \(mjn\), and the enforcement capacities
become \(mjM\). Since \(i\le m\), \(j\le n\), and \(M\) is chosen
polynomially bounded, all capacities after scaling are polynomially bounded
integers. The graph is directed and has \(O(n+m)\) vertices and \(O(nm)\)
edges, so the algorithm of \citet{ChenKyngLiuPengProbstGutenbergSachdeva2022}
computes the corresponding minimum cut in \((mn)^{1+o(1)}\) time.
\end{proof}

 \subsection{Critical $\epsilon$ and Generalization}
 Recently, \citet{chooi2026findingcommongroundsea} proposed the following measure in the setting without consent.
\begin{definition}[\citet{chooi2026findingcommongroundsea}]
    Let $\mathcal I$ be any instance, $c\in C\setminus \pvc(\mathcal I)$.
    Recall that $B_S(c)= \{d\in C\mid \forall i\in S: \quad d\succ_i c\}$.
    We say that $c$ has \emph{critical epsilon} $\varepsilon$ for $\varepsilon\geq 0$, if this is the smallest non-negative value such that $\mu(B_S(c)) \le 1- \frac{\lvert  S\rvert}{n} + \varepsilon$ for all $S\subseteq N$.
\end{definition}

From the perspective of the normalized veto budgeting game, for a fixed coalition $S$, the corresponding minimal epsilon describes the amount of aggregated budget that the coalition has remaining after blocking $c$.\footnote{To see this, let $W_S(c)= C\setminus B_S(c) = \{d\in C\mid \exists i\in S\colon \, c\succsim_i d\}$.
Then, the inequality above is equivalent to $\frac{\lvert  S\rvert}{n} \le 1- \mu(B_S(c)) + \varepsilon = \mu(W_S(c)) + \varepsilon$.} The critical epsilon then describes the \emph{maximum amount of budget any blocking coalition has remaining after blocking $c$}. 

While this notion already is not the most natural in our setting, it is only defined for blocked candidates. However, it may also be important to quantify the degree of minority protection that candidates contained in $\dvc$ provide: two choice sets $f(\mathcal I),g(\mathcal I)$ may be both contained in $\dvc$, but still guarantee very different levels of minority protection.
A simple adaptation to our setting is the following 

\begin{definition}
    Let $\mathcal I$ be any instance, $c\in C$.
    We say that $c$ has \emph{local excess budget} $\varepsilon^*(c)$, if it is the smallest $\varepsilon\in \mathbb R$ such that $\mu(B_S(c)) \le 1- \frac{\lvert  S\rvert}{n} + \varepsilon$ for all non-empty $S\subseteq N$ with $c\notin A_i$ for all $i\in S$.
\end{definition}
Note that for $c\in \dvc(\mathcal I)$, the local excess budget is negative and indicates the amount of total budget that the coalition closest to vetoing $c$ is lacking. 


\section{Computation}\label{app:finite_C_Computation}
If the candidate set if finite, all of the solution concepts in our work can be computed in polynomial-time. 

\paragraph{Disapproval Veto Core.}
To test whether a candidate \(c\) is blocked under \(\dvc\), we can use a simple adaptation of the approach by \cite{ianovski2023computingproportionalvetocore}: We construct the same flow network for \(c\), but include outgoing voter--candidate edges only from voters \(i\) with \(c\notin A_i\), and only to candidates \(d\) such that \(c\succsim_i d\).  
The resulting max-flow/min-cut computation determines whether there exists a coalition \(S\) of disapprovers such that
$\mu(B_S(c))>1-\frac{|S|}{n}$.
Running this test for every \(c\in C\) and returning precisely the candidates for which no blocking coalition is found computes \(\dvc(\mathcal I)\) in polynomial time.

\paragraph{Approval-Weighted Veto Core.}
The computation of AWVC uses the same flow-based blocking test candidate by candidate as the Disapproval Veto Core, with two modifications.  
First, candidates with \(\score{c}{\mathcal I}=0\) are blocked immediately by convention; second, for candidates with positive approval score, the blocking threshold is changed from \(|S|/n\) to $\frac{|S|}{\score{c}{\mathcal I}}.$
Equivalently, in the Ianovski--Kondratev network, the candidate-specific budget normalization is adjusted from the electorate size \(n\) to the approval score of \(c\) (i.e. it is equivalent to locally assuming that the electorate has size approval score of \(c\) and then applying the approach of \cite{ianovski2023computingproportionalvetocore}, while voter--candidate edges are again included only for voters who disapprove \(c\) and only toward candidates they rank weakly below \(c\).

For more details on the flow construction, we refer to the proof of \Cref{thm:measure}, in which we implicitly show how to compute the DVC and AWVC more efficiently than \cite{ianovski2023computingproportionalvetocore}.
\paragraph{Disapproval Smith Set.}
To compute the disapproval Smith set, construct the directed graph \(G=(C,E)\)
whose vertices are candidates and where \((c,d)\in E\) iff
\(n(c,d)\ge n(d,c)\). Compute the strongly connected components of \(G\)
using a standard depth-first-search algorithm, such as Tarjan's algorithm
\citep{Tarjan1972DFS}, and form the condensation graph by contracting each
strongly connected component to a single node. Since \(R\) is complete, between
any two distinct strongly connected components there is an edge in exactly one
direction; hence the condensation graph is an acyclic tournament and has a
unique source component, which reaches every other component. The disapproval
Smith set is exactly the set of candidates contained in this unique source
component.

\section{Modifying Threshold Based Rules}\label{app:threshold_based_rule_fix}

We now state how we fix the rules once approval ballots are not defined via thresholds any more. Assume utilities normalized to $[0,1]$ for ease of threshold choice.

Generally, if we want to translate arbitrary approval sets into thresholds, we can proceed as follows:
If voters submit $(u^i, A_i)$ such that $A_i$ is a general approval set instead of an upper contour set defined via some $\tau_i$. 
we first define $\tau_i$ as the minimal achievable 
$\tau_i = \inf \{ u_i(c) \mid c\in A_i: \quad \forall d\in C, u_i(d) \ge u_i(c) \Longrightarrow d\in A_i \} $. Here, if this set is empty and the infimum therefore not well-defined, we set $\tau_i = 1$ and mark this voter, as well as all her top-ranked and not approved candidates with a ``cheat'' mark. 
With these thresholds given, we update the utility of all approved candidates $c\in A_i$ with $u_i(c) < \tau_i$ by setting $u^*_i(c) \gets \tau_i$. For all other candidates we have $u^*_i(c)\gets u_i(c)$.
With these modified utilities and created thresholds we are now able to run our rules.

\begin{definition}[Rel Util]
    This rule stays almost exactly the same, i.e.,
    $\ru^*(u,A) = \{c\in C\mid \sum_{i\in N} u^*_i(c) \ge \sum_{i\in N} \tau_i\}$, unless
    a voter is marked for cheats. In that case, we have $\ru^*(u,A) = \{c\in C\mid \sum_{i\in N} u^*_i(c) > \sum_{i\in N} \tau_i\}$.
\end{definition}

\begin{definition}[Minimize Margin of Disapproval]
    Given $(u,A)$ as input, the distance that voter $i$ gives $c$ is defined as $\max(0, \tau_i - u^*_i(c))$.
    We once more calculate the set  $T =\arg \min_{c\in C} \sum_{i\in  N} \text{distance$(i,c)$}$.
    Among $T$, it chooses the candidates with the minimal number of cheats.
\end{definition}

\section{Support Monotonicity for Disapproving Voters}\label{App:SuppMonDisapprovingVoters}
We empirically test how often $2$-\dvc{} may violate support monotonicity for disapproving voters counterfactually. To be precise, we ask the following question: given an instance $\mathcal I$, are there candidates $c\in 2\text{-}\dvc(\mathcal I)$ such that there exists some set of disapprovers $S$ of $c$, such that removing them leads to $c$ becoming blocked? This can be (i) viewed as the group-version of our monotonicity property, or (ii) we can view the removal of the group one by one. For one of these voters, $c$ must go from being chosen to blocked, therefore in that step a single-voter violation of the axiom occurred.
Since it is computationally expensive to recalculate the rule for each $S\subseteq N$, we chose a heuristic approach: given some $c\in 2$-$\dvc{}$, we can slightly modify the max-flow algorithm computing the rule to obtain the set of voters $S$ who are ``closest'' to blocking $c$. We then simply check whether the abstention of all other disapprovers makes $S$ be able to block $c$, as $S$ then benefits from the budget the leaving voters do not spend any more.

The aggregated occurrences for our dataset are as follows.

\begin{center}
\begin{tikzpicture}
\node[draw, rounded corners, inner sep=6pt] {
\begin{tabular}{lrrr}
\toprule
Source &  Elections & With axiom violation & Rate \\
\midrule
kidney &  60 & 56 & 0.9333 \\
moral\_machine &  60 & 9 & 0.1500 \\
openai\_coval &  319 & 31 & 0.0972 \\
politics &  9 & 0 & 0.0000 \\
synthetic &  600 & 176 & 0.2933 \\
\bottomrule
\end{tabular}
};
\end{tikzpicture}
\end{center}

This indicates that violations of support monotonicity for disapproving voters for $2$-\dvc{} may be frequently observed in practice. While less frequent, some violations are already observable when we brute force single-voter abstentions.

\section{Experiments}\label{app:exp}

On a 12-core machine, running all experiments takes around six hours; no GPUs or other specialized computing resources are needed.

\subsection{Dataset 1: The Moral Machine} \label{app:exp:moralmachine}
In the Moral Machine experiment, users interact with a web application, making binary choices (moral dilemmas) involving autonomous vehicles. We refer the reader to \citet{awad2018moral} for a description of the Moral Machine dataset. Following \citet{kim2018computational}, we represent each 
country $i$ as a voter equipped with a learned 24-dimensional weight vector $u^{i}\in \mathbb{R}^{24}$ over moral features. The per-country vectors 
are obtained by fitting a Bradley--Terry model
(maximum likelihood with analytical gradient) to each country's pairwise response data, grouping by \texttt{UserCountry3}. Countries with fewer than $1000$ complete scenarios or whose learned model achieves below $50\%$ train accuracy are excluded, yielding $n=137$ voters. As candidates, we sample $m=100$ policy vectors $x_c$ from the uniform $(-1,1)^{24}$ distribution over the feature space. In order to derive utilities for each policy, we adopt two options:

\emph{Option $1$: } We consider the utility of voter $i$ for policy/candidate $c$ to be $u_i(c)=u^i \cdot x_c$. To generate approval profiles, we consider various threshold classes, where a voter approves the top $\alpha$-fraction of candidates by utility for $\alpha=0.2,0.3,0.4$.

\emph{Option $2$: } We sample $K=10,000$ random moral dilemma pairs $(f_a,f_b)$, where $f_a,f_b\in\mathbb{R}^{24}$ are drawn independently 
from $\mathrm{Unif}(-1,1)^{24}$ and represent the moral-feature vectors of two hypothetical outcomes in a dilemma (e.g., two groups of pedestrians 
an autonomous vehicle must choose between). For each pair, a country/voter $i$ chooses the outcome in line with its moral preferences via $\mathrm{sign}(u^i\cdot (f_a-f_b))$; a policy $p$ picks as $\text{sign}(u^p\cdot (f_a-f_b))$. We say a voter $i$ has a \textit{strong opinion} on a pair $(f_a,f_b)$ if and only if $|u^{i}(f_a-f_b)|$ exceeds its own $q-$quantile among this value for all $K$ pairs, where $q\in \{0.25, 0.50, 0.75\}$. A country then approves a candidate policy $p$ if and only if it agrees on more than half of the dilemmas the country has a strong opinion on.

Each configuration produces a pair of .csv files encoding per-voter utilities and binary approvals over candidates. Sorting each voter's utility values provides an ordinal ranking over policies. Combined with approval sets, the rankings form a profile on which all our solution concepts are analyzed. In total, this yields six experimental variants: three for \emph{Option~1} ($\alpha \in \{0.2, 0.3, 0.4\}$) and three for \emph{Option~2} ($q \in \{0.25, 0.50, 0.75\}$), each containing $10$ elections with $n=137$ voters and $m=100$ candidates ($60$ elections in total).

\subsection{Dataset 2: Politics}
Our second class of datasets are nine political elections drawn from PrefLib \citep{MaWa13a} and the Voter Autrement (VA) project \citep{VA2007,BAUJARD2014131,VA2022}. We use the Sciences Po dataset (PrefLib 00029)\cite{LaslierVanDerStraeten2004} and the elections from four survey rounds: VA-2007 (one in-situ election with approval and score ballots), VA-2012 (three cities: Louvigny, Saint-Etienne, Strasbourg, each with approval and cardinal scores), VA 2017 online (two variants: continuous 0-100 evaluations and STV 1-11 rankings) and VA-2022 online (two variants: $\{-1,0,1,2\}$ score ballots and Borda-style top-$4$-of-$12$ rankings). For elections where only ordinal rankings are available (VA~2017 STV, VA~2022 Borda), we synthesize cardinal utilities using Borda Scores, with ties receiving the average of the positions they share. Importantly, the approval sets provided in the dataset did not satisfy upward closure (i.e., there were many instances where a voter would rank a disapproved candidate $j$ higher than an approved candidate $i$). Further, many voters only ranked their most preferred candidates. In this case, all non-ranked candidates are treated as one indifference class. Note that for an election concerning minority protection, we recommend also querying the least-preferred candidates of each voter. Therefore, the results from this dataset may be heavily-influenced by the preference-retrieval procedure, in addition to voter preferences.

\subsection{Dataset 3: The CoVal Dataset}
The CoVal dataset \citep{coval2026} collects human assessments of AI-generated responses to ethical prompts. Each assessment comprises two rankings (their personal ranking, and a world ranking of what they think would be best for the world overall) of four candidate responses $(A,B,C,D)$, and allows annotators to optionally flag candidates as \texttt{unacceptable}. We restrict to voters who provided a personal ranking. Candidates flagged as \texttt{unacceptable} by a voter are disapproved, and all others are approved. We discard prompts with no disapproving voters, and we synthesize utilities via Borda scores from the rankings. This yields $319$ election instances per variants, each with $4$ candidates and variable numbers of voters.

\subsection{Dataset 4: Kidney Exchange}
We construct $60$ approval-based elections from the kidney-allocation survey data of \citet{keswani2026moral}. The dataset comprises $301$ human-reported importance weights over $8$-features relevant to kidney-allocation decision-making. We simulate $100$ potential kidney recipients per election as feature vectors sampled independently and identically distributed from Uniform$(-1,1)$ in the $8-$dimensional space. As with the Moral Machine dataset, we assume that a voter's utility for a candidate is the scalar product of their weight vector and the candidate's feature vector. A voter approves a candidate if and only if their utility for the candidate is positive.

\subsection{Dataset 5: Synthetic Data} \label{synthetic_data_description}
For the purposes of analyzing our solution concepts against suitable domain restrictions, we generate elections under six standard models from the literature. Under the \textit{Impartial Culture} model, each of $n=50$ voters draws from a uniformly random ranking of $m=5$ candidates and approves a random number of their top-ranked candidates (approval counts drawn uniformly from $\{0,1,\ldots, m\}$.) Under the \textit{Mallows} model, rankings are drawn from a distribution centered on a fixed reference ranking with Kendall-$\tau$ distance and dispersion parameter $\phi\in \{0.2,0.5,0.8\}$, where lower $\phi$ produces stronger consensus. Approval sets are generated randomly as before. Lastly, under the spatial model, candidates and voters are placed uniformly on the line $[0,1]$ and the unit cube $[0,1]^2$. We assume utilities to be the negative Euclidean distance between the voter and the candidate. A voter approves every candidate within a fixed radius. For each configuration, we generate $100$ independent elections.

\begin{table}[H]
    \centering
    \resizebox{\textwidth}{!}{%
        \input{Current_Work/Neurips/plots/diagram6_unanimity.tex}    }
    \caption{%
    \emph{Decisively Better Candidates}. Column headers report
    \(N\), the total number of elections in the dataset, and \(N_u\),
    the number of those elections that contain at least one unanimously
    approved candidate. For each cell, the first column (\#) gives the
    number of those \(N_u\) elections in which the rule selects at
    least one non-unanimously-approved candidate, witnessing a
    violation of exclusive unanimity. The second column
    (tot.)\ gives the total count of non-unanimously-approved
    candidates the rule selected across those elections. 
    }
    \label{tab:unanimity}
\end{table}

\begin{table}[H]
    \centering
    \resizebox{\textwidth}{!}{%
        \input{Current_Work/Neurips/plots/diagram2_submajority.tex}%
    }
    \caption{%
    \emph{Sufficient support}. The first column (\#) gives the number of
    elections in which the concepts selects a candidate \(c\) with \(,\score{c}{\mathcal I}<\frac{n}{2}\). The second column
    ($\widetilde{\min\,s_c/n}$) gives the median over the minimal approval fraction $\min_{c\in f(\mathcal I)} \score{c}{\mathcal I}/n$ of a selected candidate. Instances in which
    the concept returns the empty set contribute neither to the count nor
    to the median.%
    }
    \label{tab:submajority}
\end{table}

\section{Discussion on Experiments}

\subsection{Effect of Candidate Pool Size on the Behavior of Solution Concepts}
One set of experiments we run on the synthetically-generated data is to sweep over the number of candidates $m$, measuring how each of the solution concepts behave. For each value of $m\in \{5,10,15,\ldots,50\}$, we generate $100$ synthetic elections (split across the six standard models described in section \ref{synthetic_data_description}) and observe changes in selectiveness, sub-majority selections, minority protection and unanimity violations. 

\paragraph{Selectiveness: } Observing Figure \ref{m_sweep:selectivity}, it is clear that the DSS$\circ$AWVC becomes increasingly more selective as $m$ grows, dropping from $\approx 0.15$ to $0.20$ at $m=5$ to $\approx 0.035$ by $m=50$. The DVC by contrast remains consistent in the fraction it approves across all $m$, staying around $0.55$ to $0.60$. The AWVC stays stable at around $0.27$ to $0.30$. On the other hand, ARU stands out as the only solution concept that becomes \textit{less} selective with growing $m$, rising from $0.19$ at $m=5$ to $\approx 0.37$ to $0.40$ by $m=40$. 

\begin{figure}[H]
    \centering
    \resizebox{\textwidth}{!}{%
        \input{Current_Work/Neurips/plots/m_sweep/m_sweep_selectivity}%
    }
    \caption{%
    \emph{Selectivity across $m$}: On the vertical axis, the average fraction of candidates selected ($f(\mathcal{I})/m$) is plotted as the number of candidates $m$ varies from $5$ to $50$, with $n=50$ voters and $100$ elections per value of $m$. This data is averaged across the results from all six models.
    }
    \label{m_sweep:selectivity}
\end{figure}
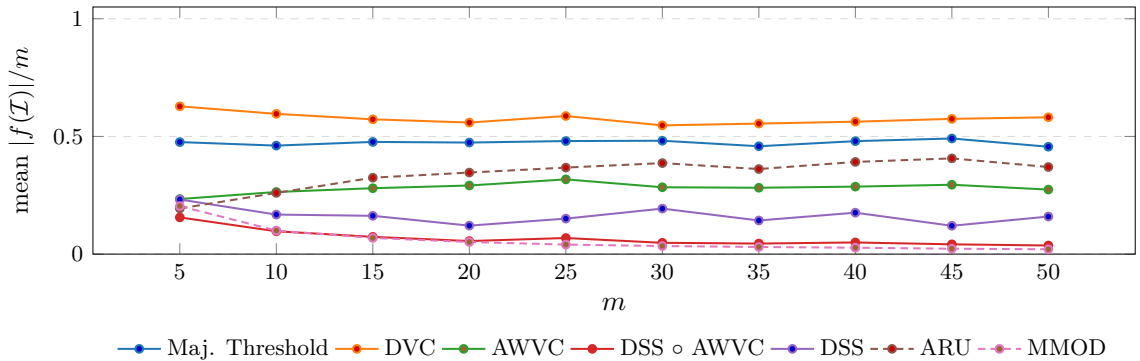

\paragraph{Sufficient Support: } With respect to sufficient support, as observed in Figure \ref{m_sweep:submajorityviolations}, the DVC's rate of selecting sub-majority candidates rises with growing $m$. 

\begin{figure}[H]
    \centering
    \resizebox{\textwidth}{!}{%
        \input{Current_Work/Neurips/plots/m_sweep/m_sweep_submajorityviolations}%
    }
    \caption{\emph{Sufficient Support across $m$: }Left panel: fraction of elections in which the solution concepts select at least one sub-majority candidate, i.e., candidate $c$ with $s_c(\mathcal{I}) < n/2$. Right panel: median across elections of the minimum approval fraction among the selected candidates, $\min_{c \in f(\mathcal{I})} s_c(\mathcal{I})/n$. Both panels sweep $m$ from $5$ to $50$ with $n=50$ voters and $100$ elections per value of $m$. This data is averaged across the results from all six models.}
    \label{m_sweep:submajorityviolations}
\end{figure}
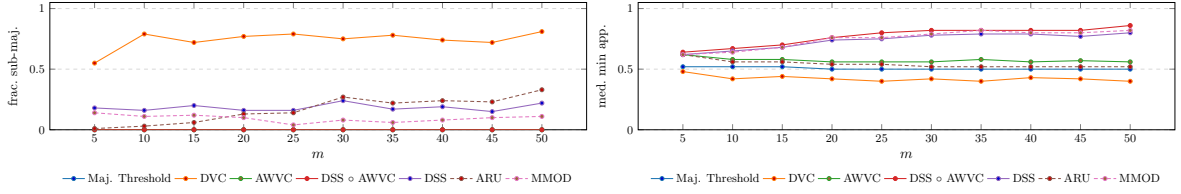

\paragraph{Minority Protection: }The minority protection trends (Figure \ref{m_sweep:minority_protection}) show that the DVC and AWVC never produce $\omega^*<1$ at any $m$, but their median $\omega^*$s behave differently with growing $m$: DVC's median decreases from $1.25$ to $1.00$, whilst the AWVC remains stable around $1.9$ to $2.0$. The DSS$\circ$AWVC displays the strongest increase, with it median $\omega^*$ rising from $2.6$ at $m=5$ to $7.1$ at $m=50$, showing that selected candidates become harder to block as the candidate pool grows. The DSS displays a similar upward trend. 
\begin{figure}[H]
    \centering
    \resizebox{\textwidth}{!}{%
        \input{Current_Work/Neurips/plots/m_sweep/m_sweep_minority_protection}%
    }
    \caption{\emph{Minority Protection Across $m$:} Left panel: fraction of elections in which the solution concepts select at least one candidate with $\omega^*(c)<1$. Right panel: median across elections of $\min_{c\in f(\mathcal{I})}\omega^*(c)$. Both panels sweep $m$ from $5$ to $50$ with $n=50$ voters and $100$ elections per value of $m$, pooled across the six synthetic preference models}
    \label{m_sweep:minority_protection}
\end{figure}
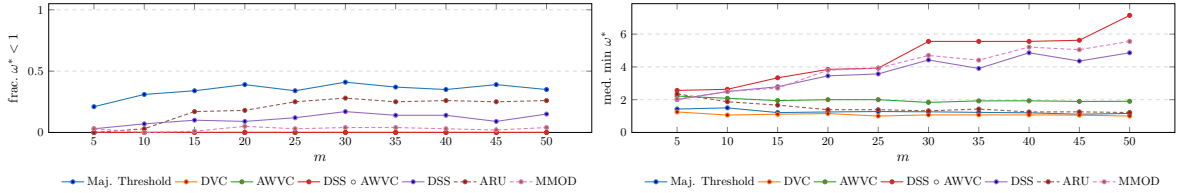

\paragraph{Decisively Better Candidate: } On exclusive unanimity, Figure \ref{m_sweep:unanimity} highlights how the solution concepts are partitioned into two groups: DVC, AWVC, Majority Threshold and ARU violate unanimity in $100\%$ of elections, i.e., select at least one candidate that is not unanimously approved in $100\%$ of elections where a unanimously approved candidate exists. The total count of non-unanimous selections grows with $m$, peaking at $m=40$ at $\approx 170$ to $190$ for DVC and Majority Threshold. 
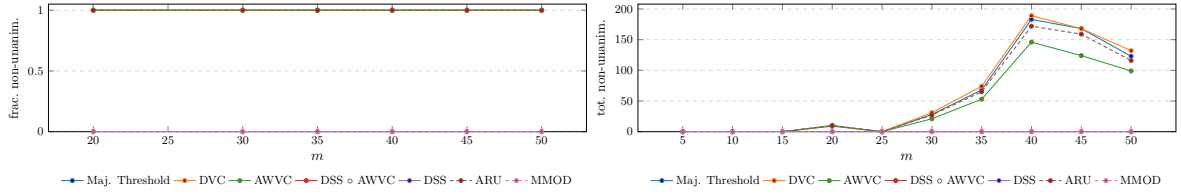
\begin{figure}[H]
    \centering
    \resizebox{\textwidth}{!}{%
        \input{Current_Work/Neurips/plots/m_sweep/m_sweep_unanimity}%
    }
    \caption{\emph{Decisively Better Candidates Across $m$: }Left panel: fraction of elections (among those containing at least one unanimously approved candidate) in which the rule selects at least one non-unanimously-approved candidate. Right panel: total count of non-unanimously-approved candidates selected across those elections. Both panels sweep $m$ from $5$ to $50$ with $n=50$ voters and $100$ elections per value of $m$, pooled across six synthetic preference models. The left panel shows no data points for small $m$ because no elections in those configurations contain a unanimously approved candidate.}
    \label{m_sweep:unanimity}
\end{figure}

\subsection{Rank Consistency of $\omega^*$ and $\epsilon^*$}
The experiment below tests whether the two independent measures of candidate protection, i.e., critical endowment $\omega^*(c)$ and critical excess budget $\epsilon^*(c)$ produce consistent rankings across candidates within each election. For every election, we pair $\omega^*(c)$ and $\epsilon^*(c)$ for each candidate, discarding those where we were unable to compute $\epsilon^*$ (cf.\ \Cref{app:RelWorkBudgetingGames}). We then compute Spearman $\rho$ and Kendall $\tau$ over the remaining candidates. Since higher $\omega^*$ signals safety and higher $\epsilon^*$ signals vulnerability, $\rho \approx -1$ when the two measures agree. 

\begin{figure}[H]
    \centering
    \resizebox{\textwidth}{!}{%
        \input{Current_Work/Neurips/plots/correlation_epsilon_omega/clamped_histogram_simple}%
    }
    \caption{Distribution of the number of candidates per election whose critical excess budget $\epsilon^*(c)$ could not be computed. For small-$m$ datasets (Synthetic, CoVal, Political), most elections have few or no clamped candidates. For large-$m$ datasets (Kidney with $m=100$, Moral Machine with $m=100$), the majority of candidates per election are clamped, reflecting that most candidates are deeply safe and far from the blocking threshold.}
    \label{correlation_epsilon_omega:simple}
\end{figure}
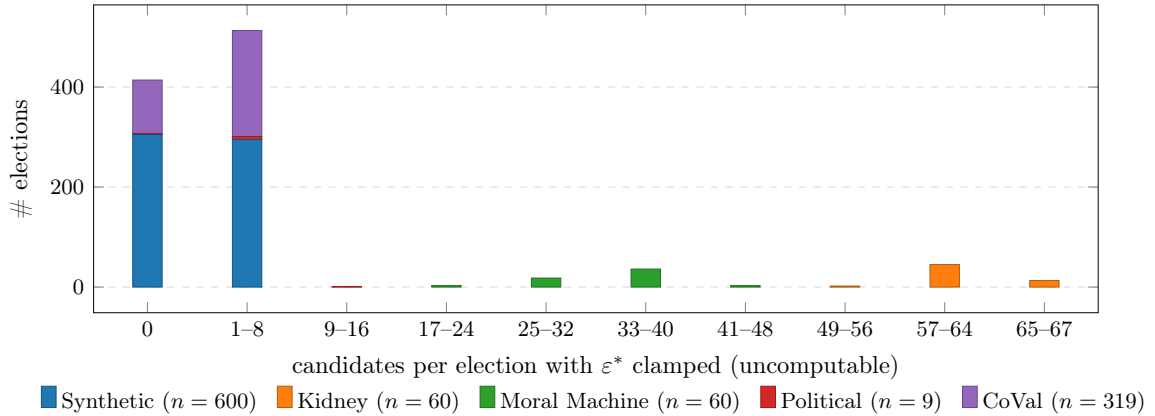

Figure~\ref{correlation_epsilon_omega:scatter_simple} plots $\omega^*(c)$ against $\epsilon^*(c)$ for every (election, candidate) pair where both values are well-defined, pooled across all five datasets. The two measures exhibit a clear monotone negative relationship: candidates with high $\omega^*$ (robust) consistently have low or negative $\epsilon^*$ (safe), and vice versa. The dashed line at $\epsilon^*=0$ separates blocked candidates (above) from safe ones (below).

Quantitatively, the per-election Spearman $\rho$ between $\omega^*$ and $\epsilon^*$ is close to $-1$ across all datasets. The Moral Machine elections show the tightest agreement, with mean $\rho$ between $-0.997$ and $-0.998$ across all six variants. The Synthetic elections range from $\rho=-0.953$ (Mallows $\phi=0.5$) to $\rho=-0.996$ (Mallows $\phi=0.2$), and the Kidney dataset achieves $\rho=-0.924$. CoVal, with only $m=4$ candidates per election, yields mean $\rho=-0.971$ (median $-1.0$). The Political elections are more variable---mean $\rho$ ranges from $-0.653$ (Virginia 2012) to $-1.0$ (PrefLib, Virginia 2007, Virginia 2022)---but these groups contain at most 3 elections each, so individual outliers have outsized influence. Overall, the strong negative correlation empirically confirm that $\omega^*$ and $\epsilon^*$ rank candidates by safety in the same order.

\begin{figure}[H]
    \centering
    \resizebox{\textwidth}{!}{%
        \input{Current_Work/Neurips/plots/correlation_epsilon_omega/pooled_scatter_simple.tex}%
    }
    \caption{Scatter plot of critical endowment $\omega^*(c)$ against critical excess budget $\epsilon^*(c)$ for DVC, pooled across all (election, candidate) pairs where both values are well-defined. The dashed line at $\epsilon^*=0$ separates blocked candidates (above) from safe ones (below). The original dataset contains $8{,}911$ points; the plot displays a uniform random subsample of $2{,}970$ points (one-third) for compilation efficiency. The negative monotone relationship confirms that the two measures rank candidates by safety in the same order.}
    \label{correlation_epsilon_omega:scatter_simple}
\end{figure}
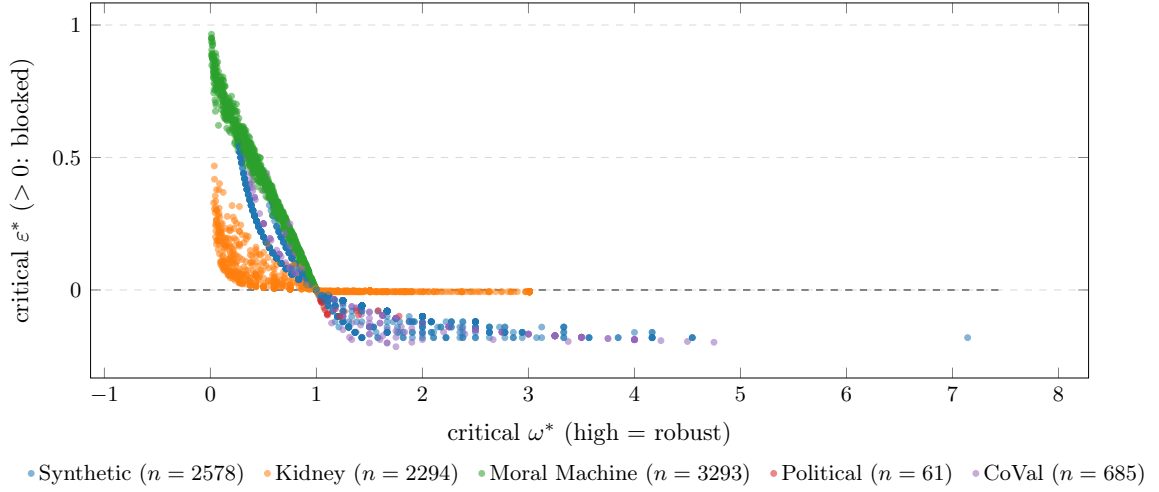

\subsection{Solution Concepts on the CoVal Dataset}
We highlight a few observations across the experiments on the CoVal dataset.
\paragraph{Selectivity across CoVal: } First, we observe the trends from Figure \ref{fig:min_approval_coval}, which report the distribution of $\min_{c \in f(\mathcal{I})} s_c(\mathcal{I})/n$ across $N=319$ elections. The DVC exhibits the most spread among the solution concepts, with several elections where the least-approved selected candidate falls below the majority threshold (represented by the dashed line at $0.5$). The AWVC, on the other hand, concentrates its mass higher, showing how approval-weighting penalizes low-approval candidates heavily. DSS $\circ$ AWVC and DSS both push the distribution higher, with the DSS achieving its median minimum approval fraction near $1.0$. This indicates how in most elections, every DSS-selected candidate has nearly-unanimous approval. On the other hand, ARU and MMOD behave similarly, with medians around $0.8$ and occasional outliers below $0.5$, while the Majority Threshold rule never selects below $0.5$ by construction and serves as a baseline.

\paragraph{Minimum Critical Endowments, CoVal: } Next, we observe Figure \ref{fig:min_omega_coval}, which reports the distribution of minimum critical endowments of selected candidates per solution concept. Majority Threshold and the ARU occasionally dip below the $\omega^*=1$ threshold, indicating that these solution concepts can select candidates that are blockable under the standard veto budget. Compared to DVC and majority Threshold, the AWVC shifts the distribution upward, with $43$ elections where every selected candidate is unanimously approved. DSS$\circ$AWVC, DSS, and MMOD push further: all three have medians above $4$ and place the majority of their elections ($175$ out of $319$) in the $\infty$ band, indicating that on most CoVal instances these solution concepts select only unanimously approved candidates. ARU sits between these two groups, with a median around $2.5$ and 42 elections at $\infty$, comparable to AWVC. The overall pattern suggests the following observation: in terms of providing minority protection, DVC $<$ AWVC $\approx$ ARU $<$ DSS~$\circ$~AWVC $\approx$ DSS $\approx$
MMOD on this dataset.

\begin{figure}[H]
    \centering
    \includegraphics[width=\linewidth]{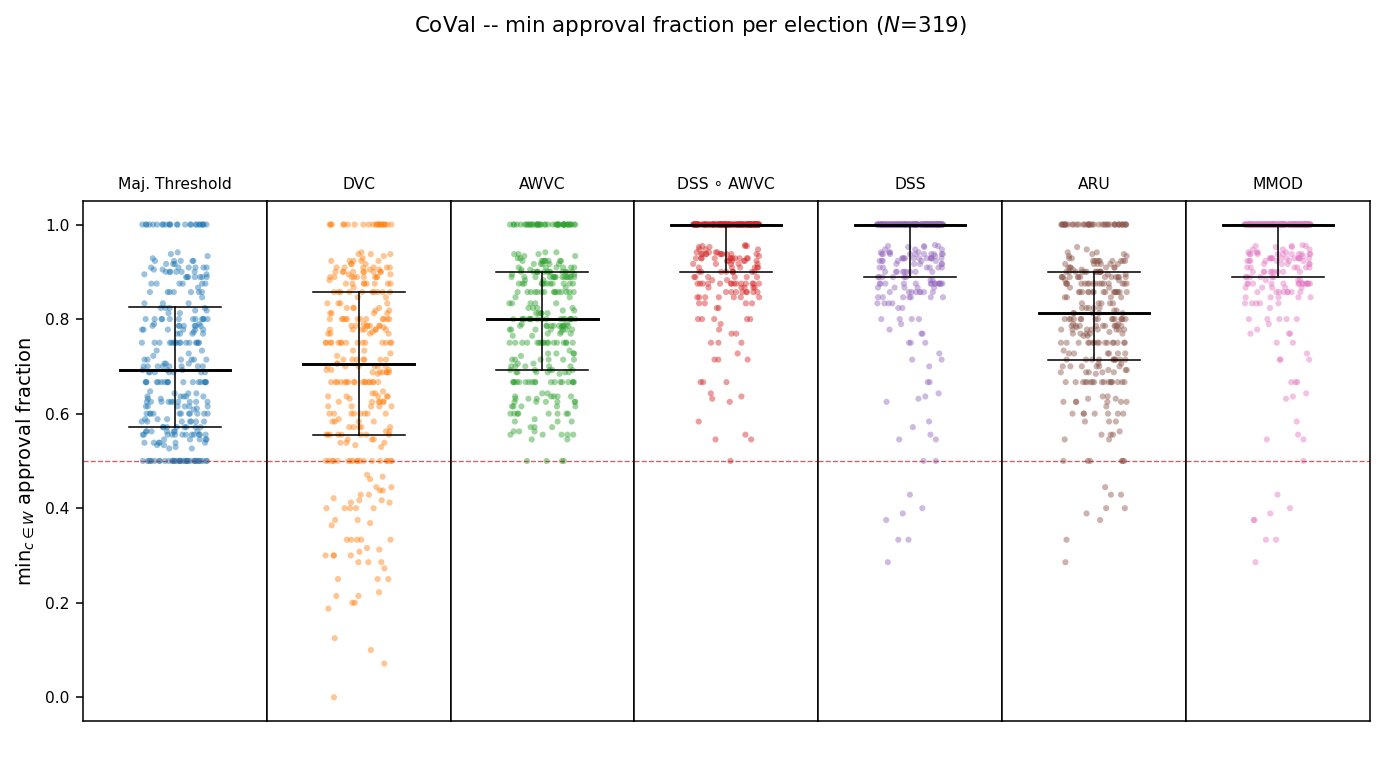}
    \caption{%
    Distribution of the minimum approval fraction of a selected candidate on the CoVal dataset, per
    solution concept. Each dot represents one election: its vertical position is
    \(\min_{c\in f(\mathcal I)}\score{c}{\mathcal I}/n\), the lowest
    approval fraction among the candidates the solution concept selected in that
    election. Elections in which the solution concept returned the empty set are
    omitted. The dashed red line marks the majority threshold of \(0.5\):
    points below it correspond to elections in which the solution concept selected at
    least one sub-majority candidate. Concentration of points near the
    top of the strip indicates that the solution concept consistently restricts itself
    to broadly approved candidates. 
    }
    \label{fig:min_approval_coval}
\end{figure}

\begin{figure}[H]
    \centering
    \includegraphics[width=\linewidth]{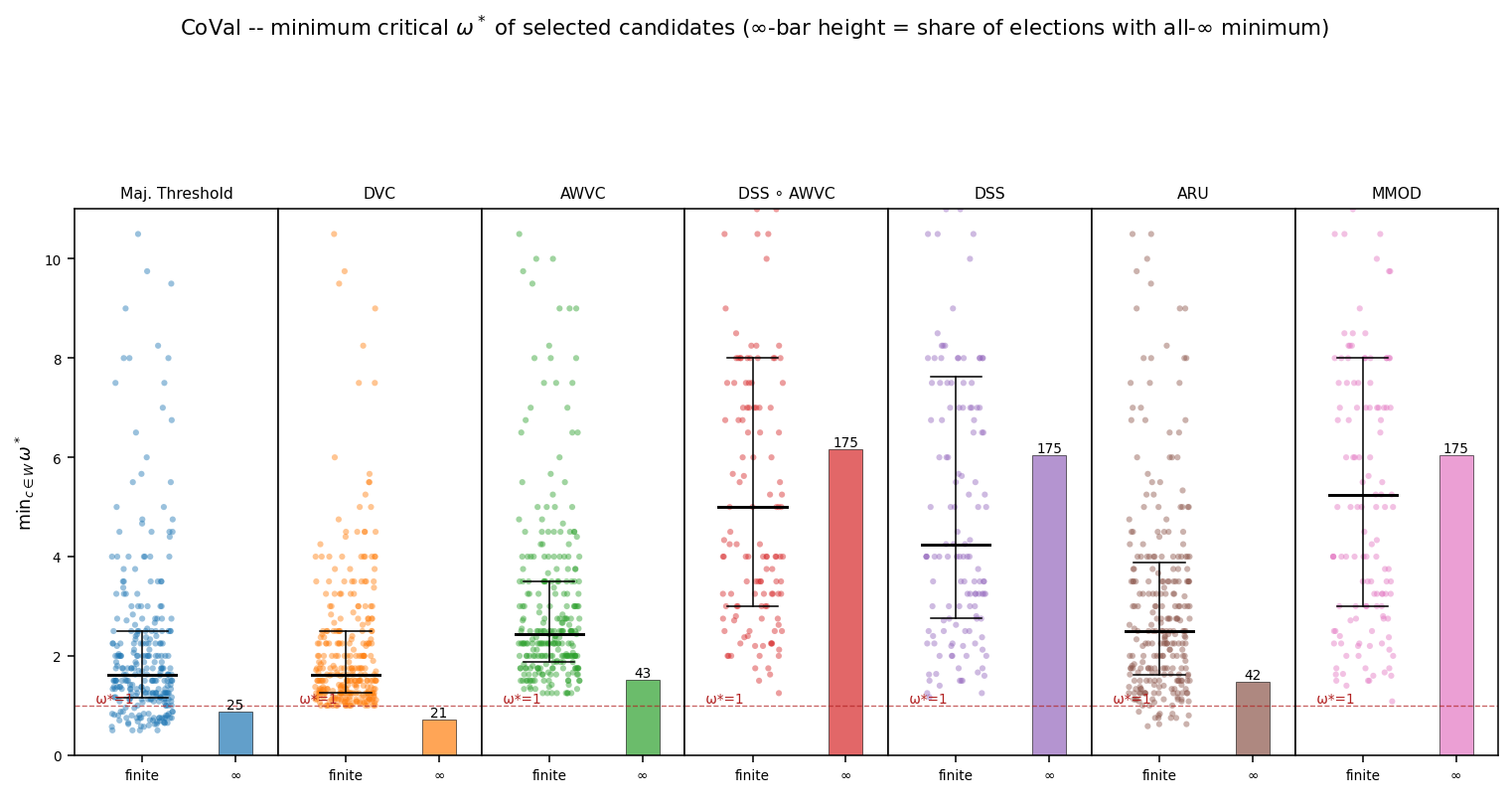}
    \caption{%
    Distribution of minimum critical endowments of selected candidates on the CoVal dataset, per
    solution concept. Each dot represents one election; its vertical position is
    \(\min_{c\in f(\mathcal I)}\omega^*(c)\), the lowest critical
    endowment among the candidates the solution concept selected (Definition~13).
    Each strip is split into two bands: the left ``finite'' band shows
    elections with finite minimum \(\omega^*\), and the right ``\(\infty\)''
    band aggregates elections in which every selected candidate is
    unanimously approved and therefore unblockable under any
    \(\omega\)-\dvc{}. Lower finite values indicate that the solution concept selects
    candidates the electorate could veto with less than the standard
    veto budget; concentration in the \(\infty\) band indicates the solution concept
    is selecting only unanimously-approved candidates on this dataset.
    }
    \label{fig:min_omega_coval}
\end{figure}

\subsection{Experimental Variants: Synthetic and Moral Machine}
Our next set of observations distinguishes the experimental variants on the Moral Machine and the synthetically-generated datasets to make sure that no signals are aggregated away when reporting them each as one combined dataset. 

\paragraph{Synthetic Data Variants: }For our discussion on the patterns across the experiments on synthetically-generated preference models, we refer to Figures \ref{fig:synthetic_variation_selection_fraction}, \ref{fig:synthetic_variation_sufficient_support} and \ref{fig:synthetic_minority_protection}. Figure \ref{fig:synthetic_variation_selection_fraction} reveals that selectivity varies across models. Under Impartial Culture, DVC and Majority Threshold select the largest fractions, while AWVC and DSS $\circ$ AWVC are the most selective. As the Mallows dispersion parameter $\phi$ increases from $0.2$ to $0.8$, DVC and Majority Threshold are more permissive, whereas solution concepts like DSS, DSS$\circ$AWVC and MMOD remain relatively stable throughout. The spatial models produce the most selective behavior on the whole, across all solution concepts. 

In terms of sufficient support, Figure \ref{fig:synthetic_variation_sufficient_support} shows that the DVC is most prone to selecting sub-majority candidates, with median minimum approval consistently below $0.50$. The DSS too is prone to selecting sub-majority candidates, albeit less than the DVC. On the other hand, ARU only selects a sub-majority candidate in $1$ and $2$ elections in IC and Mallows $\phi=0.2$ respectively. AWVC and DSS$\circ$ AWVC never select a sub-majority candidate, maintaining median minimum approval above $0.50$ across all models. 

Lastly, Figure \ref{fig:synthetic_minority_protection} shows that the Majority Threshold frequently selects candidates with $\omega^*(c)<1$, especially in Spatial 1D, where it does so in $68$ out of $100$ elections, with a median $\omega^*$ of $0.67$. While the DVC never falls below $\omega^*=1$, its median is modest ($1.0-1.43$). The DSS falls below $1$ in $9$ out of $100$ IC elections. The AWVC, DSS~$\circ$~AWVC, and MMOD never produce $\omega^*<1$ in any
model and achieve the highest median endowments ($2.00$--$4.55$), showing that they select candidates that remain unblockable even under inflated veto budgets. 

\paragraph{Moral Machine Variants: }For our discussion on the patterns across the experimental variants in the Moral Machine dataset, we refer to Figures \ref{fig:moral_selecivity},\ref{fig:moral_support},\ref{fig:moral_minority}, and \ref{fig:moral_unanimous}. Across all six variants, DSS$\circ$AWVC and AWVC are the most selective (Figure \ref{fig:moral_selecivity}), roughly selecting $10-15 \%$ of candidates. The DVC and Majority Threshold are the least-selective. This selectivity gap widens under the Option $2$ variants with increasing quantile thresholds ($q$). The DVC is the only solution concept that selects sub-majority candidates across every variant (Figure \ref{fig:moral_support}). Majority threshold also exhibits instances of sub-majority selections, albeit with higher median minimum approvals. In strict contrast, DSS, DSS$\circ$AWVC and AWVC never select a sub-majority candidate in every variant. In terms of minority protection, the results from Figure \ref{fig:moral_minority} show that DSS, DSS$\circ$AWVC and MMOD achieve $\omega^*=\infty$ across all variants, showing that every selected candidate is unanimously approved and unblockable. DVC and Majority Threshold remain finite, with DVC's median $\omega^*$ values being close to $1$, i.e., close to the boundary at which candidates become blockable. The ARU lies in between, having finite $\omega^*$ values that lie comfortably above $1$. Lastly, we observe the unanimity consistency results, where Figure \ref{fig:moral_unanimous} show that the DSS, DSS$\circ$AWVC and MMOD never select a non-unanimously-approved candidate in any variant, while Majority Threshold, DVC and AWVC violate unanimity in all $10$ elections of every variant. The DVC has the highest count of non-unanimous selections. 

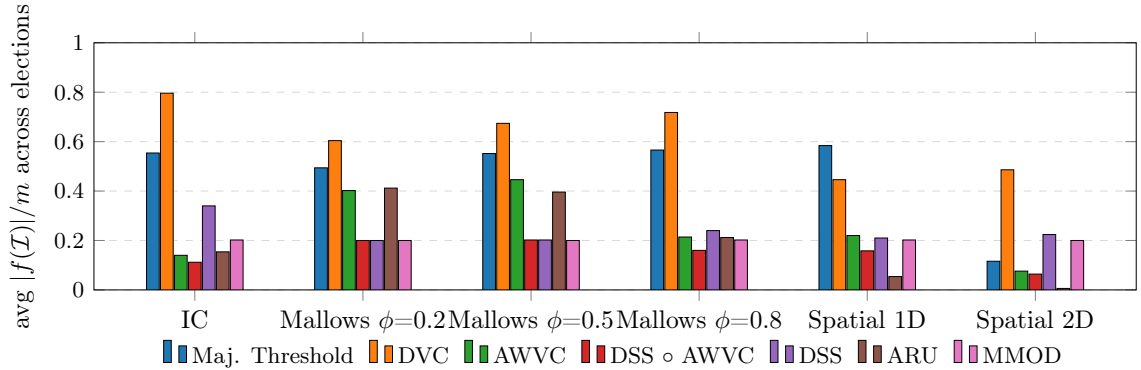
\begin{figure}[H]

    \centering
    \resizebox{\textwidth}{!}{%
        \input{Current_Work/Neurips/plots/Synthetic_MoralVariation_Plots/synthetic_selection_fraction}%
    }
    \caption{Selectivity Across Domains}
    \label{fig:synthetic_variation_selection_fraction}
\end{figure}

\begin{figure}[H]
    \centering
    \resizebox{\textwidth}{!}{%
        \input{Current_Work/Neurips/plots/Synthetic_MoralVariation_Plots/synthetic_submajority}%
    }
    \caption{\emph{Sufficient support across domains }The first column in each domain (\#) gives the number of instances in which the concept selected at least one candidate with approval below a strict majority (i.e., $s_c(\mathcal{I}) < n/2$). The second column (med min app.) gives the median across elections of the minimum approval fraction among selected candidates, $\min_{c \in f(\mathcal{I})} s_c(\mathcal{I})/n$. A higher median min approval and a lower $\#$ indicate that the solution concept consistently selects well-supported candidates. The table shows results across six synthetic preference models ($N=100$ elections each).}\label{fig:synthetic_variation_sufficient_support}
\end{figure}

\begin{figure}[H]
    \centering
    \resizebox{\textwidth}{!}{%
        \input{Current_Work/Neurips/plots/Synthetic_MoralVariation_Plots/synthetic_minority_protection}%
    }
    \caption{\emph{Minority Protection Across Domains:} The first column in each domain(\#) gives
    the number of instances in which the concept selected at least one
    candidate with \(\omega^*(c)<1\). The
    second column (med \(\omega^*\)) gives the median across elections
    of \(\min_{c\in f(\mathcal I)} \omega^*(c)\); this value can be  \(\infty\)
    in case all selected candidates are unanimously approved
    and therefore unblockable. The table shows results across six synthetic preference models ($N=100$ elections each).}
    \label{fig:synthetic_minority_protection}
\end{figure}

\begin{figure}[H]
    \centering
    \resizebox{\textwidth}{!}{%
        \input{Current_Work/Neurips/plots/Synthetic_MoralVariation_Plots/moral_selectivity}%
    }
    \caption{\emph{Selectivity} across experimental variants (Moral Machine): Each bar shows the average fraction of candidates selected ($|W|/m$) across elections for a given solution concept and variant. Option $1$ variants sample candidate weights with noise parameter $\alpha$; Option 2 variants construct candidates from pairwise comparisons with quantile threshold $q$. Each variant comprises $N=10$ elections with $m=100$ candidates.}
    \label{fig:moral_selecivity}
\end{figure}
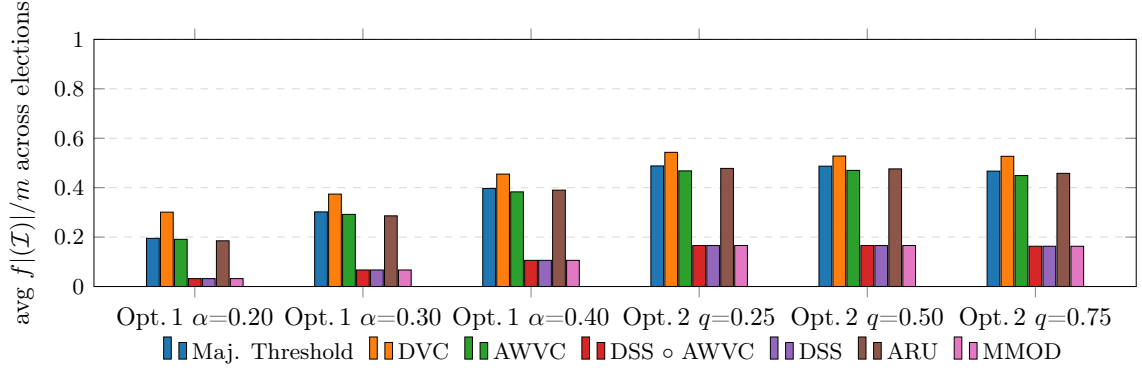

\begin{figure}[H]
    \centering
    \resizebox{\textwidth}{!}{%
        \input{Current_Work/Neurips/plots/Synthetic_MoralVariation_Plots/moral_submajority}%
    }
    \caption{\emph{Sufficient Support }across experimental variants (Moral Machine). The first column in each variant (\#) gives the number of instances in which the concept selected at least one candidate with approval below a strict majority (i.e., $s_c(\mathcal{I}) < n/2$). The second column (med min app.) gives the median across elections of the minimum approval fraction among selected candidates, $\min_{c \in f(\mathcal{I})} s_c(\mathcal{I})/n$. A higher median min approval and a lower \# indicate that the solution concept consistently selects well-supported candidates. Each variant comprises $N=10$ elections with $m=100$ candidates.}
    \label{fig:moral_support}
\end{figure}

\begin{figure}[H]
    \centering
    \resizebox{\textwidth}{!}{%
        \input{Current_Work/Neurips/plots/Synthetic_MoralVariation_Plots/moral_minority_protection}%
    }
    \caption{\emph{Minority Protection} across experimental variants (Moral Machine). The first column in each variant (\#) gives the number of instances in which the concept selected at least one candidate with $\omega^*(c) < 1$. The second column (med $\omega^*$) gives the median across elections of $\min_{c \in f(\mathcal{I})} \omega^*(c)$; this value can be $\infty$ in case all selected candidates are unanimously approved and therefore unblockable. Each variant comprises $N=10$ elections with $m=100$ candidates.
}
    \label{fig:moral_minority}
\end{figure}

\begin{figure}[H]
    \centering
    \resizebox{\textwidth}{!}{%
        \input{Current_Work/Neurips/plots/Synthetic_MoralVariation_Plots/moral_unanimity}%
    }
    \caption{\emph{Decisively better candidates } across experimental variants (Moral Machine). Column headers report $N$, the total number of elections in the variant, and $N_u$, the number of those elections that contain at least one unanimously approved candidate. For each cell, the first column (\#) gives the number of those $N_u$ elections in which the solution concept selects at least one non-unanimously-approved candidate, witnessing a violation of exclusive unanimity. The second column (tot.) gives the total count of non-unanimously-approved candidates the solution concept selected across those elections. Lower values indicate better unanimity consistency. Each variant comprises $N=10$ elections with $m=100$ candidates.}
    \label{fig:moral_unanimous}
\end{figure}

\section{Relational Blocking Revisited}
\label{subsec:global-blocking}

In economic theory, a natural way to declare that a candidate is clearly better than another one is through \emph{Pareto domination}, and it may be interesting to adapt this concept to our setting. We describe one possible way of doing so: we write $c \mathrel{P^*} d$ (and say that \emph{$c$ disapproval Pareto-dominates $d$}) for $c,d\in C$, if $c\in A_i$ or $c\succsim_i d$ for all $i\in N$, while further there is at least one $j\in N$ with $d\notin A_j$ and ($c\succ_j d$ or $c\in A_j$).

\begin{definition}[disapproval Pareto set]
    The \emph{disapproval Pareto set} is
    $
        \po(\mathcal I)
        =
        \{c\in C \mid d \mathrel{P^*} c \text{ for no } d\in C\},
    $.
\end{definition}

Under this solution concept, a candidate is blocked if and only if it is disapproval Pareto-dominated. It may be more suitable than $\tc{}$ for real-world elections, as there the latter often only outputs a single candidate.


\end{document}

%% file: Current_Work/Neurips/plots/diagram1_selection_fraction.tex
%
%
\definecolor{papercolord1rule0}{RGB}{31,119,180}
\definecolor{papercolord1rule1}{RGB}{255,127,14}
\definecolor{papercolord1rule2}{RGB}{44,160,44}
\definecolor{papercolord1rule3}{RGB}{214,39,40}
\definecolor{papercolord1rule4}{RGB}{148,103,189}
\definecolor{papercolord1rule5}{RGB}{140,86,75}
\definecolor{papercolord1rule6}{RGB}{227,119,194}
\begin{tikzpicture}
    \begin{axis}[
        ybar=0.5pt,
        bar width=5pt,
        width=\linewidth,
        height=3.5cm,
        ymin=0, ymax=1.0,
        ylabel={avg. $|f(\mathcal
        {I})|/m$},
        ylabel style={font=\small},
        symbolic x coords={CoVal,Synthetic,Political,Kidney,Moral Machine},
        xtick=data,
        x tick label style={font=\small},
        y tick label style={font=\footnotesize},
        ymajorgrids=true,
        grid style={dashed,gray!30},
        legend image code/.code={
            \draw[#1] (0cm,-0.1cm) rectangle (0.3cm,0.1cm);
        },
        legend style={
            at={(0.5,1.4)},
            anchor=north,
            legend columns=7,
            /tikz/every even column/.append style={column sep=4pt},
            font=\footnotesize,
            draw=none,
        },
        enlarge x limits=0.12,
    ]
        \addplot[ybar, fill=papercolord1rule0, draw=black, line width=0.3pt] coordinates {(CoVal,0.7955) (Synthetic,0.4777) (Political,0.1182) (Kidney,0.5045) (Moral Machine,0.3893)};
        \addlegendentry{Maj. Threshold}
        \addplot[ybar, fill=papercolord1rule1, draw=black, line width=0.3pt] coordinates {(CoVal,0.8307) (Synthetic,0.6207) (Political,0.9798) (Kidney,0.7920) (Moral Machine,0.4547)};
        \addlegendentry{DVC}
        \addplot[ybar, fill=papercolord1rule2, draw=black, line width=0.3pt] coordinates {(CoVal,0.6873) (Synthetic,0.2497) (Political,0.1182) (Kidney,0.4633) (Moral Machine,0.3755)};
        \addlegendentry{AWVC}
        \addplot[ybar, fill=papercolord1rule3, draw=black, line width=0.3pt] coordinates {(CoVal,0.3683) (Synthetic,0.1493) (Political,0.0795) (Kidney,0.0120) (Moral Machine,0.1167)};
        \addlegendentry{DSS $\circ$ AWVC}
        \addplot[ybar, fill=papercolord1rule4, draw=black, line width=0.3pt] coordinates {(CoVal,0.3801) (Synthetic,0.2360) (Political,0.0887) (Kidney,0.0120) (Moral Machine,0.1167)};
        \addlegendentry{DSS}
        \addplot[ybar, fill=papercolord1rule5, draw=black, line width=0.3pt] coordinates {(CoVal,0.6552) (Synthetic,0.2057) (Political,0.0000) (Kidney,0.4903) (Moral Machine,0.3788)};
        \addlegendentry{ARU}
        \addplot[ybar, fill=papercolord1rule6, draw=black, line width=0.3pt] coordinates {(CoVal,0.3292) (Synthetic,0.2010) (Political,0.0887) (Kidney,0.0100) (Moral Machine,0.1167)};
        \addlegendentry{MMOD}
    \end{axis}
\end{tikzpicture}

%% file: Current_Work/Neurips/plots/diagram5_low_omega.tex
\begin{tabular}{lrrrrrrrrrr}
\toprule
 & \multicolumn{2}{c}{CoVal} & \multicolumn{2}{c}{Synthetic} & \multicolumn{2}{c}{Political} & \multicolumn{2}{c}{Kidney} & \multicolumn{2}{c}{Moral Machine} \\
 & \multicolumn{2}{c}{($N$=319)} & \multicolumn{2}{c}{($N$=600)} & \multicolumn{2}{c}{($N$=9)} & \multicolumn{2}{c}{($N$=60)} & \multicolumn{2}{c}{($N$=60)} \\
\cmidrule(lr){2-3}\cmidrule(lr){4-5}\cmidrule(lr){6-7}\cmidrule(lr){8-9}\cmidrule(lr){10-11}
Rule & \# & med $\omega^*$ & \# & med $\omega^*$ & \# & med $\omega^*$ & \# & med $\omega^*$ & \# & med $\omega^*$ \\
\midrule
Maj. Threshold & 50 & 1.75 & 115 & 1.43 & 0 & 2.45 & 37 & 0.75 & 0 & 1.66 \\
DVC & 0 & 1.75 & 0 & 1.25 & 0 & 1.04 & 0 & 1.00 & 0 & 1.03 \\
AWVC & 0 & 2.55 & 0 & 2.14 & 0 & 2.45 & 0 & 1.91 & 0 & 1.97 \\
DSS $\circ$ AWVC & 0 & $\infty$ & 0 & 2.63 & 0 & 2.68 & 0 & 22.36 & 0 & $\infty$ \\
DSS & 0 & $\infty$ & 24 & 2.17 & 0 & 2.46 & 0 & 22.36 & 0 & $\infty$ \\
ARU & 12 & 2.87 & 3 & 2.22 & --- & --- & 27 & 1.00 & 0 & 1.82 \\
MMOD & 0 & $\infty$ & 18 & 2.17 & 0 & 2.46 & 0 & 22.33 & 0 & $\infty$ \\
\bottomrule
\end{tabular}

%% file: Current_Work/Neurips/plots/diagram6_unanimity.tex
\begin{tabular}{lrrrrrrrrrr}
\toprule
 & \multicolumn{2}{c}{CoVal} & \multicolumn{2}{c}{Synthetic} & \multicolumn{2}{c}{Political} & \multicolumn{2}{c}{Kidney} & \multicolumn{2}{c}{Moral Machine} \\
 & \multicolumn{2}{c}{($N$=319, $N_u$=175)} & \multicolumn{2}{c}{($N$=600, $N_u$=0)} & \multicolumn{2}{c}{($N$=9, $N_u$=0)} & \multicolumn{2}{c}{($N$=60, $N_u$=0)} & \multicolumn{2}{c}{($N$=60, $N_u$=60)} \\
\cmidrule(lr){2-3}\cmidrule(lr){4-5}\cmidrule(lr){6-7}\cmidrule(lr){8-9}\cmidrule(lr){10-11}
Rule & \# & tot. & \# & tot. & \# & tot. & \# & tot. & \# & tot. \\
\midrule
Maj. Threshold & 150 & 310 & 0 & 0 & 0 & 0 & 0 & 0 & 60 & 1636 \\
DVC & 154 & 304 & 0 & 0 & 0 & 0 & 0 & 0 & 60 & 2028 \\
AWVC & 132 & 230 & 0 & 0 & 0 & 0 & 0 & 0 & 60 & 1553 \\
DSS $\circ$ AWVC & 0 & 0 & 0 & 0 & 0 & 0 & 0 & 0 & 0 & 0 \\
DSS & 0 & 0 & 0 & 0 & 0 & 0 & 0 & 0 & 0 & 0 \\
ARU & 132 & 243 & 0 & 0 & 0 & 0 & 0 & 0 & 60 & 1573 \\
MMOD & 0 & 0 & 0 & 0 & 0 & 0 & 0 & 0 & 0 & 0 \\
\bottomrule
\end{tabular}

%% file: Current_Work/Neurips/plots/diagram2_submajority.tex
\begin{tabular}{lrcrcrcrcrc}
\toprule
 & \multicolumn{2}{c}{CoVal} & \multicolumn{2}{c}{Synthetic} & \multicolumn{2}{c}{Political} & \multicolumn{2}{c}{Kidney} & \multicolumn{2}{c}{Moral Machine} \\
 & \multicolumn{2}{c}{($N$=319)} & \multicolumn{2}{c}{($N$=600)} & \multicolumn{2}{c}{($N$=9)} & \multicolumn{2}{c}{($N$=60)} & \multicolumn{2}{c}{($N$=60)} \\
\cmidrule(lr){2-3}\cmidrule(lr){4-5}\cmidrule(lr){6-7}\cmidrule(lr){8-9}\cmidrule(lr){10-11}
Rule & \# & $\widetilde{\min\,s_c/n}$ & \# & $\widetilde{\min\,s_c/n}$ & \# & $\widetilde{\min\,s_c/n}$ & \# & $\widetilde{\min\,s_c/n}$ & \# & $\widetilde{\min\,s_c/n}$ \\
\midrule
Maj.\ Threshold & 0 & 0.69 & 0 & 0.54 & 0 & 0.59 & 0 & 0.50 & 0 & 0.52 \\
\dvc{}          & 52 & 0.71 & 347 & 0.48 & 9 & 0.04 & 60 & 0.08 & 60 & 0.28 \\
AWVC            & 0 & 0.80 & 0 & 0.60 & 0 & 0.59 & 0 & 0.51 & 0 & 0.61 \\
DSS $\circ$ AWVC & 0 & 1.00 & 0 & 0.65 & 0 & 0.62 & 0 & 0.96 & 0 & 1.00 \\
DSS              & 7 & 1.00 & 106 & 0.60 & 1 & 0.59 & 0 & 0.96 & 0 & 1.00 \\
ARU              & 9 & 0.81 & 3 & 0.60 & 0 & --- & 41 & 0.48 & 5 & 0.58 \\
MMOD            & 8 & 1.00 & 95 & 0.60 & 1 & 0.59 & 0 & 0.96 & 0 & 1.00 \\
\bottomrule
\end{tabular}

%% file: Current_Work/Neurips/plots/m_sweep/m_sweep_selectivity.tex
%
%
\definecolor{papercolormrule0}{RGB}{31,119,180}
\definecolor{papercolormrule1}{RGB}{255,127,14}
\definecolor{papercolormrule2}{RGB}{44,160,44}
\definecolor{papercolormrule3}{RGB}{214,39,40}
\definecolor{papercolormrule4}{RGB}{148,103,189}
\definecolor{papercolormrule5}{RGB}{140,86,75}
\definecolor{papercolormrule6}{RGB}{227,119,194}

\begin{tikzpicture}
    \begin{axis}[
        width=\linewidth, height=5cm,
        xlabel={$m$}, ylabel={mean $|f(\mathcal{I})|/m$},
        xtick={5,10,15,20,25,30,35,40,45,50}, x tick label style={font=\footnotesize},
        y tick label style={font=\footnotesize},
        ylabel style={font=\small},
        ymin=0, ymax=1.05, 
        ymajorgrids=true, grid style={dashed,gray!30},
        legend style={
            at={(0.5,-0.30)}, anchor=north,
            legend columns=7, font=\footnotesize, draw=none,
        },
    ]
        \addplot+[mark=*, mark size=1.5pt, color=papercolormrule0, thick] coordinates {(5,0.4760) (10,0.4610) (15,0.4767) (20,0.4740) (25,0.4804) (30,0.4817) (35,0.4583) (40,0.4800) (45,0.4913) (50,0.4562)};
        \addlegendentry{Maj. Threshold}
        \addplot+[mark=*, mark size=1.5pt, color=papercolormrule1, thick] coordinates {(5,0.6280) (10,0.5960) (15,0.5727) (20,0.5590) (25,0.5868) (30,0.5470) (35,0.5546) (40,0.5627) (45,0.5747) (50,0.5814)};
        \addlegendentry{DVC}
        \addplot+[mark=*, mark size=1.5pt, color=papercolormrule2, thick] coordinates {(5,0.2340) (10,0.2640) (15,0.2800) (20,0.2915) (25,0.3176) (30,0.2840) (35,0.2820) (40,0.2867) (45,0.2947) (50,0.2742)};
        \addlegendentry{AWVC}
        \addplot+[mark=*, mark size=1.5pt, color=papercolormrule3, thick] coordinates {(5,0.1560) (10,0.0970) (15,0.0733) (20,0.0555) (25,0.0684) (30,0.0480) (35,0.0449) (40,0.0495) (45,0.0416) (50,0.0366)};
        \addlegendentry{DSS $\circ$ AWVC}
        \addplot+[mark=*, mark size=1.5pt, color=papercolormrule4, thick] coordinates {(5,0.2320) (10,0.1680) (15,0.1627) (20,0.1210) (25,0.1508) (30,0.1933) (35,0.1429) (40,0.1760) (45,0.1204) (50,0.1596)};
        \addlegendentry{DSS}
        \addplot+[mark=*, mark size=1.5pt, color=papercolormrule5, thick] coordinates {(5,0.1940) (10,0.2600) (15,0.3247) (20,0.3465) (25,0.3676) (30,0.3867) (35,0.3614) (40,0.3917) (45,0.4067) (50,0.3700)};
        \addlegendentry{ARU}
        \addplot+[mark=*, mark size=1.5pt, color=papercolormrule6, thick] coordinates {(5,0.2040) (10,0.1000) (15,0.0687) (20,0.0515) (25,0.0404) (30,0.0343) (35,0.0303) (40,0.0272) (45,0.0229) (50,0.0206)};
        \addlegendentry{MMOD}
    \end{axis}
\end{tikzpicture}

%% file: Current_Work/Neurips/plots/m_sweep/m_sweep_submajorityviolations.tex
%
%
\definecolor{papercolormrule0}{RGB}{31,119,180}
\definecolor{papercolormrule1}{RGB}{255,127,14}
\definecolor{papercolormrule2}{RGB}{44,160,44}
\definecolor{papercolormrule3}{RGB}{214,39,40}
\definecolor{papercolormrule4}{RGB}{148,103,189}
\definecolor{papercolormrule5}{RGB}{140,86,75}
\definecolor{papercolormrule6}{RGB}{227,119,194}

\begin{tikzpicture}
    \begin{axis}[
        width=\linewidth, height=5cm,
        xlabel={$m$}, ylabel={frac.\ sub-maj.},
        xtick={5,10,15,20,25,30,35,40,45,50}, x tick label style={font=\footnotesize},
        y tick label style={font=\footnotesize},
        ylabel style={font=\small},
        ymin=0, ymax=1.05, 
        ymajorgrids=true, grid style={dashed,gray!30},
        legend style={
            at={(0.5,-0.30)}, anchor=north,
            legend columns=7, font=\footnotesize, draw=none,
        },
    ]
        \addplot+[mark=*, mark size=1.5pt, color=papercolormrule0, thick] coordinates {(5,0.0000) (10,0.0000) (15,0.0000) (20,0.0000) (25,0.0000) (30,0.0000) (35,0.0000) (40,0.0000) (45,0.0000) (50,0.0000)};
        \addlegendentry{Maj. Threshold}
        \addplot+[mark=*, mark size=1.5pt, color=papercolormrule1, thick] coordinates {(5,0.5500) (10,0.7900) (15,0.7200) (20,0.7700) (25,0.7900) (30,0.7500) (35,0.7800) (40,0.7400) (45,0.7200) (50,0.8100)};
        \addlegendentry{DVC}
        \addplot+[mark=*, mark size=1.5pt, color=papercolormrule2, thick] coordinates {(5,0.0000) (10,0.0000) (15,0.0000) (20,0.0000) (25,0.0000) (30,0.0000) (35,0.0000) (40,0.0000) (45,0.0000) (50,0.0000)};
        \addlegendentry{AWVC}
        \addplot+[mark=*, mark size=1.5pt, color=papercolormrule3, thick] coordinates {(5,0.0000) (10,0.0000) (15,0.0000) (20,0.0000) (25,0.0000) (30,0.0000) (35,0.0000) (40,0.0000) (45,0.0000) (50,0.0000)};
        \addlegendentry{DSS $\circ$ AWVC}
        \addplot+[mark=*, mark size=1.5pt, color=papercolormrule4, thick] coordinates {(5,0.1800) (10,0.1600) (15,0.2000) (20,0.1600) (25,0.1600) (30,0.2400) (35,0.1700) (40,0.1900) (45,0.1500) (50,0.2200)};
        \addlegendentry{DSS}
        \addplot+[mark=*, mark size=1.5pt, color=papercolormrule5, thick] coordinates {(5,0.0100) (10,0.0300) (15,0.0600) (20,0.1300) (25,0.1400) (30,0.2700) (35,0.2200) (40,0.2400) (45,0.2300) (50,0.3300)};
        \addlegendentry{ARU}
        \addplot+[mark=*, mark size=1.5pt, color=papercolormrule6, thick] coordinates {(5,0.1400) (10,0.1100) (15,0.1200) (20,0.1000) (25,0.0400) (30,0.0800) (35,0.0600) (40,0.0800) (45,0.1000) (50,0.1100)};
        \addlegendentry{MMOD}
    \end{axis}
\end{tikzpicture}

\par\smallskip

\begin{tikzpicture}
    \begin{axis}[
        width=\linewidth, height=5cm,
        xlabel={$m$}, ylabel={med.\ min app.},
        xtick={5,10,15,20,25,30,35,40,45,50}, x tick label style={font=\footnotesize},
        y tick label style={font=\footnotesize},
        ylabel style={font=\small},
        ymin=0, ymax=1.05, 
        ymajorgrids=true, grid style={dashed,gray!30},
        legend style={
            at={(0.5,-0.30)}, anchor=north,
            legend columns=7, font=\footnotesize, draw=none,
        },
    ]
        \addplot+[mark=*, mark size=1.5pt, color=papercolormrule0, thick] coordinates {(5,0.5200) (10,0.5200) (15,0.5200) (20,0.5000) (25,0.5000) (30,0.5000) (35,0.5000) (40,0.5000) (45,0.5000) (50,0.5000)};
        \addlegendentry{Maj. Threshold}
        \addplot+[mark=*, mark size=1.5pt, color=papercolormrule1, thick] coordinates {(5,0.4800) (10,0.4200) (15,0.4400) (20,0.4200) (25,0.4000) (30,0.4200) (35,0.4000) (40,0.4300) (45,0.4200) (50,0.4000)};
        \addlegendentry{DVC}
        \addplot+[mark=*, mark size=1.5pt, color=papercolormrule2, thick] coordinates {(5,0.6200) (10,0.5800) (15,0.5800) (20,0.5600) (25,0.5600) (30,0.5600) (35,0.5800) (40,0.5600) (45,0.5700) (50,0.5600)};
        \addlegendentry{AWVC}
        \addplot+[mark=*, mark size=1.5pt, color=papercolormrule3, thick] coordinates {(5,0.6400) (10,0.6700) (15,0.7000) (20,0.7600) (25,0.8000) (30,0.8200) (35,0.8200) (40,0.8200) (45,0.8200) (50,0.8600)};
        \addlegendentry{DSS $\circ$ AWVC}
        \addplot+[mark=*, mark size=1.5pt, color=papercolormrule4, thick] coordinates {(5,0.6200) (10,0.6500) (15,0.6800) (20,0.7400) (25,0.7500) (30,0.7800) (35,0.7900) (40,0.7900) (45,0.7700) (50,0.8000)};
        \addlegendentry{DSS}
        \addplot+[mark=*, mark size=1.5pt, color=papercolormrule5, thick] coordinates {(5,0.6200) (10,0.5600) (15,0.5600) (20,0.5400) (25,0.5400) (30,0.5200) (35,0.5200) (40,0.5200) (45,0.5200) (50,0.5200)};
        \addlegendentry{ARU}
        \addplot+[mark=*, mark size=1.5pt, color=papercolormrule6, thick] coordinates {(5,0.6200) (10,0.6400) (15,0.6800) (20,0.7600) (25,0.7600) (30,0.7900) (35,0.8200) (40,0.8000) (45,0.8000) (50,0.8200)};
        \addlegendentry{MMOD}
    \end{axis}
\end{tikzpicture}

%% file: Current_Work/Neurips/plots/m_sweep/m_sweep_minority_protection.tex
%
%
\definecolor{papercolormrule0}{RGB}{31,119,180}
\definecolor{papercolormrule1}{RGB}{255,127,14}
\definecolor{papercolormrule2}{RGB}{44,160,44}
\definecolor{papercolormrule3}{RGB}{214,39,40}
\definecolor{papercolormrule4}{RGB}{148,103,189}
\definecolor{papercolormrule5}{RGB}{140,86,75}
\definecolor{papercolormrule6}{RGB}{227,119,194}

\begin{tikzpicture}
    \begin{axis}[
        width=\linewidth, height=5cm,
        xlabel={$m$}, ylabel={frac.\ $\omega^*<1$},
        xtick={5,10,15,20,25,30,35,40,45,50}, x tick label style={font=\footnotesize},
        y tick label style={font=\footnotesize},
        ylabel style={font=\small},
        ymin=0, ymax=1.05, 
        ymajorgrids=true, grid style={dashed,gray!30},
        legend style={
            at={(0.5,-0.30)}, anchor=north,
            legend columns=7, font=\footnotesize, draw=none,
        },
    ]
        \addplot+[mark=*, mark size=1.5pt, color=papercolormrule0, thick] coordinates {(5,0.2100) (10,0.3100) (15,0.3400) (20,0.3900) (25,0.3400) (30,0.4100) (35,0.3700) (40,0.3500) (45,0.3900) (50,0.3500)};
        \addlegendentry{Maj. Threshold}
        \addplot+[mark=*, mark size=1.5pt, color=papercolormrule1, thick] coordinates {(5,0.0000) (10,0.0000) (15,0.0000) (20,0.0000) (25,0.0000) (30,0.0000) (35,0.0000) (40,0.0000) (45,0.0000) (50,0.0000)};
        \addlegendentry{DVC}
        \addplot+[mark=*, mark size=1.5pt, color=papercolormrule2, thick] coordinates {(5,0.0000) (10,0.0000) (15,0.0000) (20,0.0000) (25,0.0000) (30,0.0000) (35,0.0000) (40,0.0000) (45,0.0000) (50,0.0000)};
        \addlegendentry{AWVC}
        \addplot+[mark=*, mark size=1.5pt, color=papercolormrule3, thick] coordinates {(5,0.0000) (10,0.0000) (15,0.0000) (20,0.0000) (25,0.0000) (30,0.0000) (35,0.0000) (40,0.0000) (45,0.0000) (50,0.0000)};
        \addlegendentry{DSS $\circ$ AWVC}
        \addplot+[mark=*, mark size=1.5pt, color=papercolormrule4, thick] coordinates {(5,0.0300) (10,0.0700) (15,0.1000) (20,0.0900) (25,0.1200) (30,0.1700) (35,0.1400) (40,0.1400) (45,0.0900) (50,0.1500)};
        \addlegendentry{DSS}
        \addplot+[mark=*, mark size=1.5pt, color=papercolormrule5, thick] coordinates {(5,0.0000) (10,0.0300) (15,0.1700) (20,0.1800) (25,0.2500) (30,0.2800) (35,0.2500) (40,0.2600) (45,0.2500) (50,0.2600)};
        \addlegendentry{ARU}
        \addplot+[mark=*, mark size=1.5pt, color=papercolormrule6, thick] coordinates {(5,0.0300) (10,0.0000) (15,0.0100) (20,0.0500) (25,0.0300) (30,0.0400) (35,0.0400) (40,0.0300) (45,0.0200) (50,0.0400)};
        \addlegendentry{MMOD}
    \end{axis}
\end{tikzpicture}

\par\smallskip

\begin{tikzpicture}
    \begin{axis}[
        width=\linewidth, height=5cm,
        xlabel={$m$}, ylabel={med.\ min $\omega^*$},
        xtick={5,10,15,20,25,30,35,40,45,50}, x tick label style={font=\footnotesize},
        y tick label style={font=\footnotesize},
        ylabel style={font=\small},
        ymin=0, 
        ymajorgrids=true, grid style={dashed,gray!30},
        legend style={
            at={(0.5,-0.30)}, anchor=north,
            legend columns=7, font=\footnotesize, draw=none,
        },
    ]
        \addplot+[mark=*, mark size=1.5pt, color=papercolormrule0, thick] coordinates {(5,1.4286) (10,1.5000) (15,1.2121) (20,1.2500) (25,1.2772) (30,1.2599) (35,1.2245) (40,1.1979) (45,1.1111) (50,1.1600)};
        \addlegendentry{Maj. Threshold}
        \addplot+[mark=*, mark size=1.5pt, color=papercolormrule1, thick] coordinates {(5,1.2500) (10,1.0620) (15,1.1111) (20,1.1538) (25,1.0000) (30,1.0714) (35,1.0714) (40,1.0714) (45,1.0581) (50,1.0000)};
        \addlegendentry{DVC}
        \addplot+[mark=*, mark size=1.5pt, color=papercolormrule2, thick] coordinates {(5,2.2222) (10,2.0833) (15,1.9444) (20,2.0000) (25,2.0000) (30,1.8333) (35,1.9255) (40,1.9318) (45,1.8922) (50,1.9000)};
        \addlegendentry{AWVC}
        \addplot+[mark=*, mark size=1.5pt, color=papercolormrule3, thick] coordinates {(5,2.5658) (10,2.6316) (15,3.3333) (20,3.8462) (25,3.9231) (30,5.5556) (35,5.5556) (40,5.5556) (45,5.6250) (50,7.1429)};
        \addlegendentry{DSS $\circ$ AWVC}
        \addplot+[mark=*, mark size=1.5pt, color=papercolormrule4, thick] coordinates {(5,2.0000) (10,2.5000) (15,2.7864) (20,3.4524) (25,3.5714) (30,4.4231) (35,3.9054) (40,4.8611) (45,4.3561) (50,4.8611)};
        \addlegendentry{DSS}
        \addplot+[mark=*, mark size=1.5pt, color=papercolormrule5, thick] coordinates {(5,2.3529) (10,1.8750) (15,1.6667) (20,1.3889) (25,1.3880) (30,1.3194) (35,1.4286) (40,1.2500) (45,1.2560) (50,1.2174)};
        \addlegendentry{ARU}
        \addplot+[mark=*, mark size=1.5pt, color=papercolormrule6, thick] coordinates {(5,2.0417) (10,2.5000) (15,2.7047) (20,3.7981) (25,3.9231) (30,4.7009) (35,4.4104) (40,5.2083) (45,5.0505) (50,5.5556)};
        \addlegendentry{MMOD}
    \end{axis}
\end{tikzpicture}

%% file: Current_Work/Neurips/plots/m_sweep/m_sweep_unanimity.tex
%
%
\definecolor{papercolormrule0}{RGB}{31,119,180}
\definecolor{papercolormrule1}{RGB}{255,127,14}
\definecolor{papercolormrule2}{RGB}{44,160,44}
\definecolor{papercolormrule3}{RGB}{214,39,40}
\definecolor{papercolormrule4}{RGB}{148,103,189}
\definecolor{papercolormrule5}{RGB}{140,86,75}
\definecolor{papercolormrule6}{RGB}{227,119,194}

\begin{tikzpicture}
    \begin{axis}[
        width=\linewidth, height=5cm,
        xlabel={$m$}, ylabel={frac.\ non-unanim.},
        xtick={5,10,15,20,25,30,35,40,45,50}, x tick label style={font=\footnotesize},
        y tick label style={font=\footnotesize},
        ylabel style={font=\small},
        ymin=0, ymax=1.05, 
        ymajorgrids=true, grid style={dashed,gray!30},
        legend style={
            at={(0.5,-0.30)}, anchor=north,
            legend columns=7, font=\footnotesize, draw=none,
        },
    ]
        \addplot+[mark=*, mark size=1.5pt, color=papercolormrule0, thick] coordinates {(5,nan) (10,nan) (15,nan) (20,1.0000) (25,nan) (30,1.0000) (35,1.0000) (40,1.0000) (45,1.0000) (50,1.0000)};
        \addlegendentry{Maj. Threshold}
        \addplot+[mark=*, mark size=1.5pt, color=papercolormrule1, thick] coordinates {(5,nan) (10,nan) (15,nan) (20,1.0000) (25,nan) (30,1.0000) (35,1.0000) (40,1.0000) (45,1.0000) (50,1.0000)};
        \addlegendentry{DVC}
        \addplot+[mark=*, mark size=1.5pt, color=papercolormrule2, thick] coordinates {(5,nan) (10,nan) (15,nan) (20,1.0000) (25,nan) (30,1.0000) (35,1.0000) (40,1.0000) (45,1.0000) (50,1.0000)};
        \addlegendentry{AWVC}
        \addplot+[mark=*, mark size=1.5pt, color=papercolormrule3, thick] coordinates {(5,nan) (10,nan) (15,nan) (20,0.0000) (25,nan) (30,0.0000) (35,0.0000) (40,0.0000) (45,0.0000) (50,0.0000)};
        \addlegendentry{DSS $\circ$ AWVC}
        \addplot+[mark=*, mark size=1.5pt, color=papercolormrule4, thick] coordinates {(5,nan) (10,nan) (15,nan) (20,0.0000) (25,nan) (30,0.0000) (35,0.0000) (40,0.0000) (45,0.0000) (50,0.0000)};
        \addlegendentry{DSS}
        \addplot+[mark=*, mark size=1.5pt, color=papercolormrule5, thick] coordinates {(5,nan) (10,nan) (15,nan) (20,1.0000) (25,nan) (30,1.0000) (35,1.0000) (40,1.0000) (45,1.0000) (50,1.0000)};
        \addlegendentry{ARU}
        \addplot+[mark=*, mark size=1.5pt, color=papercolormrule6, thick] coordinates {(5,nan) (10,nan) (15,nan) (20,0.0000) (25,nan) (30,0.0000) (35,0.0000) (40,0.0000) (45,0.0000) (50,0.0000)};
        \addlegendentry{MMOD}
    \end{axis}
\end{tikzpicture}

\par\smallskip

\begin{tikzpicture}
    \begin{axis}[
        width=\linewidth, height=5cm,
        xlabel={$m$}, ylabel={tot.\ non-unanim.},
        xtick={5,10,15,20,25,30,35,40,45,50}, x tick label style={font=\footnotesize},
        y tick label style={font=\footnotesize},
        ylabel style={font=\small},
        ymin=0, 
        ymajorgrids=true, grid style={dashed,gray!30},
        legend style={
            at={(0.5,-0.30)}, anchor=north,
            legend columns=7, font=\footnotesize, draw=none,
        },
    ]
        \addplot+[mark=*, mark size=1.5pt, color=papercolormrule0, thick] coordinates {(5,0.0000) (10,0.0000) (15,0.0000) (20,10.0000) (25,0.0000) (30,28.0000) (35,68.0000) (40,183.0000) (45,168.0000) (50,123.0000)};
        \addlegendentry{Maj. Threshold}
        \addplot+[mark=*, mark size=1.5pt, color=papercolormrule1, thick] coordinates {(5,0.0000) (10,0.0000) (15,0.0000) (20,10.0000) (25,0.0000) (30,31.0000) (35,74.0000) (40,189.0000) (45,168.0000) (50,132.0000)};
        \addlegendentry{DVC}
        \addplot+[mark=*, mark size=1.5pt, color=papercolormrule2, thick] coordinates {(5,0.0000) (10,0.0000) (15,0.0000) (20,9.0000) (25,0.0000) (30,21.0000) (35,53.0000) (40,146.0000) (45,124.0000) (50,99.0000)};
        \addlegendentry{AWVC}
        \addplot+[mark=*, mark size=1.5pt, color=papercolormrule3, thick] coordinates {(5,0.0000) (10,0.0000) (15,0.0000) (20,0.0000) (25,0.0000) (30,0.0000) (35,0.0000) (40,0.0000) (45,0.0000) (50,0.0000)};
        \addlegendentry{DSS $\circ$ AWVC}
        \addplot+[mark=*, mark size=1.5pt, color=papercolormrule4, thick] coordinates {(5,0.0000) (10,0.0000) (15,0.0000) (20,0.0000) (25,0.0000) (30,0.0000) (35,0.0000) (40,0.0000) (45,0.0000) (50,0.0000)};
        \addlegendentry{DSS}
        \addplot+[mark=*, mark size=1.5pt, color=papercolormrule5, thick] coordinates {(5,0.0000) (10,0.0000) (15,0.0000) (20,10.0000) (25,0.0000) (30,27.0000) (35,65.0000) (40,172.0000) (45,159.0000) (50,116.0000)};
        \addlegendentry{ARU}
        \addplot+[mark=*, mark size=1.5pt, color=papercolormrule6, thick] coordinates {(5,0.0000) (10,0.0000) (15,0.0000) (20,0.0000) (25,0.0000) (30,0.0000) (35,0.0000) (40,0.0000) (45,0.0000) (50,0.0000)};
        \addlegendentry{MMOD}
    \end{axis}
\end{tikzpicture}

%% file: Current_Work/Neurips/plots/correlation_epsilon_omega/clamped_histogram_simple.tex
%
%
%
%
\definecolor{papercoldshsynthetic}{RGB}{31,119,180}
\definecolor{papercoldshkidney}{RGB}{255,127,14}
\definecolor{papercoldshmmachine}{RGB}{44,160,44}
\definecolor{papercoldshpolitics}{RGB}{214,39,40}
\definecolor{papercoldshcoval}{RGB}{148,103,189}

\begin{tikzpicture}
    \begin{axis}[
        ybar stacked,
        bar width=12pt,
        width=\linewidth, height=6cm,
        xlabel={candidates per election with $\varepsilon^*$ clamped (uncomputable)},
        ylabel={\# elections},
        xlabel style={font=\small}, ylabel style={font=\small},
        x tick label style={font=\footnotesize},
        y tick label style={font=\footnotesize},
        ymajorgrids=true, grid style={dashed,gray!30},
        symbolic x coords={b0, b1to8, b9to16, b17to24, b25to32, b33to40, b41to48, b49to56, b57to64, b65to67},
        xtick=data,
        xticklabels={{0}, {1--8}, {9--16}, {17--24}, {25--32}, {33--40}, {41--48}, {49--56}, {57--64}, {65--67}},
        enlarge x limits=0.06,
        legend style={
            at={(0.5,-0.22)}, anchor=north,
            legend columns=5, draw=none,
            font=\footnotesize,
            /tikz/every even column/.append style={column sep=6pt},
        },
    ]
        \addplot+[ybar, fill=papercoldshsynthetic, draw=papercoldshsynthetic!70!black, line width=0.3pt] table[x=bin, y=synthetic] {{Current_Work/Neurips/plots/correlation_epsilon_omega/clamped_histogram_simple.dat}};
        \addplot+[ybar, fill=papercoldshkidney, draw=papercoldshkidney!70!black, line width=0.3pt] table[x=bin, y=kidney] {{Current_Work/Neurips/plots/correlation_epsilon_omega/clamped_histogram_simple.dat}};
        \addplot+[ybar, fill=papercoldshmmachine, draw=papercoldshmmachine!70!black, line width=0.3pt] table[x=bin, y=mmachine] {{Current_Work/Neurips/plots/correlation_epsilon_omega/clamped_histogram_simple.dat}};
        \addplot+[ybar, fill=papercoldshpolitics, draw=papercoldshpolitics!70!black, line width=0.3pt] table[x=bin, y=politics] {{Current_Work/Neurips/plots/correlation_epsilon_omega/clamped_histogram_simple.dat}};
        \addplot+[ybar, fill=papercoldshcoval, draw=papercoldshcoval!70!black, line width=0.3pt] table[x=bin, y=coval] {{Current_Work/Neurips/plots/correlation_epsilon_omega/clamped_histogram_simple.dat}};
        \legend{Synthetic ($n=600$), Kidney ($n=60$), Moral Machine ($n=60$), Political ($n=9$), CoVal ($n=319$)}
    \end{axis}
\end{tikzpicture}

%% file: Current_Work/Neurips/plots/correlation_epsilon_omega/pooled_scatter_simple.tex
%
%
%
%
\definecolor{papercoldssynthetic}{RGB}{31,119,180}
\definecolor{papercoldskidney}{RGB}{255,127,14}
\definecolor{papercoldsmmachine}{RGB}{44,160,44}
\definecolor{papercoldspolitics}{RGB}{214,39,40}
\definecolor{papercoldscoval}{RGB}{148,103,189}

\begin{tikzpicture}
    \begin{axis}[
        width=\linewidth, height=7cm,
        xlabel={critical $\omega^*$ (high $=$ robust)},
        ylabel={critical $\varepsilon^*$ ($>0$: blocked)},
        xlabel style={font=\small}, ylabel style={font=\small},
        x tick label style={font=\footnotesize},
        y tick label style={font=\footnotesize},
        ymajorgrids=true, grid style={dashed,gray!30},
        scatter/classes={
            synthetic={mark=*, mark size=1.2pt, papercoldssynthetic, opacity=0.55},
            kidney={mark=*, mark size=1.2pt, papercoldskidney, opacity=0.55},
            mmachine={mark=*, mark size=1.2pt, papercoldsmmachine, opacity=0.55},
            politics={mark=*, mark size=1.2pt, papercoldspolitics, opacity=0.55},
            coval={mark=*, mark size=1.2pt, papercoldscoval, opacity=0.55}
        },
        legend style={
            at={(0.5,-0.20)}, anchor=north,
            legend columns=5, draw=none,
            font=\footnotesize,
            /tikz/every even column/.append style={column sep=6pt},
        },
    ]
        \addplot[scatter, only marks, scatter src=explicit symbolic]
            table[meta=source] {{Current_Work/Neurips/plots/correlation_epsilon_omega/pooled_scatter_simple_downsampled.dat}};
        \addplot[gray, dashed, thick, forget plot]
            coordinates {(-0.346,0) (7.499,0)};
        \legend{Synthetic ($n=2578$), Kidney ($n=2294$), Moral Machine ($n=3293$), Political ($n=61$), CoVal ($n=685$)}
    \end{axis}
\end{tikzpicture}

%% file: Current_Work/Neurips/plots/Synthetic_MoralVariation_Plots/synthetic_selection_fraction.tex
%
%
\definecolor{papercolord1rule0}{RGB}{31,119,180}
\definecolor{papercolord1rule1}{RGB}{255,127,14}
\definecolor{papercolord1rule2}{RGB}{44,160,44}
\definecolor{papercolord1rule3}{RGB}{214,39,40}
\definecolor{papercolord1rule4}{RGB}{148,103,189}
\definecolor{papercolord1rule5}{RGB}{140,86,75}
\definecolor{papercolord1rule6}{RGB}{227,119,194}

\begin{tikzpicture}
    \begin{axis}[
        ybar=0.5pt,
        bar width=5pt,
        width=\linewidth,
        height=5cm,
        ymin=0, ymax=1.0,
        ylabel={avg $|f(\mathcal{I})|/m$ across elections},
        ylabel style={font=\small},
        symbolic x coords={IC,Mallows $\phi$=0.2,Mallows $\phi$=0.5,Mallows $\phi$=0.8,Spatial 1D,Spatial 2D},
        xtick=data,
        x tick label style={font=\small},
        y tick label style={font=\footnotesize},
        ymajorgrids=true,
        grid style={dashed,gray!30},
        legend style={
            at={(0.5,-0.18)},
            anchor=north,
            legend columns=7,
            /tikz/every even column/.append style={column sep=4pt},
            font=\footnotesize,
            draw=none,
        },
        enlarge x limits=0.12,
    ]
        \addplot+[ybar, fill=papercolord1rule0, draw=black, line width=0.3pt] coordinates {(IC,0.5540) (Mallows $\phi$=0.2,0.4940) (Mallows $\phi$=0.5,0.5520) (Mallows $\phi$=0.8,0.5660) (Spatial 1D,0.5840) (Spatial 2D,0.1160)};
        \addlegendentry{Maj. Threshold}
        \addplot+[ybar, fill=papercolord1rule1, draw=black, line width=0.3pt] coordinates {(IC,0.7960) (Mallows $\phi$=0.2,0.6040) (Mallows $\phi$=0.5,0.6740) (Mallows $\phi$=0.8,0.7180) (Spatial 1D,0.4460) (Spatial 2D,0.4860)};
        \addlegendentry{DVC}
        \addplot+[ybar, fill=papercolord1rule2, draw=black, line width=0.3pt] coordinates {(IC,0.1400) (Mallows $\phi$=0.2,0.4020) (Mallows $\phi$=0.5,0.4460) (Mallows $\phi$=0.8,0.2140) (Spatial 1D,0.2200) (Spatial 2D,0.0760)};
        \addlegendentry{AWVC}
        \addplot+[ybar, fill=papercolord1rule3, draw=black, line width=0.3pt] coordinates {(IC,0.1120) (Mallows $\phi$=0.2,0.2000) (Mallows $\phi$=0.5,0.2020) (Mallows $\phi$=0.8,0.1600) (Spatial 1D,0.1580) (Spatial 2D,0.0640)};
        \addlegendentry{DSS $\circ$ AWVC}
        \addplot+[ybar, fill=papercolord1rule4, draw=black, line width=0.3pt] coordinates {(IC,0.3400) (Mallows $\phi$=0.2,0.2000) (Mallows $\phi$=0.5,0.2020) (Mallows $\phi$=0.8,0.2400) (Spatial 1D,0.2100) (Spatial 2D,0.2240)};
        \addlegendentry{DSS}
        \addplot+[ybar, fill=papercolord1rule5, draw=black, line width=0.3pt] coordinates {(IC,0.1540) (Mallows $\phi$=0.2,0.4120) (Mallows $\phi$=0.5,0.3960) (Mallows $\phi$=0.8,0.2120) (Spatial 1D,0.0540) (Spatial 2D,0.0060)};
        \addlegendentry{ARU}
        \addplot+[ybar, fill=papercolord1rule6, draw=black, line width=0.3pt] coordinates {(IC,0.2020) (Mallows $\phi$=0.2,0.2000) (Mallows $\phi$=0.5,0.2000) (Mallows $\phi$=0.8,0.2020) (Spatial 1D,0.2020) (Spatial 2D,0.2000)};
        \addlegendentry{MMOD}
    \end{axis}
\end{tikzpicture}

%% file: Current_Work/Neurips/plots/Synthetic_MoralVariation_Plots/synthetic_submajority.tex
\begin{tabular}{lrrrrrrrrrrrr}
\toprule
 & \multicolumn{2}{c}{IC} & \multicolumn{2}{c}{Mallows $\phi$=0.2} & \multicolumn{2}{c}{Mallows $\phi$=0.5} & \multicolumn{2}{c}{Mallows $\phi$=0.8} & \multicolumn{2}{c}{Spatial 1D} & \multicolumn{2}{c}{Spatial 2D} \\
 & \multicolumn{2}{c}{($N$=100)} & \multicolumn{2}{c}{($N$=100)} & \multicolumn{2}{c}{($N$=100)} & \multicolumn{2}{c}{($N$=100)} & \multicolumn{2}{c}{($N$=100)} & \multicolumn{2}{c}{($N$=100)} \\
\cmidrule(lr){2-3}\cmidrule(lr){4-5}\cmidrule(lr){6-7}\cmidrule(lr){8-9}\cmidrule(lr){10-11}\cmidrule(lr){12-13}
Rule & \# & med min app. & \# & med min app. & \# & med min app. & \# & med min app. & \# & med min app. & \# & med min app. \\
\midrule
Maj. Threshold & 0 & 0.52 & 0 & 0.58 & 0 & 0.54 & 0 & 0.52 & 0 & 0.54 & 0 & 0.52 \\
DVC & 71 & 0.46 & 53 & 0.48 & 50 & 0.49 & 60 & 0.48 & 23 & 0.54 & 90 & 0.36 \\
AWVC & 0 & 0.56 & 0 & 0.66 & 0 & 0.60 & 0 & 0.57 & 0 & 0.60 & 0 & 0.54 \\
DSS $\circ$ AWVC & 0 & 0.58 & 0 & 0.78 & 0 & 0.70 & 0 & 0.60 & 0 & 0.62 & 0 & 0.54 \\
DSS & 25 & 0.56 & 0 & 0.78 & 0 & 0.70 & 10 & 0.58 & 8 & 0.60 & 63 & 0.46 \\
ARU & 1 & 0.58 & 2 & 0.65 & 0 & 0.62 & 0 & 0.58 & 0 & 0.62 & 0 & 0.62 \\
MMOD & 16 & 0.56 & 0 & 0.78 & 0 & 0.70 & 7 & 0.57 & 8 & 0.60 & 64 & 0.46 \\
\bottomrule
\end{tabular}

%% file: Current_Work/Neurips/plots/Synthetic_MoralVariation_Plots/synthetic_minority_protection.tex
\begin{tabular}{lrrrrrrrrrrrr}
\toprule
 & \multicolumn{2}{c}{IC} & \multicolumn{2}{c}{Mallows $\phi$=0.2} & \multicolumn{2}{c}{Mallows $\phi$=0.5} & \multicolumn{2}{c}{Mallows $\phi$=0.8} & \multicolumn{2}{c}{Spatial 1D} & \multicolumn{2}{c}{Spatial 2D} \\
 & \multicolumn{2}{c}{($N$=100)} & \multicolumn{2}{c}{($N$=100)} & \multicolumn{2}{c}{($N$=100)} & \multicolumn{2}{c}{($N$=100)} & \multicolumn{2}{c}{($N$=100)} & \multicolumn{2}{c}{($N$=100)} \\
\cmidrule(lr){2-3}\cmidrule(lr){4-5}\cmidrule(lr){6-7}\cmidrule(lr){8-9}\cmidrule(lr){10-11}\cmidrule(lr){12-13}
Rule & \# & med $\omega^*$ & \# & med $\omega^*$ & \# & med $\omega^*$ & \# & med $\omega^*$ & \# & med $\omega^*$ & \# & med $\omega^*$ \\
\midrule
Maj. Threshold & 23 & 1.11 & 1 & 1.87 & 0 & 1.87 & 21 & 1.11 & 68 & 0.67 & 2 & 2.00 \\
DVC & 0 & 1.00 & 0 & 1.30 & 0 & 1.43 & 0 & 1.11 & 0 & 1.30 & 0 & 1.25 \\
AWVC & 0 & 2.00 & 0 & 2.65 & 0 & 2.18 & 0 & 2.08 & 0 & 2.08 & 0 & 2.08 \\
DSS $\circ$ AWVC & 0 & 2.00 & 0 & 4.55 & 0 & 3.33 & 0 & 2.27 & 0 & 2.17 & 0 & 2.14 \\
DSS & 9 & 1.67 & 0 & 4.55 & 0 & 3.33 & 0 & 2.00 & 8 & 2.00 & 7 & 1.78 \\
ARU & 1 & 1.43 & 0 & 2.50 & 0 & 2.35 & 0 & 1.83 & 2 & 1.74 & 0 & 2.63 \\
MMOD & 1 & 1.67 & 0 & 4.55 & 0 & 3.33 & 0 & 2.08 & 11 & 2.00 & 6 & 1.74 \\
\bottomrule
\end{tabular}

%% file: Current_Work/Neurips/plots/Synthetic_MoralVariation_Plots/moral_selectivity.tex
%
%
\definecolor{papercolord1rule0}{RGB}{31,119,180}
\definecolor{papercolord1rule1}{RGB}{255,127,14}
\definecolor{papercolord1rule2}{RGB}{44,160,44}
\definecolor{papercolord1rule3}{RGB}{214,39,40}
\definecolor{papercolord1rule4}{RGB}{148,103,189}
\definecolor{papercolord1rule5}{RGB}{140,86,75}
\definecolor{papercolord1rule6}{RGB}{227,119,194}

\begin{tikzpicture}
    \begin{axis}[
        ybar=0.5pt,
        bar width=5pt,
        width=\linewidth,
        height=5cm,
        ymin=0, ymax=1.0,
        ylabel={avg $f|(\mathcal{I})|/m$ across elections},
        ylabel style={font=\small},
        symbolic x coords={Opt.\,1 $\alpha$=0.20,Opt.\,1 $\alpha$=0.30,Opt.\,1 $\alpha$=0.40,Opt.\,2 $q$=0.25,Opt.\,2 $q$=0.50,Opt.\,2 $q$=0.75},
        xtick=data,
        x tick label style={font=\small},
        y tick label style={font=\footnotesize},
        ymajorgrids=true,
        grid style={dashed,gray!30},
        legend style={
            at={(0.5,-0.18)},
            anchor=north,
            legend columns=7,
            /tikz/every even column/.append style={column sep=4pt},
            font=\footnotesize,
            draw=none,
        },
        enlarge x limits=0.12,
    ]
        \addplot+[ybar, fill=papercolord1rule0, draw=black, line width=0.3pt] coordinates {(Opt.\,1 $\alpha$=0.20,0.1950) (Opt.\,1 $\alpha$=0.30,0.3020) (Opt.\,1 $\alpha$=0.40,0.3970) (Opt.\,2 $q$=0.25,0.4880) (Opt.\,2 $q$=0.50,0.4870) (Opt.\,2 $q$=0.75,0.4670)};
        \addlegendentry{Maj. Threshold}
        \addplot+[ybar, fill=papercolord1rule1, draw=black, line width=0.3pt] coordinates {(Opt.\,1 $\alpha$=0.20,0.3010) (Opt.\,1 $\alpha$=0.30,0.3740) (Opt.\,1 $\alpha$=0.40,0.4550) (Opt.\,2 $q$=0.25,0.5430) (Opt.\,2 $q$=0.50,0.5280) (Opt.\,2 $q$=0.75,0.5270)};
        \addlegendentry{DVC}
        \addplot+[ybar, fill=papercolord1rule2, draw=black, line width=0.3pt] coordinates {(Opt.\,1 $\alpha$=0.20,0.1910) (Opt.\,1 $\alpha$=0.30,0.2920) (Opt.\,1 $\alpha$=0.40,0.3830) (Opt.\,2 $q$=0.25,0.4680) (Opt.\,2 $q$=0.50,0.4700) (Opt.\,2 $q$=0.75,0.4490)};
        \addlegendentry{AWVC}
        \addplot+[ybar, fill=papercolord1rule3, draw=black, line width=0.3pt] coordinates {(Opt.\,1 $\alpha$=0.20,0.0320) (Opt.\,1 $\alpha$=0.30,0.0670) (Opt.\,1 $\alpha$=0.40,0.1060) (Opt.\,2 $q$=0.25,0.1660) (Opt.\,2 $q$=0.50,0.1660) (Opt.\,2 $q$=0.75,0.1630)};
        \addlegendentry{DSS $\circ$ AWVC}
        \addplot+[ybar, fill=papercolord1rule4, draw=black, line width=0.3pt] coordinates {(Opt.\,1 $\alpha$=0.20,0.0320) (Opt.\,1 $\alpha$=0.30,0.0670) (Opt.\,1 $\alpha$=0.40,0.1060) (Opt.\,2 $q$=0.25,0.1660) (Opt.\,2 $q$=0.50,0.1660) (Opt.\,2 $q$=0.75,0.1630)};
        \addlegendentry{DSS}
        \addplot+[ybar, fill=papercolord1rule5, draw=black, line width=0.3pt] coordinates {(Opt.\,1 $\alpha$=0.20,0.1850) (Opt.\,1 $\alpha$=0.30,0.2860) (Opt.\,1 $\alpha$=0.40,0.3900) (Opt.\,2 $q$=0.25,0.4780) (Opt.\,2 $q$=0.50,0.4760) (Opt.\,2 $q$=0.75,0.4580)};
        \addlegendentry{ARU}
        \addplot+[ybar, fill=papercolord1rule6, draw=black, line width=0.3pt] coordinates {(Opt.\,1 $\alpha$=0.20,0.0320) (Opt.\,1 $\alpha$=0.30,0.0670) (Opt.\,1 $\alpha$=0.40,0.1060) (Opt.\,2 $q$=0.25,0.1660) (Opt.\,2 $q$=0.50,0.1660) (Opt.\,2 $q$=0.75,0.1630)};
        \addlegendentry{MMOD}
    \end{axis}
\end{tikzpicture}

%% file: Current_Work/Neurips/plots/Synthetic_MoralVariation_Plots/moral_submajority.tex
\begin{tabular}{lrrrrrrrrrrrr}
\toprule
 & \multicolumn{2}{c}{Opt.\,1 $\alpha$=0.20} & \multicolumn{2}{c}{Opt.\,1 $\alpha$=0.30} & \multicolumn{2}{c}{Opt.\,1 $\alpha$=0.40} & \multicolumn{2}{c}{Opt.\,2 $q$=0.25} & \multicolumn{2}{c}{Opt.\,2 $q$=0.50} & \multicolumn{2}{c}{Opt.\,2 $q$=0.75} \\
 & \multicolumn{2}{c}{($N$=10)} & \multicolumn{2}{c}{($N$=10)} & \multicolumn{2}{c}{($N$=10)} & \multicolumn{2}{c}{($N$=10)} & \multicolumn{2}{c}{($N$=10)} & \multicolumn{2}{c}{($N$=10)} \\
\cmidrule(lr){2-3}\cmidrule(lr){4-5}\cmidrule(lr){6-7}\cmidrule(lr){8-9}\cmidrule(lr){10-11}\cmidrule(lr){12-13}
Rule & \# & med min app. & \# & med min app. & \# & med min app. & \# & med min app. & \# & med min app. & \# & med min app. \\
\midrule
Maj. Threshold & 0 & 0.55 & 0 & 0.51 & 0 & 0.53 & 0 & 0.53 & 0 & 0.54 & 0 & 0.51 \\
DVC & 10 & 0.10 & 10 & 0.18 & 10 & 0.28 & 10 & 0.33 & 10 & 0.34 & 10 & 0.34 \\
AWVC & 0 & 0.59 & 0 & 0.56 & 0 & 0.58 & 0 & 0.66 & 0 & 0.66 & 0 & 0.62 \\
DSS $\circ$ AWVC & 0 & 1.00 & 0 & 1.00 & 0 & 1.00 & 0 & 1.00 & 0 & 1.00 & 0 & 1.00 \\
DSS & 0 & 1.00 & 0 & 1.00 & 0 & 1.00 & 0 & 1.00 & 0 & 1.00 & 0 & 1.00 \\
ARU & 0 & 0.61 & 0 & 0.58 & 1 & 0.57 & 1 & 0.57 & 1 & 0.61 & 2 & 0.53 \\
MMOD & 0 & 1.00 & 0 & 1.00 & 0 & 1.00 & 0 & 1.00 & 0 & 1.00 & 0 & 1.00 \\
\bottomrule
\end{tabular}

%% file: Current_Work/Neurips/plots/Synthetic_MoralVariation_Plots/moral_minority_protection.tex
\begin{tabular}{lrrrrrrrrrrrr}
\toprule
 & \multicolumn{2}{c}{Opt.\,1 $\alpha$=0.20} & \multicolumn{2}{c}{Opt.\,1 $\alpha$=0.30} & \multicolumn{2}{c}{Opt.\,1 $\alpha$=0.40} & \multicolumn{2}{c}{Opt.\,2 $q$=0.25} & \multicolumn{2}{c}{Opt.\,2 $q$=0.50} & \multicolumn{2}{c}{Opt.\,2 $q$=0.75} \\
 & \multicolumn{2}{c}{($N$=10)} & \multicolumn{2}{c}{($N$=10)} & \multicolumn{2}{c}{($N$=10)} & \multicolumn{2}{c}{($N$=10)} & \multicolumn{2}{c}{($N$=10)} & \multicolumn{2}{c}{($N$=10)} \\
\cmidrule(lr){2-3}\cmidrule(lr){4-5}\cmidrule(lr){6-7}\cmidrule(lr){8-9}\cmidrule(lr){10-11}\cmidrule(lr){12-13}
Rule & \# & med $\omega^*$ & \# & med $\omega^*$ & \# & med $\omega^*$ & \# & med $\omega^*$ & \# & med $\omega^*$ & \# & med $\omega^*$ \\
\midrule
Maj. Threshold & 0 & 2.02 & 0 & 1.78 & 0 & 1.69 & 0 & 1.49 & 0 & 1.49 & 0 & 1.44 \\
DVC & 0 & 1.01 & 0 & 1.03 & 0 & 1.03 & 0 & 1.04 & 0 & 1.03 & 0 & 1.06 \\
AWVC & 0 & 2.21 & 0 & 1.97 & 0 & 1.86 & 0 & 2.02 & 0 & 1.96 & 0 & 1.93 \\
DSS $\circ$ AWVC & 0 & $\infty$ & 0 & $\infty$ & 0 & $\infty$ & 0 & $\infty$ & 0 & $\infty$ & 0 & $\infty$ \\
DSS & 0 & $\infty$ & 0 & $\infty$ & 0 & $\infty$ & 0 & $\infty$ & 0 & $\infty$ & 0 & $\infty$ \\
ARU & 0 & 2.31 & 0 & 2.00 & 0 & 1.80 & 0 & 1.59 & 0 & 1.89 & 0 & 1.52 \\
MMOD & 0 & $\infty$ & 0 & $\infty$ & 0 & $\infty$ & 0 & $\infty$ & 0 & $\infty$ & 0 & $\infty$ \\
\bottomrule
\end{tabular}

%% file: Current_Work/Neurips/plots/Synthetic_MoralVariation_Plots/moral_unanimity.tex
\begin{tabular}{lrrrrrrrrrrrr}
\toprule
 & \multicolumn{2}{c}{Opt.\,1 $\alpha$=0.20} & \multicolumn{2}{c}{Opt.\,1 $\alpha$=0.30} & \multicolumn{2}{c}{Opt.\,1 $\alpha$=0.40} & \multicolumn{2}{c}{Opt.\,2 $q$=0.25} & \multicolumn{2}{c}{Opt.\,2 $q$=0.50} & \multicolumn{2}{c}{Opt.\,2 $q$=0.75} \\
 & \multicolumn{2}{c}{($N$=10, $N_u$=10)} & \multicolumn{2}{c}{($N$=10, $N_u$=10)} & \multicolumn{2}{c}{($N$=10, $N_u$=10)} & \multicolumn{2}{c}{($N$=10, $N_u$=10)} & \multicolumn{2}{c}{($N$=10, $N_u$=10)} & \multicolumn{2}{c}{($N$=10, $N_u$=10)} \\
\cmidrule(lr){2-3}\cmidrule(lr){4-5}\cmidrule(lr){6-7}\cmidrule(lr){8-9}\cmidrule(lr){10-11}\cmidrule(lr){12-13}
Rule & \# & tot. & \# & tot. & \# & tot. & \# & tot. & \# & tot. & \# & tot. \\
\midrule
Maj. Threshold & 10 & 163 & 10 & 235 & 10 & 291 & 10 & 322 & 10 & 321 & 10 & 304 \\
DVC & 10 & 269 & 10 & 307 & 10 & 349 & 10 & 377 & 10 & 362 & 10 & 364 \\
AWVC & 10 & 159 & 10 & 225 & 10 & 277 & 10 & 302 & 10 & 304 & 10 & 286 \\
DSS $\circ$ AWVC & 0 & 0 & 0 & 0 & 0 & 0 & 0 & 0 & 0 & 0 & 0 & 0 \\
DSS & 0 & 0 & 0 & 0 & 0 & 0 & 0 & 0 & 0 & 0 & 0 & 0 \\
ARU & 10 & 153 & 10 & 219 & 10 & 284 & 10 & 312 & 10 & 310 & 10 & 295 \\
MMOD & 0 & 0 & 0 & 0 & 0 & 0 & 0 & 0 & 0 & 0 & 0 & 0 \\
\bottomrule
\end{tabular}

%% file: arxiv.bbl
\begin{thebibliography}{44}
\providecommand{\natexlab}[1]{#1}
\providecommand{\url}[1]{\texttt{#1}}
\expandafter\ifx\csname urlstyle\endcsname\relax
  \providecommand{\doi}[1]{doi: #1}\else
  \providecommand{\doi}{doi: \begingroup \urlstyle{rm}\Url}\fi

\bibitem[Alamdari et~al.(2024)Alamdari, Ebadian, and Procaccia]{alamdari2024policyaggregation}
Parand~A. Alamdari, Soroush Ebadian, and Ariel~D. Procaccia.
\newblock Policy aggregation, 2024.
\newblock URL \url{https://arxiv.org/abs/2411.03651}.

\bibitem[Awad et~al.(2018)Awad, Dsouza, Kim, Schulz, Henrich, Shariff, Bonnefon, and Rahwan]{awad2018moral}
Edmond Awad, Sohan Dsouza, Richard Kim, Jonathan Schulz, Joseph Henrich, Azim Shariff, Jean-Fran{\c{c}}ois Bonnefon, and Iyad Rahwan.
\newblock The moral machine experiment.
\newblock \emph{Nature}, 563\penalty0 (7729):\penalty0 59--64, 2018.
\newblock \doi{10.1038/s41586-018-0637-6}.

\bibitem[Baujard and Igersheim(2007)]{VA2007}
Antoinette Baujard and Herrade Igersheim.
\newblock Voter autrement 2007 --- dataset of the in situ experiments, 2007.

\bibitem[Baujard et~al.(2014)Baujard, Igersheim, Lebon, Gavrel, and Laslier]{BAUJARD2014131}
Antoinette Baujard, Herrade Igersheim, Isabelle Lebon, Frédéric Gavrel, and Jean-François Laslier.
\newblock Who's favored by evaluative voting? an experiment conducted during the 2012 french presidential election.
\newblock \emph{Electoral Studies}, 34:\penalty0 131--145, 2014.
\newblock ISSN 0261-3794.
\newblock \doi{https://doi.org/10.1016/j.electstud.2013.11.003}.
\newblock URL \url{https://www.sciencedirect.com/science/article/pii/S0261379413001807}.

\bibitem[Baujard et~al.(2022)Baujard, Igersheim, Laslier, and Lebon]{VA2022}
Antoinette Baujard, Herrade Igersheim, Jean-Fran\c{c}ois Laslier, and Isabelle Lebon.
\newblock Voter autrement 2022 --- online experiment on the {French} presidential election, 2022.
\newblock URL \url{https://zenodo.org/records/10998451}.

\bibitem[Boehmer et~al.(2026)Boehmer, Kreisel, and Peters]{boehmer2026explanation}
Niclas Boehmer, Luca Kreisel, and Jannik Peters.
\newblock Explanation systems for approval-based multiwinner voting.
\newblock \emph{arXiv preprint arXiv:2604.24307}, 2026.

\bibitem[Brams and Sanver(2009)]{Brams2009}
Steven~J. Brams and M.~Remzi Sanver.
\newblock \emph{Voting Systems that Combine Approval and Preference}, pages 215--237.
\newblock Springer Berlin Heidelberg, Berlin, Heidelberg, 2009.

\bibitem[Brandl and Peters(2019)]{brandl2019axiomatic}
Florian Brandl and Dominik Peters.
\newblock An axiomatic characterization of the borda mean rule.
\newblock \emph{Social choice and welfare}, 52\penalty0 (4):\penalty0 685--707, 2019.

\bibitem[Brandt et~al.(2016)Brandt, Conitzer, Endriss, Lang, and Procaccia]{HandbookofCSC}
Felix Brandt, Vincent Conitzer, Ulle Endriss, J{\'e}r{\^o}me Lang, and Ariel~D Procaccia.
\newblock \emph{Handbook of computational social choice}.
\newblock Cambridge University Press, 2016.

\bibitem[Chaudhury et~al.(2024)Chaudhury, Murhekar, Yuan, Li, Mehta, and Procaccia]{VolumePVC}
Bhaskar~Ray Chaudhury, Aniket Murhekar, Zhuowen Yuan, Bo~Li, Ruta Mehta, and Ariel~D Procaccia.
\newblock Fair federated learning via the proportional veto core.
\newblock In \emph{Forty-first International Conference on Machine Learning}, 2024.

\bibitem[Chen et~al.(2022)Chen, Kyng, Liu, Peng, Gutenberg, and Sachdeva]{ChenKyngLiuPengProbstGutenbergSachdeva2022}
Li~Chen, Rasmus Kyng, Yang~P. Liu, Richard Peng, Maximilian~Probst Gutenberg, and Sushant Sachdeva.
\newblock Maximum flow and minimum-cost flow in almost-linear time.
\newblock In \emph{Proceedings of the 63rd IEEE Annual Symposium on Foundations of Computer Science (FOCS)}, pages 612--623, 2022.
\newblock \doi{10.1109/FOCS54457.2022.00064}.

\bibitem[Chooi et~al.(2026)Chooi, Gölz, Procaccia, Schiffer, and Zhang]{chooi2026findingcommongroundsea}
Jay Chooi, Paul Gölz, Ariel~D. Procaccia, Benjamin Schiffer, and Shirley Zhang.
\newblock Finding common ground in a sea of alternatives, 2026.
\newblock URL \url{https://arxiv.org/abs/2603.16751}.

\bibitem[Christiano et~al.(2017)Christiano, Leike, Brown, Martic, Legg, and Amodei]{christiano2017deep}
Paul~F Christiano, Jan Leike, Tom~B Brown, Miljan Martic, Shane Legg, and Dario Amodei.
\newblock Deep reinforcement learning from human preferences.
\newblock In \emph{Advances in Neural Information Processing Systems (NeurIPS)}, volume~30, 2017.

\bibitem[Condorcet(1785)]{CondorcetCycles}
{Marquis de} Condorcet.
\newblock \emph{Essai sur l'application de l'analyse {\`a} la probabilit{\'e} des d{\'e}cisions rendues {\`a} la pluralit{\'e} des voix}.
\newblock Imprimerie Royale, 1785.

\bibitem[Debreu(1954)]{debreu1954representation}
G\'erard Debreu.
\newblock Representation of a preference ordering by a numerical function.
\newblock \emph{Decision processes}, 3:\penalty0 159--165, 1954.

\bibitem[Dhillon and Mertens(1999)]{DhillonMertens1999}
Amrita Dhillon and Jean-François Mertens.
\newblock Relative utilitarianism.
\newblock \emph{Econometrica}, 67\penalty0 (3):\penalty0 471--498, May 1999.
\newblock \doi{10.1111/1468-0262.00033}.

\bibitem[Dong and Peters(2026)]{DoPe2026a}
Chris Dong and Jannik Peters.
\newblock An axiomatic analysis of proportionality notions in approval-based multiwinner voting, 2026.
\newblock URL \url{https://arxiv.org/abs/2605.04612}.

\bibitem[Dong et~al.(2021)Dong, Li, He, and Chen]{dong2021preference}
Yucheng Dong, Yao Li, Ying He, and Xia Chen.
\newblock Preference--approval structures in group decision making: Axiomatic distance and aggregation.
\newblock \emph{Decision Analysis}, 18\penalty0 (4):\penalty0 273--295, 2021.

\bibitem[Endriss(2017)]{10.5555/3180776}
Ulle Endriss.
\newblock \emph{Trends in Computational Social Choice}.
\newblock AI Access, 2017.
\newblock ISBN 1326912097.

\bibitem[Faliszewski et~al.(2017)Faliszewski, Skowron, Slinko, and Talmon]{faliszewski2017multiwinner}
Piotr Faliszewski, Piotr Skowron, Arkadii Slinko, and Nimrod Talmon.
\newblock Multiwinner voting: A new challenge for social choice theory.
\newblock \emph{Trends in computational social choice}, 74\penalty0 (2017):\penalty0 27--47, 2017.

\bibitem[Faliszewski et~al.(2020)Faliszewski, Slinko, and Talmon]{faliszewski2020multiwinner}
Piotr Faliszewski, Arkadii Slinko, and Nimrod Talmon.
\newblock Multiwinner rules with variable number of winners.
\newblock In \emph{Proceedings of the 24th European Conference on Artificial Intelligence (ECAI)}, pages 67--74, 2020.

\bibitem[Freeman et~al.(2020)Freeman, Kahng, and Pennock]{freeman2020proportionality}
Rupert Freeman, Anson Kahng, and David~M Pennock.
\newblock Proportionality in approval-based elections with a variable number of winners.
\newblock In \emph{Proceedings of the 29th International Conference on International Joint Conference on Artificial Intelligence (IJCAI)}, pages 132--138, 2020.

\bibitem[Halpern et~al.(2025)Halpern, Procaccia, and Suksompong]{halpern2025ApprovalVC}
Daniel Halpern, Ariel~D Procaccia, and Warut Suksompong.
\newblock The proportional veto principle for approval ballots.
\newblock \emph{arXiv preprint arXiv:2505.01395}, 2025.

\bibitem[Haret et~al.(2024)Haret, Klumper, Maly, and Sch{\"a}fer]{BudgetingGames}
Adrian Haret, Sophie Klumper, Jan Maly, and Guido Sch{\"a}fer.
\newblock Committees and equilibria: Multiwinner approval voting through the lens of budgeting games.
\newblock In \emph{Proceedings of the 25th ACM Conference on Economics and Computation}, pages 51--70, 2024.

\bibitem[Hitzig et~al.(2026)Hitzig, Gordon, Eloundou, Kalai, and Agarwal]{coval2026}
Zo\"{e} Hitzig, Mitchell Gordon, Tyna Eloundou, Adam Kalai, and Sandhini Agarwal.
\newblock Coval: Learning values-aware rubrics from the crowd.
\newblock OpenAI Alignment Research Blog, Jan 2026.
\newblock URL \url{https://alignment.openai.com/coval/}.

\bibitem[Ianovski and Kondratev(2023)]{ianovski2023computingproportionalvetocore}
Egor Ianovski and Aleksei~Y. Kondratev.
\newblock Computing the proportional veto core, 2023.
\newblock URL \url{https://arxiv.org/abs/2003.09153}.

\bibitem[Keswani et~al.(2026)Keswani, Cousins, Nguyen, Conitzer, Heidari, Borg, and Sinnott-Armstrong]{keswani2026moral}
Vijay Keswani, Cyrus Cousins, Breanna Nguyen, Vincent Conitzer, Hoda Heidari, Jana~Schaich Borg, and Walter Sinnott-Armstrong.
\newblock Moral change or noise? on problems of aligning {AI} with temporally unstable human feedback.
\newblock In \emph{Proceedings of the AAAI Conference on Artificial Intelligence (AAAI)}, 2026.

\bibitem[Kim et~al.(2018)Kim, Kleiman-Weiner, Abeliuk, Awad, Dsouza, Tenenbaum, and Rahwan]{kim2018computational}
Richard Kim, Max Kleiman-Weiner, Andr{\'e}s Abeliuk, Edmond Awad, Sohan Dsouza, Joshua~B. Tenenbaum, and Iyad Rahwan.
\newblock A computational model of commonsense moral decision making.
\newblock In \emph{Proceedings of the 2018 AAAI/ACM Conference on AI, Ethics, and Society}, pages 197--203, 2018.
\newblock \doi{10.1145/3278721.3278770}.

\bibitem[Kizilkaya and Kempe(2023)]{KizilkayaKempe2023}
Fatih~Erdem Kizilkaya and David Kempe.
\newblock Generalized veto core and a practical voting rule with optimal metric distortion.
\newblock In \emph{Proceedings of the 24th ACM Conference on Economics and Computation}, pages 913--936, 2023.

\bibitem[Kizilkaya and Kempe(2025)]{kizilkaya2025k}
Fatih~Erdem Kizilkaya and David Kempe.
\newblock $ k $-approval veto: A spectrum of voting rules balancing metric distortion and minority protection.
\newblock \emph{arXiv preprint arXiv:2507.17981}, 2025.

\bibitem[Kondratev and Ianovski(2024)]{kondratev2024veto}
Aleksei~Y Kondratev and Egor Ianovski.
\newblock Veto core consistent preference aggregation.
\newblock In \emph{Proceedings of the 23rd International Conference on Autonomous Agents and Multiagent Systems}, pages 1020--1028, 2024.

\bibitem[Kraiczy et~al.(2025)Kraiczy, Papasotiropoulos, Skowron, et~al.]{kraiczy2025proportionality}
Sonja Kraiczy, Georgios Papasotiropoulos, Piotr Skowron, et~al.
\newblock Proportionality in thumbs up and down voting.
\newblock \emph{arXiv preprint arXiv:2503.01985}, 2025.

\bibitem[Kruger and Sanver(2021)]{KrugerArrovian}
Justin Kruger and M.~Remzi Sanver.
\newblock An arrovian impossibility in combining ranking and evaluation.
\newblock \emph{Soc. Choice Welf.}, 57\penalty0 (3):\penalty0 535--555, 2021.

\bibitem[Kyi et~al.(2026)Kyi, G{\"o}lz, Berjon, and Biega]{Kyi2026}
Lin Kyi, Paul G{\"o}lz, Robin Berjon, and Asia~J. Biega.
\newblock From clicks to consensus: Collective consent assemblies for data governance.
\newblock In \emph{Proceedings of the 2026 {CHI} Conference on Human Factors in Computing Systems, {CHI} 2026, Barcelona, Spain, April 13--17, 2026}, pages 348:1--348:17, 2026.
\newblock \doi{10.1145/3772318.3790690}.

\bibitem[Lackner and Maly(2025)]{lackner2025approval}
Martin Lackner and Jan Maly.
\newblock Approval-based shortlisting.
\newblock \emph{Social Choice and Welfare}, 64\penalty0 (1):\penalty0 97--142, 2025.

\bibitem[Lackner and Skowron(2023)]{lackner2023multi}
Martin Lackner and Piotr Skowron.
\newblock \emph{Multi-winner voting with approval preferences}.
\newblock Springer, 2023.

\bibitem[Laslier and der Straeten(2004)]{LaslierVanDerStraeten2004}
Jean-François Laslier and Karine~Van der Straeten.
\newblock Une expérience de vote par assentiment lors de l'élection présidentielle française de 2002.
\newblock \emph{Revue française de science politique}, 54\penalty0 (1):\penalty0 99--130, 2004.
\newblock \doi{10.3917/rfsp.541.0099}.
\newblock URL \url{https://doi.org/10.3917/rfsp.541.0099}.

\bibitem[Mattei and Walsh(2013)]{MaWa13a}
Nicholas Mattei and Toby Walsh.
\newblock Preflib: A library of preference data \textsc{http://preflib.org}.
\newblock In \emph{Proceedings of the 3rd International Conference on Algorithmic Decision Theory (ADT 2013)}, Lecture Notes in Artificial Intelligence. Springer, 2013.

\bibitem[Moulin(1981)]{Moulin1981}
Herv{\'e} Moulin.
\newblock The proportional veto principle.
\newblock \emph{The Review of Economic Studies}, 48\penalty0 (3):\penalty0 407--416, 1981.

\bibitem[Moulin(1982)]{Moulin1982}
Herv{\'e} Moulin.
\newblock Voting with proportional veto power.
\newblock \emph{Econometrica: Journal of the Econometric Society}, pages 145--162, 1982.

\bibitem[Ouyang et~al.(2022)Ouyang, Wu, Jiang, Almeida, Wainwright, Mishkin, Zhang, Agarwal, Slama, Ray, et~al.]{ouyang2022training}
Long Ouyang, Jeffrey Wu, Xu~Jiang, Diogo Almeida, Carroll Wainwright, Pamela Mishkin, Chong Zhang, Sandhini Agarwal, Katarina Slama, Alex Ray, et~al.
\newblock Training language models to follow instructions with human feedback.
\newblock In \emph{Advances in Neural Information Processing Systems (NeurIPS)}, volume~35, pages 27730--27744, 2022.

\bibitem[Peter(2023)]{peter2023political}
Fabienne Peter.
\newblock Political legitimacy.
\newblock In Edward~N. Zalta and Uri Nodelman, editors, \emph{The {Stanford} Encyclopedia of Philosophy}. Metaphysics Research Lab, Stanford University, winter 2023 edition, 2023.
\newblock URL \url{https://plato.stanford.edu/archives/win2023/entries/legitimacy/}.

\bibitem[Tarjan(1972)]{Tarjan1972DFS}
Robert~E. Tarjan.
\newblock Depth-first search and linear graph algorithms.
\newblock \emph{SIAM Journal on Computing}, 1\penalty0 (2):\penalty0 146--160, 1972.
\newblock \doi{10.1137/0201010}.

\bibitem[Weijer and Emanuel(2000)]{weijer2000protecting}
Charles Weijer and Ezekiel~J. Emanuel.
\newblock Protecting communities in biomedical research.
\newblock \emph{Science}, 289\penalty0 (5482):\penalty0 1142--1144, 2000.
\newblock \doi{10.1126/science.289.5482.1142}.

\end{thebibliography}
